\newif\iflong\longfalse

\documentclass[a4paper,cleveref, autoref, thm-restate]{lipics-v2021}
\pdfoutput=1
\hideLIPIcs
\nolinenumbers

\usepackage{amsthm}
\usepackage{newtxmath}
\usepackage{stmaryrd}
\usepackage{mathtools} %
\usepackage{mathpartir}
\usepackage{xspace}

\usepackage{tikz,tikz-cd}

\usepackage[]{commonmacros}

\title{Wiring the \texorpdfstring{\( \pi \)}{pi}-calculus to Denotational Semantics} %

\author{Ken Sakayori}{The University of Tokyo, Japan }{
sakayori@is.s.u-tokyo.ac.jp
}{https://orcid.org/0000-0003-3238-9279}{JSPS KAKENHI Grant Number 24K20731}

\author{Davide Sangiorgi}{University of Bologna, Italy \and Centre Inria d'Université Côte d'Azur, France}{davide.sangiorgi@unibo.it}{0000-0001-5823-3235}{MIUR11PRIN project ‘Resource Awareness in Programming: Algebra, Rewriting, and
Analysis’ (RAP, ID P2022HXNSC)}
\author{Simon Castellan}{
Inria, Univ Rennes, CNRS, IRISA, France
}{simon.castellan@inria.fr}{https://orcid.org/0000-0001-5886-5793}{}
\author{Pierre Clairambault}{
  CNRS, Aix Marseille Univ, LIS, France
}{pierre.clairambault@cnrs.fr}{0000-0002-3285-6028}{ANR project DyVerSe (ANR-19-CE48-0010-01)}

\authorrunning{K.\ Sakayori, D.\ Sangiorgi, S.\ Castellan and P.\ Clairambault} %

\Copyright{Ken Sakayori, Davide Sangiorgi, Simon Castellan and Pierre Clairambault} %



\category{} %

\acknowledgements{We are grateful to the anonymous reviewers for their feedback.}%

\EventEditors{Claudia Faggian and Joost-Pieter Katoen}
\EventNoEds{2}
\EventLongTitle{41st Annual Symposium on Logic in Computer Science (LICS 2026)}
\EventShortTitle{LICS 2026}
\EventAcronym{LICS}
\EventYear{2026}
\EventDate{July 20--23, 2026}
\EventLocation{Lisbon, Portugal}
\EventLogo{}
\SeriesVolume{380}
\ArticleNo{17}

\begin{document}

\maketitle

\begin{abstract}
We introduce a dialect of the Asynchronous \( \pi \)-calculus, called \wapi{}, in which
(1) an input name may be owned, at any time,  by at most one process;
(2) each name has either only the input or only the output capability.
As a result, special processes called wires (aka forwarders, that is, processes that receive values at one name and re-transmit) behave as substitutions when composed with any \wapi{} process.
Thus \api{} naturally yields a category,  whose morphisms are \wapi{} processes (modulo the
reference behavioural equivalence,  barbed congruence) and whose objects are types; and
where wires act as identity morphisms.
We show that the category of processes can be further organised into (sub)categories with the structures needed for the interpretation of
common higher-order language features in the literature by drawing on insights from game semantics; notably, we construct
a relative Seely category, the categorical structure that concurrent
game semantics has.
At the same time,  \wapi{} follows the tradition of ordinary \( \pi \)-calculi in that
expressiveness is preserved and the operational and
algebraic theory are developed in a similar manner, notwithstanding substantial technical differences in their development and proofs.
In short, the goal of \api{} is to remain faithful to the operational and algebraic tradition of the \( \pi \)-calculi while connecting to  the tradition of denotational models for programming languages.
\end{abstract}

\section{Introduction}
\label{sec:intro}

The \( \pi \)-calculus has been widely used as a basic tool for representing message-passing computation due to its expressiveness and simplicity.
It has been advocated as a \emph{metalanguage} for describing, and reasoning about, programming languages (or core fragments of them), even sequential ones, such as  \( \lambda \)-calculi~\cite{Milner90,SangiorgiWalker01} and object-oriented languages~\cite{Walker91,Jones93}.
Indeed, its message-passing abilities can be leveraged to represent the interactions of open programs with their contexts, allowing reasoning on open programs.
These encodings and translations not only reveal the interactive nature of the source
 objects, but also give access to powerful algebraic and %
operational proof techniques, which the \( \pi \)-calculus offers.

Another methodology to reason about the interactive behaviour of programs
is given by denotational semantics. Following denotational semantics,
types and open programs are interpreted into an appropriate
mathematical universe, reflecting the relevant aspects of their
execution. Precise settings vary -- in particular, \emph{game
semantics} represents a program as a strategy, a process exchanging
messages with an arbitrary execution environment, with a clear kinship
with process calculi (in fact, the $\pi$-calculus has been naturally
used to provide \emph{syntactic representations of denotational
objects}, notably  strategies of game
semantics~\cite{HylandOng95,FioreHonda98,YoshidaCS20}). But the
distinctive feature of denotational semantics is that models are
organised as structured categories (typically cartesian closed
categories), and the bulk of the interpretation of expressive
programming languages is streamlined following this categorical machinery.

Both of these established methodologies rest on a whole literature with
powerful developments and reasoning techniques; and in principle, one
could have the best of both worlds: translations to process calculi
could be instances of denotational semantics if processes were to be
organised as denotational models; or from the other side, denotational
interpretations could be given a syntactic description via a translation
into a process calculus.

However, in trying to accommodate these two strands of research, there is a difficulty: \emph{processes cannot be easily organised into categories}.
To align with denotational semantics, we must form a category where morphisms are processes and objects are types; for the composition of morphisms, it appears natural to adopt `parallel composition plus hiding' in an interactive semantics, as advocated, for example, by Abramsky~\cite{AbramskyGN96}.
The equality of the category should match the operational equality on the processes,
usually  contextually-defined (e.g., \emph{barbed congruence}).
\newcommand*{\idProc}[2]{\mathrm{Id}(#1, #2)}
If one attempts to form a category along these lines, an immediate difficulty arises:
finding a process \( \idProc a b \) that satisfies the following identity law
\begin{equation}
\label{eq:Asynchronous-law}
P \sub a b =   \res b ( \idProc a b | P )\,,
\end{equation}
in other words, one needs  a process that works as a substitution.
In the literature,  instances of the
 law may be found taking
 for \( \idProc a b \)  special processes called \emph{wires}~\cite{SangiorgiWalker01} or
\emph{forwarders}~\cite{HondaYoshida95}  (or their bidirectional version
called equators~\cite{HondaYoshida95,Merro99}).
 A wire    \( \links b a  \defi !\inp b x .\out a x \) receives values at
$b$ and re-emits them at $a$.
For example,
in \iflong the 
\fi
\emph{Asynchronous Localised $\pi$-calculus}~\cite{MerroSangiorgi04},
\iflong
(\ALpi{})
\fi
the
law  holds for \( \idProc b a = \links b a \) 
\iflong
provided
\else
when
\fi
 $b$ is only used in output in $P$;
\iflong
even so, 
 \ALpi{} 
\else
still, the calculus 
\fi
cannot be seen as a category as the law does not hold in general.
\iflong for all processes.\fi

Certain calculi achieve this law, but at the cost of departing from the standard operational theory of the \( \pi \)-calculus.
For example, (1) by significantly modifying the reduction semantics of
the \( \pi \)-calculus (i.e.,  allowing certain \iflong conventional \fi multi-step
reductions to happen in one step)~\cite{SakayoriTsukada21}, (2) by
considering processes up to permutation of certain prefixes (i.e., 
enlarging structural congruence)~\cite{YoshidaCS20} or (3) by using
a coarse process equivalence such as may-testing equivalence coupled with
asynchrony and  i/o types~\cite{SakayoriTsukada17}.

In this paper, we study a dialect  of   Asynchronous \( \pi \)-calculus, called \api{}~--- where A and W emphasise the role of asynchrony and wires.
It satisfies (the analogue of)
law \eqref{eq:Asynchronous-law}, \emph{while maintaining the standard operational machinery of
$\pi$-calculi}:   reduction semantics, a branching-time 
contextually-defined
reference behavioural
equivalence such as barbed congruence,   proof techniques based on \emph{labelled bisimilarity}
and \emph{up-to enhancements}  for it (adapted to the new calculus).
The modifications do not affect
  the expressive power of the calculus~--- demonstrated
by an encoding of \ALpi{}.
In \api, processes (and types) naturally form a category which can be equipped with standard structures in denotational semantics, in particular that of a cartesian-closed category (thus a model of call-by-name $\lambda$-calculus).
We show that the interpretation of call-by-name $\lambda$-calculus by the categorical
semantics yields a process encoding of \( \lambda \)-calculus along the lines of 
 that obtained via game semantics~\cite{HylandOng95}.

The calculus \api is derived  from   Asynchronous $\pi$-calculus by imposing
 two additional features, which
we shall call \emph{1-input property} and \emph{I/O separation}.
The 1-input property is the requirement  that an \emph{input name may be owned, at any
  given time,  by at most one process}.
For example, a process \( \inp a x  . \out c x \mid \inp b x . \out c x \) follows the 1-input property, whereas \( \inp a x . \out c x \mid \inp a x . \out c x \) does not.
Sequential use of the same name in input is allowed, and the input capability of channels may be exported.
For instance, \( \inp a x . \inp a y . P \) and \( \out a c \mid Q \), where \( c \) has the input capability, are valid processes if \( c \) is not used as an input in \( Q \).
Roughly, the 1-input condition excludes races caused by multiple receivers.
I/O separation, in contrast,  requires any name to have either only the input or only the output capability.
This differs from the conventional notion of a name in $\pi$-calculi,
where a name is introduced, in a restriction, with
 both the input and the output capability.
In \api{},  name restriction takes a form \( \resB x y P \), and  binds an input name \( x
\) and output name \( y \),
so that communications may happen along these now-connected names.

Neither feature, per se, is new. For instance,  the 1-input feature may be found in
actor-based languages such as Erlang (where communication occurs via \emph{process
  names}), and  in 
certain dialects of the $\pi$-calculus such as Join~\cite{FournetGonthier96} and
$\pi_1$ \cite{Amadio97} (however, usually in these languages,
an input capability may not be exported);
and I/O separation may be found in denotational studies of Asynchronous \( \pi \)-calculus~\cite{Laird05,SakayoriTsukada19} or in session-typed languages~\cite{GayVasconcelos25}, where its use is now standard.  Both features may even be found in the asynchronous channels of Rust~\cite{RustMpsc} and
in certain
 constrained forms of concurrency such as those of
 linear-logic inspired session-typed calculi (e.g.~\cite{CairesPT16,Wadler14,KokkeMP19}).
In this respect, the novelty of \api is  to  study the combination of the
 two features within a basic calculus  of mobile processes    on an equal footing with
traditional  $\pi$-calculi (session types, or more generally behavioural types and 
 other advanced type systems, could be added onto \api in the same way as they
\iflong
 have
 been
\else
are
\fi
 added to $\pi$-calculi).

In \api, (the analogue of)
law \eqref{eq:Asynchronous-law} holds, even for demanding branching-time behavioural
equivalences such as  barbed congruence,  taking \( \idProc a b \)  to be  a wire process.
Precisely,  \( \idProc a b \)  can be  \( \links a b \) or its dual  \( \links b a \),  depending on its polarity.
Technically, for a wire to behave as the identity in 
\eqref{eq:Asynchronous-law}, the calculus must be asynchronous,  and 
 all messages emitted towards the
input name of the wire must  indeed flow through the wire. The 1-input feature guarantees
the latter property.

Proving  \eqref{eq:Asynchronous-law} in \api turns out to be
challenging. A  key ingredient, and an instance of the law itself, 
 is showing that a `free output'~---  the output of a global name~--- 
 can be transformed into a  `bound  output'~---  the output of a fresh name~---  provided
 that the fresh name is `wired' with the free name; e.g., when $P = \out ca $
 in   \eqref{eq:Asynchronous-law}.
The transformation of free outputs into bound outputs
 has appeared in the study of localised
$\pi$-calculi;  and similarly for the instance of 
\eqref{eq:Asynchronous-law} in which the substituting name $a$ is an output name.
 The  novelty and the challenge  here was to accommodate  
the case when the  emitted name is an input name and,
in \eqref{eq:Asynchronous-law},  the case in which  $a$ is an input name.
Indeed, the proofs for the two cases of  \eqref{eq:Asynchronous-law}  are quite different:
the output case is a simple structural induction, whereas 
 the input case requires a combination of  induction and up-to
techniques. 
Another benefit of the   transformation of free outputs into  bound outputs is that   \api processes may be 
 turned to equivalent \emph{internal} processes \cite{Sangiorgi96a}, i.e.,  processes in
 which all outputs are bound.  
In $\pi$-calculus,  internal processes have a simpler theory, and the same holds for \api. 
The proof techniques  for the
contextually-defined barbed congruence of \api, such as  labelled bisimilarity and their up-to
enhancements, are  precisely derived 
by  exploiting the transformation into internal processes.

The law~\eqref{eq:Asynchronous-law} only ensures that \api{} processes become a plain category; however, we can go further.
We will use game semantics as our guiding principle to carve out certain subcategories of \api{} processes equipped with richer categorical structures that allow us to interpret various programming language features.
This supports the expressiveness of  \api{} from a denotational perspective and also suggests that \api{} could be a suitable language to represent strategies of (concurrent) game semantics.
This connection with game semantics rests on the I/O separation feature, which makes the notion of duality explicit.

While motivated by categorical matters,
I/O separation has also a significant impact on the operational theory.
For instance, 
in the definition of contextual 
  behavioural equivalences, free  input and output names  \iflong of processes \fi may be connected,
by `closing' contexts, in an  arbitrary manner; moreover substitutions act separately on such
input and output ends. These aspects bring up   new  technicalities concerning substitutions and 
name synchronisation; there are also 
 similarities with the treatment of distinctions among names in
\emph{open bisimilarity} for $\pi$-calculi \cite{Sangiori93}.  
Another example of the impact of I/O separation concerns the 
proof of  the invariance of bisimilarity
under the  restriction operator: 
in contrast with $\pi$-calculi,  in
\api restriction connects names, hence     synchronisations between names
may appear  that  did not exist before.

\subparagraph*{Plan of the paper.}
The contributions and the structure  of this paper can be summarised as follows:
\begin{enumerate}
  \item  In Sections~\ref{sec:syntax} and~\ref{s:rs}
we define \api{}:
\iflong
, a \( \pi \)-calculus with the asynchrony, 1-input,  and I/O separation
features, and with  a  grammar similar to that of asynchronous $\pi$-calculi.
We define a reduction semantics for it, and then barbed congruence on top of this.
\else
syntax,  reduction semantics, barbed congruence.
\fi

  \item
    \label{it:intro-laws}
   In Section~\ref{s:alaws} we present  some laws that are specific to
 \api{}
 (i.e., they do not hold for barbed congruence in Asynchronous $\pi$-calculus),
including the laws  for wires discussed in this introduction.
\iflong
,  and a law that allows  commutativity of inputs.
\fi
  \item In Section~\ref{s:ls}, we 
 define an LTS semantics for \api, which  we use to transfer a wealth of prior results from
\( \pi \)-calculi to \api. 
On top of the LTS, we then 
devise bisimulation-based proof techniques, specific to \api{},
resting on a transformation of \api processes into  internal processes.
A labelled bisimilarity  is proved to be  sound and complete with respect to barbed congruence (as
usual in $\pi$-calculi,  completeness holds on image-finite processes).
  \item In Section~\ref{s:expr} we discuss  the encoding of Asynchronous Localised \( \pi
    \)-calculus
    into \api{}.
\item In Section~\ref{sec:esyntax} we discuss how \api{} can be extended with linear types and standard data types.
  \item
In Section~\ref{sec:category},
inspired by game semantics, and using the algebraic laws of \api, 
we show that a relative Seely category~\cite{Clairambault24} and sequoidal category~\cite{Laird02} can be constructed out of processes.
    Then, in Section~\ref{sec:encoding-lambda}, we show that interpreting the simply-typed \( \lambda \)-calculus allows us to derive the translation from \( \lambda \) to \( \pi \) obtained via game semantics~\cite{HylandOng95}.
\end{enumerate}
In short, the goal of \api is to be faithful to the process calculus tradition while
connecting  to the literature on  denotational models.

\section{Syntax } 
\label{sec:syntax}
This section introduces the monadic \api{}, i.e.~the simplest but representative version, to highlight its key features.
Later, in Section~\ref{sec:esyntax}, we consider standard extensions of
the grammar of types and, correspondingly,  of values
(e.g., polyadicity, sums, linearity).

\subsubsection*{Grammar}
\label{ss:gr}

Let ${\mathcal N}$ be a countable infinite set of  \emph{names},
 ranged over by small letters
($a,b,c, \ldots,  x,y,z$).   These include 
channel names (names used for communication); a channel name 
may be an \emph{input name} (I-name), 
which may only be used in input, or
an
\emph{output name} (O-name), 
which may only be used in output.
The grammar of (monadic) \api processes %
 is
\[ \begin{array}{rcl}
P & \Coloneqq &  \nil \midd  P | Q\midd   \inp a x .{P}
\midd \out bv \midd %
\resB* a b  \ty P
\midd ! \inp a x .P \\
 v & \Coloneqq & a \midd \unitvalue
   \end{array}
  \]

The operators are those of ordinary Asynchronous $\pi$-calculus,  except for  restriction,
that  follows certain
presentations of session-typed calculi~\cite{GayVasconcelos25}.
A restriction $\resB* a b \ty P$  connects, and therefore makes
interaction possible between, 
 an input name $a$ whose type is $T$, and an output name $b$ whose type is the dual of
 $T$. 
We say that $a$ and $b$ are
\emph{connected}, and that $a$ is the \emph{companion} of $b$ (and conversely).
The type  subscript \( \ty \) will often be omitted in the sequel.
The core calculus is monadic and $ \unitvalue$ is the only `non-channel' value. 
 In examples and
applications (e.g., encodings) we sometimes  use polyadicity.
As usual, input and restriction are binders.
The set of names, free and bound names ($\nsymb$, $\fnsymb$ and $\bnsymb$, respectively)
of a process or other entities such as actions,  \iflong contexts, 
\fi
 or typings,  are defined as expected.
  In statements, 
we suppose that  processes and actions  
abide by a Barendregt convention:
bound names are pairwise distinct and distinct from the free names.

Following the conventions  for testing equivalences, 
we assume that there are special
 names, the \emph{success names},  that are used by testing contexts
 to report the success of an experiment. 
Success names may only be used in output and may not be restricted. 
\iflong We write $\successSET$ for the set of such names. 
\fi
The value that success names carry is irrelevant and we may thus assume it to
be the unit value. 
 A \emph{closed} term is a term whose free names may only
  be success names.  
We may also assume that processes have no free name of unit type; i.e., any free name is
either an I-name or an O-name.

\subparagraph*{Notations}
We  abbreviate an output $\out  b \unitvalue $ as $\outC b$.
We write $\res \tilc P $ for a sequence of restrictions, and then 
abbreviate
  $\resB a b  \res{\tilc} P $ as
 $(\resBtil a b , \tilc) P $.
 In a statement, 
if $a$ is a bound name, then  $\comp a $  is the companion of $a$
(that is, if $\resB a b $ or $\resB b a$ is the restriction that binds $a$, then $\comp a
= b$).
\iflong
By convention, $ \comp{(\comp a)} $ is the same as $a$.  
\fi
We write $p\usepRES p'$  to mean 
$p\sepRES p'$ or $p'\sepRES p$, depending on the polarity of the two names (i.e., whether
$p$ is the I-name and $p'$  the O-name, or conversely).
Then   $\resBP a $  stands for  $\res{a \usepRES  {\comp a}}$. Finally, the bound-output
notation $\bout a c  $ indicates the  output $\out a c$ preceded by the restriction
$\resBP c $; accordingly the bound-output prefix
  $\bout a c: P  $ stands for  $\resBP c ( \out a c | P )$.
\iflong
 (sometimes called  a \emph{permeable} or \emph{asynchronous}  bound output in the literature).
\fi

\begin{remark} 
In Section~\ref{s:ls}, we show that \api can be compiled into a subset of processes with only bound outputs.
We prefer, however, to maintain the free output construct in \api
because it is convenient in practice, often allowing more direct and efficient modelling. Milner's
original encoding of the $\lambda$-calculus \cite{Milner90} for instance, I/O separation aside, is
an encoding into \api, and  uses free outputs (including outputs of 
input capabilities, while respecting the 1-input property).
\end{remark}

\subsubsection*{Types}
We begin with a core syntax for types, postponing standard extensions
(e.g., product types, sum types, and linear types)
 to Section~\ref{sec:esyntax}: 
\begin{align*}
\ty \Coloneqq %
\och \ty \mid \ich \ty \mid \unittype 
\end{align*}

The type \( \och \ty \) is the type of names that may only be used 
in
output,  to send values  of type \( \ty \); whereas \( \ich
\ty \) is the type of names that may only be used in input to
 receive values of type \( \ty \); and $\unittype $ is the
unit type. 
The \emph{dual} of \( \och { \ty} \), denoted  \( \dual {\och { \ty}}\) is 
\iflong
defined as
\fi
 \( \ich { \ty}\) and conversely.
Typing judgments are of the form \( \tyenv \vdash P \) where \( \tyenv \) is a \emph{type environment}, a finite set of bindings of the form \( a : \ty \).
The typing rules are 
\iflong
given 
\fi
in Figure~\ref{fig:typing}, where  
$ \iname \tyenv $ is the set of I-names appearing  in $\tyenv$ (i.e., names $a$ such that $a:  \ich { \ty} \in \tyenv$), and $ \tyenv_1 , \tyenv_2$ is the union of $
\tyenv_1$ and $ \tyenv_2$, and is only defined if they agree on the type assignments for
shared names.   Moreover, 
 \( \och* \tyenv \)
indicates a type environment in which each name is an O-name
(i.e., 
 each binding %
is of the form \( a : \och {\seqq  \ty} \)).

\begin{figure}[t]
  \begin{mathpar}
  \textsc{Nil}\inferrule{ }{\tyenv \vdash \nil}\qquad
  \inferrule[Par]{\tyenv_1 \vdash P_1 \\ \tyenv_2 \vdash P_2 \\ \iname{\tyenv_1} \cap \iname{\tyenv_2} = \emptyset}{\tyenv_1, \tyenv_2 \vdash P_1 \mid P_2}
  \qquad
  \inferrule[Res]{\tyenv, a : \ich \ty, b : \och \ty \vdash P}{\tyenv \vdash \resB* a b {\ich \ty} P}\\
  \textsc{Out}\inferrule{ }{\tyenv, a : \och {\seqq \ty}, \seqq b : \seqq \ty \vdash \out a {\seqq b}}
  \qquad
  \textsc{In}\inferrule{\tyenv, a : \ich {\seqq \ty}, \seqq b : \seqq \ty \vdash P}{\tyenv, a : \ich {\seqq \ty} \vdash \inp a {\seqq b}. P} \qquad
  \textsc{RIn}\inferrule{\och* \tyenv, \seqq b : \seqq \ty \vdash P}{\och* \tyenv, a : \ich {\seqq \ty}, \tyenv' \vdash  !\inp a {\seqq b}. P}
\end{mathpar}
\caption{Typing rules}
\label{fig:typing}
\end{figure}

The type system enforces the 1-input property. 
This comes up in rule \trans{Par}, where the I-names are split between the two components
of the composition, and in the rule \trans{RIn} for replication  
 \( !\inp a b .P \), which ensures that 
no I-name may appear in the body $P$ of the replication (not even the head name $a$), 
possibly with the exception of
the formal parameter $b$. (The rule \trans{RIn} is thus analogous to the promotion rule of linear logic.)
Rule \trans{In} shows that a process that owns an I-name may perform nested inputs at
that name.
The typing system is \emph{affine}, i.e.~the processes need not use all
the names declared in the type environment, including the declared
I-names.

\begin{example}
  \label{ex:server}
In this example (for readability, using integers and binary communications) a server
accepts, at $a$, requests consisting of an integer and a return name; it answers with the
successor integer and then delegates, using $b$, to a different permanent (replicated)  process the task of answering
further requests:
\[ 
\resBP b (  \inp a {n,q}. (\out q {n+1} |  \out b a  ) | \inp {\comp b } x . ! \inp  x
{m,r}. P )   
\]
\end{example} 
\begin{example}
  \label{ex:internal-choice}
Because of the  1-input property,  in \api race conditions, and hence non-determinism, 
may only  be created  on access to \emph{inputs}.
In asynchronous $\pi$-calculi without the 1-input constraint, 
an internal choice operator \( P \oplus Q \), which may non-deterministically evolve to
either $P$ or $Q$, is usually encoded by a race condition on \emph{outputs}
(e.g., \( \resBP a  (\outC {\comp a}  |  a . P \mid a. Q) \), for \( \breve a \)
fresh).
This may be mimicked in \api, via a race condition on inputs:
 \( \resBP a \resBP b \resBP c (\inp a x . \outC x \mid \out {\comp a} {\comp b} \! \mid \out
 {\comp a} {\comp c} \! \mid \! b.P \! \mid \!c . Q ) \), where  $\breve a,\breve b, \breve c $
 are fresh.
\end{example}

\section{Reduction Semantics}
\label{s:rs}
Reduction between %
processes, 
written $P \longrightarrow P' $, is defined %
as expected. We just recall that, with the separation between I-names and O-names, the
basic reduction rules must be justified by a restriction connecting the I-name and
O-name involved \cite{Laird05,Vasconcelos12}, e.g.:
\[
\resB a b (\inp a x . P | \out b v | Q ) \longrightarrow   
\resB a b (  P \sub v x  | Q )    
\] 

See Appendix~\ref{a:b} for the details,
including the definition of the auxiliary relation
of structural congruence ($\equiv$), and the Subject Reduction property relating reductions
and types.  
The relation $ \Longrightarrow $ is the reflexive and transitive closure of $ \longrightarrow
$.  
As usual, 
a \emph{barb} $P \dwaB \success$ , where $b$ is a success name,  
holds if 
$P$ has an unguarded occurrence of an output at
$b$
(i.e., $P \equiv 
\outC b |
P') $,  for some %
$P'$
(we recall that a success name may not be restricted);
then 
$P \DwaB \success \defi P \Longrightarrow \dwaB \success$.
We only report the (standard) definition of 
barbed bisimilarity (types  play no role in its definition,  and are omitted).

\iflong
\begin{definition}[barbed bisimulation]
A symmetric relation $\R$ on (closed?)  processes
 is a  {\it barbed bisimulation} if 
whenever $P \RR Q$ the following holds:
 
\begin{enumerate}

\item If $P \dwaB \success $ then $Q  \DwaB \success$.

\item If $P \longrightarrow  P'$ then $Q \Longrightarrow  Q'$ and
  $P'\RR Q'$.

\end{enumerate}
We write  $\wbb$  for  the largest barbed bisimulation.
\end{definition}
We recall that a process is \emph{closed} it all its free names are in
$\Omega$ (i.e, success names). 

\begin{itemize}

\item $ \Delta, \Delta'$ is an \emph{extension} of $\Delta$ if $\Delta'$ only contains
  O-names.   
A context $\qct$ is a  \emph{$\Gamma/\Delta$ context} if $\ok \qct \Gamma $ holds,
assuming that the hole is typable under an extension of $ \Delta $.  
\item
A testing context $\qct$ is \emph{closing for $\Delta$} if  
$\qct$ is a $\Gamma/\Delta$ context, 
for some $\Gamma$ and
 $\Gamma$ only has success names. 
\end{itemize}

\begin{definition}[$ \Delta $-barbed congruence]
\label{d:}
If 
 $\ok \Delta {P,Q}$, then 
we say that \emph{$P$ and $Q$
 are $\Delta$-barbed congruent}, 
 written
$ P \wbcTT \Delta{} Q$, if 
 for any context $\qct$ that is closing for $\Delta$, 
  it holds that
$\ct P \wbb \ct Q$. 
\end{definition} 

\else
\begin{definition}[Barbed bisimilarity]
\emph{Barbed bisimilarity} is the largest symmetric relation $\wbb$  on closed  processes
such that, whenever $P \wbb Q$:
 \begin{enumerate}
\item If $P \dwaB \success $ then $Q  \DwaB \success$;
\item if $P \longrightarrow  P'$ then $Q \Longrightarrow  Q'$ and
  $P'\wbb Q'$.
\end{enumerate}
\end{definition}
\fi

We may now define $\Delta$-barbed congruence.
 $ \Delta, \Delta'$ is an \emph{extension} of $\Delta$ if $\Delta'$ only contains
  O-names.   
A context $\qct$ is a  \emph{$\Gamma/\Delta$ context} if $\ok \tyenv \qct$ holds,
assuming that the hole is typable under an extension of $ \Delta $.  
A testing context $\qct$ is \emph{closing for $\Delta$} if  
$\qct$ is a $\Gamma/\Delta$ context, 
for some $\Gamma$ and
 $\Gamma$ only has success names.

\begin{definition}[$\Delta$-barbed congruence]
\label{d:dbc}
Assume $\ok \Delta {P,Q}$. 
Then \emph{$P$ and $Q$
 are $\Delta$-barbed congruent}, 
 written
$ P \wbcTT \Delta{} Q$, if 
 for any context $\qct$ that is closing for $\Delta$, 
  it holds that
$\ct P \wbb \ct Q$.
\end{definition} 

Intuitively
 $\wbcTT \Delta{} $
relates processes that are typable at $\Delta$, and that  cannot be distinguished by contexts that 
cannot use the input names in $ \Delta $. 
This means that the context can decide to connect the free I-names and
O-names (with dual types) in $P$ and
$Q$ in an arbitrary manner. 
Sometimes it is useful to fix certain connections between names, and  allow only
contexts that respect such connections. We therefore introduce a more
general notion of barbed congruence, parametrised on a  \emph{connection set} $ \delta $,  
i.e., a 
  finite set of connections between names: 
\(
  \delta \Coloneqq a\sepRES b, \delta  \midd \emptyset
\).

A connection set $\delta$ will be used on processes whose typing includes the names in $\delta$.
A  context $\qct$, closing for $\Delta$,  \emph{respects the connection set $\delta$}
 if $\delta \subseteq \delta_\qct$, where
 $\delta_\qct$ is the set of restrictions in the context $\qct$ whose scope embraces the
hole of the context.

\begin{definition}
\label{d:deltabc}
Assume $ \n \delta  \subseteq \n \Delta $ 
and  $\ok \Delta {P,Q}$. 
Then
  \emph{$P$ and $Q$
 are $\Delta$-barbed congruent at $\delta$}, 
 written
$ P \wbcTT \Delta \delta  Q$, if 
 for any context $\qct$ that is closing for $\Delta$ and that respects
$\delta$,  
  it holds that
$\ct P \wbb \ct Q$. 
 \end{definition}

\section{Algebraic Laws}
\label{s:alaws}
In this section,
we report  some laws that are valid in \api, and that bring up the
peculiarities of the calculus, notably its 1-input property. 
The most important law is that in Theorem~\ref{t:wire-subst-law}. 
Most of the laws are proved with the technique based on labelled bisimilarity that will be  introduced in
Section~\ref{s:ls}.

A \emph{wire},  between an  I-name $a$ and  an O-name  $b$  is 
a process $ !  \inp a x . \out bx $
(once linear types are introduced in Section~\ref{sec:esyntax}, the replication will be dropped when
$a$ and $b$ are linear). 
We abbreviate such a wire process 
as
$ \links a b$.
Moreover 
  $\ulinks  a b $ stands for either
$\links  a b $ or 
$\links  ba $, depending  on the polarity of the two names (which one is the I-name and
the O-name).

The  following \emph{law}  allows us to transform a free output into a
bound output, using  a wire.
We recall that \( \bout a b : P \) is \( \resBP b (\out a b \mid P ) \).

\begin{lemma}
\label{l:wire}
$ \bout c b : \ulinks {\comp b} a        \wbcTT \Delta{}  \out c a $. 
\end{lemma}

Thus, if $a$ is an O-name, the law says that
$ 
\bout  c b : \links {\comp b} a \wbcTT \Delta{}
\out ca $, whereas  
 if $a$ is an I-name, the law says that
$ 
\bout  c b : \links a {\comp b}  \wbcTT \Delta{}
\out ca $.
The  law for O-names has been widely used in  \emph{localised}  $\pi$-calculi, 
 where only the output capability of names may be exported
 \cite{Boreale98,MerroSangiorgi04,Yoshida02}. In contrast,
the wire law for I-names is specific to \api: the constraint on the
input end of a channel being owned, at any time,  by a single process,
is essential.

Theorem~\ref{t:wire-subst-law}
formalises
a behavioural effect of wires, relating them to substitutions. 

\begin{theorem}[Wire substitution law]
\label{t:wire-subst-law}
$ $  \newline
$ 
\resBP b  ( \ulinks {\comp b }a| P ) 
 \wbcTT \Delta{}
 P \sub a{b}$,   with $\comp b$  fresh for $P$. 
\end{theorem}
Again, it is useful to spell out the meaning of the law in the theorem
according to whether  the
substituted name $a $ is an O-name  or an I-name.
In the case of an  O-name, the law becomes: 
\begin{equation}
\label{e:oslaw}
\resBP b  ( \links {\comp b }a| P ) 
 \wbcTT \Delta{}
 P \sub a{b}
\end{equation} 
where `$\comp b$  fresh for $P$' can then be left  implicit, as otherwise 
$ \links {\comp b }a| P$ would not type.
Whereas for  an I-name,  the law becomes
\begin{equation}
\label{e:islaw}
\resBP b  ( \links a {\comp b }| P ) 
 \wbcTT \Delta{}
 P \sub a{b}   \mbox{ with $\comp b$  fresh for $P$. }
\end{equation} 

In this case it is implicit, by similar typing issues, that 
$a$ should be fresh for $P$.
Technically, the proofs for the two cases are quite different and, moreover, their dependencies on Lemma~\ref{l:wire} are subtle.
Lemma~\ref{l:wire} for I-names is used in the proof of Theorem~\ref{t:wire-subst-law} for I-names, but we use Theorem~\ref{t:wire-subst-law} for O-names to prove Lemma~\ref{l:wire} for I-names.
See Section~\ref{s:ls} and
Appendix~\ref{a:laws}
for details.

Another way of expressing law \reff{e:islaw} above, and making explicit all typing assumptions, is the
following one. 
\begin{corollary}
\label{c:ur}
Suppose $a \in \iname \Delta $, and $b, \comp b $ are fresh names. Then  
$ 
\resBP b  ( \links a {\comp b }| P  \sub b a  ) 
 \wbcTT \Delta{}
 P$.
\end{corollary}
 Corollary~\ref{c:ur}  shows that any free name that a process owns in input,
 like $a$ in the corollary, may be made~-- in  process calculus  terminology~-- 
 \emph{uniformly receptive} \cite{Sangiorgi97}:
that is, the name is immediately available (i.e., placed at the outermost level) and is always
available (i.e., replicated).  For this, however, a new  name has to be introduced~-- name
$b$ in the corollary~--- that  is \emph{private}  to the process  and need not be
uniformly receptive.
By repeatedly applying the law, all free input names of a process may be made uniformly
receptive, and moreover, such a property remains invariant under reductions (and labelled transitions, introduced in Section~\ref{s:ls}).
\iflong
\fi

Lemma~\ref{l:commutativity} shows that inputs can be commuted:
\begin{lemma}
\label{l:commutativity}
$\inp a x . \inp by . P
 \wbcTT \Delta{}
\inp by .
\inp a x .  P
$
\end{lemma}
The (analogue of the) commutativity of input prefixes  may be found in the literature.
For example, in asynchronous session-subtyping~\cite{MostrousYH09}, as
a process representation of the commuting conversions of linear
logic~\cite{BellinScott94,Wadler14,KokkeMP19}, in concurrent game semantics~\cite{GhicaMurawski08,CastellanCRW17} (under the name \emph{courtesy}).\footnote{
The commutation of inputs is often noted together with a condition saying that neither inputs nor outputs depend on outputs (when no causality is induced by name bindings); this condition also holds in \api{} because outputs do not have continuations.}
While the commutation of inputs holds for asynchronous may-testing equivalence~\cite{BorealeNP02}, 1-input property is needed to justify this in a branching equivalence.

Receptiveness and courtesy are the key conditions imposed to strategies in concurrent game semantics.
It has been proved that the two are sufficient and necessary
conditions to make strategies remain invariant under their composition
with copycat~\cite{CastellanCRW17}, the game semantic counterpart of
the wire process.\footnote{More precisely, copycat strategies
  correspond to `dynamic' wires, i.e.,  wires written as
  internal processes (\( \transf{\links a b} \) introduced in Section~\ref{s:ls}).}

In Lemma~\ref{l:otherlaws}, the first two laws  again rely on the 1-input property. In the
first, the property is needed as otherwise the input might `steal' messages sent to other inputs.
In the second one, the replication at $a$ guarantees that the output at $b$ will eventually be
consumed (in $\pi$-calculi a similar law  requires $a$ to have a uniform-receptive type
\cite{SangiorgiWalker01}).  The third law is a useful unfolding law for replication and,
 in contrast, it also holds in $\pi$-calculi; it shows that in \api a non-replicated and
a  replicated input at the same name may coexist, as long as the former precedes the latter.

\begin{lemma}
\label{l:otherlaws}
\noindent
\begin{enumerate}
\item
$\inp a x .\nil
 \wbcTT \Delta{}
\nil $ 

\item
$! \inp a x .P | \out b v
 \wbcTT  \Delta  {\delta }
! \inp a x .P | P \sub v x
 $, if $a\sepRES b \in \delta $.
\iflong
\fi
  \item
    \label{law:unf}
    $! a(x).P  \wbcTT   \Delta{ }  a(x). (P | ! a(x).P)$.
\end{enumerate}
\end{lemma}

\iflong
\fi

\iflong 
\section{Labelled Transition System and Proof Techniques}
\else
\section{LTS and Proof Techniques}
\fi
\label{s:ls}

The goal of this section is to define proof techniques for behavioural equivalence in
\api.  As usual in process calculi,  for this we  set a Labelled Transition System (LTS).
Using the LTS, we first establish a correspondence with ordinary Asynchronous
$\pi$-calculus (\Api), in order to import results from its theory.
\iflong,  based on synchronous and asynchronous
forms of bisimilarity. 
\fi
However, to be able to handle also equalities
that are specific to \api and do not hold  in \Api,
we also    introduce a new form of labelled bisimilarity, and then prove  soundness and completeness
results for it with respect to barbed congruence. 
\iflong
Such proof techniques  however cannot be used to reason about equalities that only hold
in \api, that is, equalities whose counterparts fail in \Api. Indeed ordinary asynchronous
bisimilarity is sound but not complete for barbed congruence in \api. 
Therefore we  also  introduce a new form of bisimilarity, that is specific to \api. 
\fi
\iflong
Indeed, another motivation  for this  study is to see whether 
\api is sufficiently robust to be able to derive labelled characterisations of
barbed congruence, similarly to what happens in other dialects
of $\pi$-calculus, including $\pi$-calculus itself and \Api. %
\fi
 To facilitate the reading we omit type information, here and in the following section.

The LTS, in Figure~\ref{fig:LTS}, is parametrised  by  a
{connection set} \( \delta \); 
the LTS will be used in settings in which
   $\ok \Delta P$, for some $ \Delta $ with  
$\n \delta \subseteq \n \Delta $. 
Intuitively, 
$\ltsTv { \delta } P {\mu}  {P'}$ means that, assuming the  connections
in $\delta$, the process  $P$ is capable of performing  action $\mu$, thus
producing the derivative $P'$. 
In the  LTS, the  connections in  $\delta$  are used in rules \rn{Com} and \rn{Close} to derive
 $\tau$ actions (i.e.,  interactions)
between names related in $\delta $. 
\iflong
 Indeed, the only rules in which $\delta$ is looked at
are the communication rules \trans{Com} and \trans{Close}. 
In the premises of the 
rules for restriction, $\delta$ is extended so to add the new pair of
connected
 names.  
\fi
 In line with previous notation, a label 
 $\bout a c  $ stands for  $\resBP c  \out a c $
 (i.e.,   a restricted name  is exported, in which case both the name
 and its companion are extruded).
In rule \trans{Inp}, the formal parameter of an input remains
uninstantiated; instantiation occurs only later, in the communication
rules (\trans{Com} and \trans{Close}). 
Further material for this section may be found in
Appendix~\ref{a:ls},
including a Harmony Lemma relating $\tau$-transitions in the LTS and reductions in the reduction
semantics (Lemma~\ref{l:hl}).  
\iflong
 
We sometimes call \emph{ground} a transition derived in the  LTS with $\delta$ 
empty, as in   $\ltsTv { \emptyset  } P {\mu}  {P'}$.

 We abbreviate 
$\ltsTv { \delta } P {\mu}  {P'}$ as $P \arr\mu P'$ when $ \delta$ is not needed for
deriving that transition; that is,  $\mu \neq \tau$ or 
($\mu = \tau$ and 
$\ltsTv { \emptyset  } P {\mu}  {P'}$).

\fi

\begin{figure}[t]
\[\trans{Inp}\infR {}
{
\ltsTv \delta {\inp a x .  P} {\inp a x} P}
\andalso
\trans{Out}\infR {}
{\ltsTv \delta {\out a b } {\out a b} \nil}
\andalso
\trans{ParL}\infR {\ltsTv \delta P \act {P'}}
{\ltsTv \delta{ P \mid Q } \act {P'\mid Q}}
\]
\[
  \trans{Com}\infR {\ltsTv \delta P {\inp a x} {P'} \andalso
\ltsTv \delta{Q} {\out b {c}}{Q'} \andalso a\sepRES b \in \delta } 
{\ltsTv \delta{ P \mid Q } {\tau} {P' \sub {c}x \mid Q'}}
\]
\[
\trans{Close}\infR {\ltsTv \delta P {\inp a x }{P'} \andalso
\ltsTv \delta Q   {\bout b {c}} {Q'}
\andalso a\sepRES b \in \delta
} 
{
\ltsTv \delta{ P \mid Q } {\tau}{ \resBP c ( P' \sub
  { c}x \mid Q')}}
\]
\[
 \trans{Res}\infR {\ltsTv {\delta, \breve b} P  \act {P'}
   \andalso \pairR  b \cap  \fn \act = \emptyset}
 {\ltsTv \delta{\resBP b  P } \act {\resBP b  P'}}
\andalso
  \trans{ResExt}\infR {\ltsTv{\delta, \breve b} P  { \out c {b} } {P'}}
{
  \ltsTv \delta {\resBP b  P}  {  \bout  c {b} } {P'}}
\]
\caption{The LTS for \api{}}
\label{fig:LTS}
\end{figure}

\mysubsection{Transporting results from ordinary $\pi$-calculi onto \api}
\label{ss:trans}
We can straightforwardly encode an \api{} process into an  Asynchronous \( \pi \)-calculus (\Api{}) process if we are given the connection set \( \delta \).
We call the pairs \( \compP P \delta \) \emph{composite processes}.
The encoding
$\encoA { \compP P \delta }$ simply maps every pair of names \( a\sepRES b \)
  in \( \delta \) to the same name, say the O-name  \( b \), and 
similarly underneath any restriction
 \( \resB a b \) in \( P \).
Details are reported in
Appendix~\ref{aa:trans},
where we
develop a (strong) bisimilarity  between the source
 composite  \api process and
its target \Api process. 
Building on this, we can import a number of well-known definitions and
results from the
theory of the $\pi$-calculus onto \api. 
These include: 
 ordinary strong and weak
\iflong
, synchronous and asynchronous,   
\fi
 bisimilarity ($\sbS$ and $\wbS$), 
the expansion relation ($\expa$);   
 (pre)congruence
 results 
and up-to techniques
for these relations; 
soundness of these relations with respect to barbed congruence; 
algebraic laws and their validity; 
properties about action transitions, and about the interplay between
transitions and
 name
substitutions. 
Notably, we have: 
\begin{theorem}
\label{t:Api_api}
Suppose $\ok \Delta {P,Q}$. 
Then $\encoA{\compP P \emptyset } \wbS  \encoA{\compP Q \emptyset }$ implies 
  $P \wbcTT \Delta {} Q$.
\end{theorem}

\mysubsection{Internal Bisimilarity}
\label{sec:internal}
An \emph{internal process}~\cite{Sangiorgi96a} is a process that may never perform free
output actions.  
The goal here is to define a labelled bisimilarity, which will later be
shown to coincide with barbed congruence, on a subset of internal
processes rich enough to be able to represent the behaviour of all
\api processes. 
Thus we first
 define a translation \( \transf {-}\) from \api{} to  internal
\api processes.
The translation consists in repeatedly  applying the wire  law of
Lemma~\ref{l:wire}, in order to ensure that all names emitted in an
output are fresh. 
Accordingly,  the translation is a homorphism on all operators except for outputs, where we have: 
\[ 
\transf{{\out a b}} \defi
\bout a c : \transf{{\ulinks {\comp c} b}}    \andalso \mbox{(for  $c$ and
  $\comp c$  fresh)}
\]

(We recall that 
  $\ulinks  a b $ stands for either
$\links  a b $ or 
$\links  ba $, depending on the types for $a$ and $b$.)
The encoding terminates because   the calculus is simply-typed and hence
the wire law is applied a finite number of times. 
We call  \apiE the set of  target  processes of the translation. 
More details on the material in this section are in
Appendix~\ref{a:inter}.

The translation preserves typing
\iflong
 (the typing environment for
source and target terms of the translation are the same)
\fi
and  respects barbed congruence. 
\begin{lemma}
\label{l:wire_trans}
For any $P \in$ \api  and any $ \Delta $, if $\Delta \vdash P$ then also 
$\Delta \vdash   \transf P$;  moreover,
$P \wbcTT \Delta{}  {\transf P}$. 
\end{lemma}

Lemma~\ref{l:du} shows that \apiE processes obey a strict duality
discipline: a process that exports a name retains the companion
of that name, but not the name itself. This  property, together with the property that 
 internal processes may only perform \emph{bound} outputs,  
are  useful
in definitions and applications of proof techniques (notably labelled bisimilarity),
and motivate the transformation to internal processes.

\begin{lemma}
\label{l:du}
If $P \in \apiE$ and $\ltsTv \delta P {\bout b c} P'$ then $c \not \in \fn {P'}$.
\end{lemma}

A \emph{typed  (process) relation} is a  set of quadruples
$(\asetD, \Delta, P,Q)$, in which $\n \delta \subseteq \Delta $, and
$\ok \Delta {P,Q}$. 
\iflong Assuming a Curry-like assignment of types to names, 
\fi
We usually omit $\Delta $ and write 
$P\sor \S \asetD  Q$ if 
$(\asetD, \Delta, P,Q) \in \S$, for some $\Delta $.

In the bisimilarity below, the set of connections $\delta $
is relevant in two aspects: first, the LTS uses $\delta $ for
inferring synchronisations at names in $\delta $; secondly, 
I-names that are mentioned in $\delta$ are
owned
by the tested processes;
consequently,  
  outputs at the companion O-name (itself a name in $\delta $) are 
 not observable (i.e.,   visible from the
environment). This explains the condition 
`$b \not \in \n \asetD$' in clause \reff{i:out}.
Weak transitions, 
$\LtsSS \asetD {Q} {}  {\aset}{ Q'}$ and 
 $\LtsSS \asetD Q {\mu}  {\aset}{ Q'}$, are defined in the expected
manner (Appendix~\ref{a:inter}).
Both the definition of weak transitions, and the invariance of the
connection set $\delta $ in clause~\reff{i:out} of the bisimulation rely on
Lemma~\ref{l:du} (for instance, without Lemma~\ref{l:du} the connection set 
 might need to be extended, after a transition, which in turn  would impinge on the  composition of  transitions
 into `weak' transitions).
The bisimulation is in the `ground' style \cite{SangiorgiWalker01},
that is, no substitution is required on 
the formal parameter  of an input action; this
makes the  proof of congruence harder, but  
 simplifies applications
 of the technique.
Clause~\reff{i:in} is the usual input clause of asynchronous bisimilarity, tailored to fit I/O separation.
\iflong 
\fi
\begin{definition}
\label{d:hb-bisimulation}
 A symmetric typed relation $\R$ on \apiE processes is an
  {\em  internal bisimulation} if whenever $P\sor \R \asetD  Q$
 the following holds:
 \begin{enumerate}

    \item If 
$\ltsSS \asetD P \tau  {\aset}{ P'}$, then there is $Q'$ such that
$\LtsSS \asetD Q {}  {\aset}{ Q'}$
 and $P' \sor{\R}{\asetD} Q'$;

    \item
\label{i:out} 
if  $
\ltsSS \asetD P { {\bout b c}}   {\aset}{ P'}$
and   $b \not \in \n \asetD$,
         then there is $Q'$ such that 
$\LtsSS \asetD Q {{\bout b c}}  {\aset}{ Q'}$
 and
$P' \sor{\R}{\asetD} Q'$

     \item \label{i:in} If
$\ltsSS \asetD P {a(c)}  {\aset'}{ P'}$,
   then  there is $ Q'$ such that:
         \begin{enumerate}
           \item either 
$\LtsSS \asetD Q { a( c)}  {\aset'}{ Q'}$
 and
$P' \sor{\R}{\asetD} Q'$
           \item or 
$\LtsSS \asetD Q {}  {\aset}{ Q'}$
 and, either
\begin{enumerate}
\item
 $a\sepRES a' \in  \asetD$, for some $a'$,  and
$P' \sor{\R}{\asetD }
(Q' | {\outbis  {a'}  c})
$, or

\item   
$a \not \in \n \asetD$ and
$P' \sor{\R}{\asetD, a\sepRES a' }
(Q' | {\outbis  {a'}  c})
$, for $a' $ fresh. 
\end{enumerate}
          \end{enumerate}
 \end{enumerate}

Processes $P$ and $Q$ are  \emph{internal  bisimilar at $\asetD$},
written $P \sor \wb \asetD Q$, if $P  \sor \R \asetD Q$ for some
internal bisimulation $\R$.
We abbreviate $\sor \wb \emptyset $ as $\wb$.
\end{definition}

Some `up-to' techniques for internal bisimilarity, including `internal bisimulation up-to $\wbS$ and $\expa$', are given in
the Appendix.

The next goal is to derive 
\iflong
the substitutivity properties 
\else
congruence properties
\fi 
for
internal bisimilarity,  from which   
 soundness with
respect to barbed congruence follows.
In  $\pi$-calculi, when we examine substitutivity properties of
a bisimilarity with respect to the language constructs, 
 parallel composition is delicate, as it is the
operator by means of which interactions among processes are derived,
whereas the restriction operator is straightforward. 
In \api, in contrast, restriction is the most delicate
operator, because it creates a connection between an
I-name and an O-name and therefore enables interactions. 
In the remainder of the section, all processes are meant to be in \apiE.

\begin{lemma}[Restriction, in \apiE]
\label{l:res}
If $P \wbT \delta Q$ and $z,z' \not\in \n \delta  $,  then also  
$\resB z{z'} P \wbT \delta \resB z{z'} Q$.
\end{lemma} 

In the assumption of the lemma, synchronisations between names $z,z'$
are not possible, and may not be inferred; only visible transitions
at these names are observed. In the conclusion, in contrast, only the
synchronisations are taken into account. 
The proof  of the lemma
 implicitly subsumes 
 properties of closure under
substitutions, for it consists in showing that the relation
\[ ( \delta ,  \res \tilab \res \tilz (P \sigmaapb ), \res \tilab
\res \tilz
 (Q \sigmaapb )     \]
with 
 $P \wbT{\delta, \tila \sepRES \tilap} Q$
and $\tilap \cap \tilb = \emptyset $ is an internal bisimulation 
up-to $\wbS$ and $\expa$ (Lemma~\ref{l:crux}).  

Two typings $\Delta $ and $ \Gamma $ are \emph{compatible} if they are disjoint on I-names
(i.e.,  $\iname \Delta \cap \iname \Gamma  = \emptyset $) and agree on the type assignments
for shared O-names.

\begin{lemma}[Parallel composition, in \apiE]
\label{l:par}
  Suppose that \( P \wbTT \Delta {\delta}  Q \), that $ \ok \Gamma R$, and that 
$\Delta$ and $\Gamma$ are compatible.
Then also 
 \( P \mid R \wbTT{\Delta , \Gamma }{ \delta }  Q \mid R\).
\end{lemma}

In the proof, we take the set $\R$ of all (well-typed) triples of the form 
\[
(
 \delta ;  
\res {\tily} \res {\tilx} ( P \mid R)  ; 
\res {\tily}  \res {\tilx} ( Q \mid R)  ) 
\]
where, for some $\Delta $ and $\Gamma $:  
\begin{itemize}
\item
 \( P \wbTT \Delta{ \delta, \tilx }  Q \) ;
\item $\n \delta , \tilx \subseteq \n \Delta $ ;
\item  for each $y \sepRES y' \in \tily$, the set $\{y,y'\} \cap \n \Delta $ is a
    singleton;

\item  $ \ok \Gamma R$ and 
$\iname \Delta \cap \iname \Gamma  = \emptyset $;
\end{itemize} 
and show that
 $\R$ is an
internal bisimulation up-to $\wbS$ and $\expa$.

Both lemmas above  rely  on various results relating transitions of internal processes
under  extensions or restrictions of the connection set on which the LTS is parametrised;
and results relating variations on the connection set to the addition of restricted names to
processes, and to the application of substitutions.

\begin{theorem}[Soundness]
\label{t:sound}
Let $P,Q \in \apiE$.
If $P  \wbTT\Delta{\delta}  Q$ then also $ P \wbcTT \Delta {\delta}   Q$.
\end{theorem} 

\iflong
\begin{corollary}[Plain soundness]
\label{c:sound}
Let $P,Q \in \apiE$.
If $P  \wbTT\Delta{}  Q$ then also $ P \wbcTT  \Delta{}   Q$.
\end{corollary} 
\fi             

As usual in process calculi, 
completeness for barbed congruence is established for image-finite
processes, exploiting  the stratification of the labelled bisimilarity on the natural
numbers (Appendix~\ref{a:completeness}).
\iflong
The proof makes use of the fact that internal choice can be represented in \api{} (cf.~Example~\ref{ex:internal-choice}).
\fi
  A \emph{regular} process is a process without success names; these
  are the processes of interest, intended to be tested in closing contexts.

\begin{restatable}[Completeness]{theorem}{Completeness}
\label{t:completeness}
Let $P$ and $Q$ be two image-finite regular \apiE-processes. Then 
$ P \wbcTT  \Delta{\delta}  Q \; \; {\rm implies} \; \;  P 
\wbTT   \Delta \delta Q$. 
\end{restatable}

\mysubsection{Proofs of the algebraic laws}
\label{ss:pal}
 By virtue of the Soundness Theorem~\ref{t:sound}, we can use  the proof techniques based
on internal bisimilarity to validate laws of Section~\ref{s:alaws}, notably 
the `wire substitution law' of Theorem~\ref{t:wire-subst-law}.  
The proofs for the two cases of the theorem (for O-names and I-names as the substituted
names, laws \eqref{e:oslaw} and \eqref{e:islaw})
are quite different.
The first, 
\eqref{e:oslaw},
is proved by  structural induction on processes  and using    algebraic
reasoning inherited from the $\pi$-calculus.
In contrast, 
the second, 
 \eqref{e:islaw},  makes use of the more advanced proof techniques
based on the transformation to internal processes: internal  bisimilarity, and up-to techniques for it.
The proof also uses induction on the order of the types of the names in the wires of the
law, as well as some strengthening results for $\wbT \delta  $ with respect to the
parameter $\delta$.
See
Appendix~\ref{a:laws}
for the details.

The proofs of the laws in Lemma~\ref{l:otherlaws} are simpler; they still use the
transformation to internal processes and internal bisimilarity.

\section{Expressiveness: Encoding of \texorpdfstring{\ALpi{}}{ALpi}}
\label{s:expr}

To show the expressiveness of  \api, we study the encoding
of the (simply-typed) Asynchronous Localised $\pi$-calculus (\ALpi),
a well-known subset of the Asynchronous $\pi$-calculus
\iflong
 in which only the output
capability of names may be exported, and 
\fi
still
sufficient to encode the  Asynchronous $\pi$-calculus \cite{MerroSangiorgi04,SangiorgiWalker01}.
The syntax of  \ALpi is as follows: 
\[ 
P \Coloneqq \nil \midd  P | Q\midd   \inp a x .{P}
\midd \out bc \midd \res a  P
\midd ! \inp a x .P
\]
with the 
\iflong
additional
\fi constraint that only the output capability of names may be
communicated. Thus, underneath an input 
$\inp a x .{P}$,  name $x$ may appear free in $P$ only within output
prefixes, and similarly for replicated inputs. 
So \ALpi is more liberal than \api in that
it allows multiple concurrent input occurrences of the same name
(which implies that there are also no constraints on the use of input
names underneath a replication); however it forbids 
communicating 
\iflong
the input capability of names.
\else
 input capabilities.
\fi

\begin{figure}
\begin{align*}
\encoL{\res a P } f  &\defeq
(\res{\pairR a, \pairR b, \pairR c})
( \encoL P {f,a\mapsto b} \mid \out c {\comp a, \comp b} \mid !\inp {\comp c} {a',b'}.
\inp {a'}x.\inp {b'}y.
 (\out c{a',b'} \mid
\out yx ) \\
\encoL{ \inp a y . P  } f  &\defeq \bout  {f_a} z: \inp z y. \encoL P f \\
\encoL{!  \inp a y . P  } f  & \defeq
\resBP d (\outC {\comp d} | ! d . (\outC {\comp d} | \encoL { \inp a y .P} f  )
\\
\encoL{ \out  a  b  } f  & \defeq \out a b
\qquad \encoL{ P|Q  } f  \defeq \encoL{ P  } f | \encoL{ Q  } f
\end{align*}

\caption{The encoding of \ALpi{} in \api{}}
\label{fig:fromALpi}
\end{figure}

We begin with the translation of terms given in Figure~\ref{fig:fromALpi}, leaving types aside for readability.
Here, $\comp a, b,\comp b, c, \comp c, d, \comp d$ are all supposed to be fresh names.
In the encoding, each name of \ALpi  is mapped onto two names of \api, one representing the input capability, the other the output capability.
For convenience, we assume that  each name $a$ of \ALpi is an O-name of \api.
In one place for readability we use a binary communication (polyadicity is
formally introduced in Section~\ref{sec:esyntax}); it could  be avoided using  the encoding
of polyadicity in name-passing calculi \cite{SangiorgiWalker01}, as the two names carried
have the same type. 
The encoding is parametrised on a (finite) function $f$, that tells us, for each name $a$
of \ALpi the I-name of \api that $a$ is mapped to.
We write $f_a$ for $f(a)$.

The key clause in the encoding is that  of a restriction $\res a P$. 
The name $a$ may appear several times in input in the \ALpi process $P$, possibly in   parallel. 
In the encoding, a small server for $a$ is set up,
\iflong 
 called the \emph{restriction server}, 
\fi
that
collects requests from processes that would like to perform inputs  at  $a$; such requests
use an auxiliary name, say $b$. The
\iflong
 restriction
\fi server receives a value emitted at $a$,
then listens for a request at $b$, and finally forwards the value collected along $a$ to
the appropriate process. 
Accordingly, the encoding of an input process at $a$  emits a request
for it at $ f_a$,
including 
a (linear) return  name at which a value for $a$ will be received.  

The encoding of a replication must take into account the constraint of \api  forbidding
free inputs in the body of the replication. 
The constraint is met because,  as mentioned, 
  a free input is
translated into a free output followed by a bound input.
This is  however not sufficient because
replications should be input-guarded. For this, an additional input (name $d$ in the
encoding below) is added, as a trigger for the body of the replication. 
In a behavioural equivalence that abstracts from $\tau$-steps (e.g.,
barbed congruence, similarity, or bisimilarity)
 the additional triggering is not observable. 
\iflong
\fi

\iflong
We now describe the translation of types.
In simply-typed \ALpi, the grammar of value types is:
\[  T \Coloneqq \Och \ty \mid \unittype   \]
and a typed  restriction is therefore 
$(\res {a: \Bch \ty }) P $, where  $ \Bch \ty$ is the type of a name that could be used
both in input and in output.
When translating such a process, types of the names $\comp a$, $\comp b$,  
and $\comp c$ that are thus introduced are, respectively, 
$\Ich \ty$, $ \Ich{\Och \ty}$, and $ \Ich{(\Ich \ty \times  \Ich{\Och
    \ty})}$ (where $T_1 \times T_2$ indicates a pair of types); and those of the companion names $a,b$ and $c$ are
their dual.
Finally, the name $d$ introduced in the encoding of replication
has type $\Ich \unittype $.

The correctness of the encoding is studied in Appendix~\ref{a:expr}. 
\else
The encoding of types, and  correctness, are discussed
in Appendix~\ref{a:expr}.
\fi
In the light of the  correspondence between \api and 
Asynchronous
$\pi$-calculus (\Api), in Section~\ref{ss:trans},
  we first  show that the  transformation applied in the
encoding
 is valid within \Api, exploiting the
theory of \Api. 
The operational correspondence between an \ALpi process $P$ and its
encoding
$ \encoL P  \emptyset$ in \api can be formalised 
by means of  similarity equivalence $\wse$ 
(i.e., similarity in both directions)
 on closed processes. 
As the encoding may `split' an input into two actions, in
general similarity equivalence cannot be replaced by bisimilarity. 
 We recall that
 on closed processes
the only free names are the success names, which are not modified in
the translation.   
\begin{lemma}
\iflong 
  If \( P \) is a \ALpi{} process, then \( \encoL P \emptyset \) is a \api{} process.
\else
  If \( P \) is in  \ALpi, then \( \encoL P \emptyset \) is  in  \api.
\fi
\end{lemma}
\begin{theorem}
\label{t:expr_pi}
\iflong 
If $P$ is a closed \ALpi process, then
$ \encoL P \emptyset  \wse P$. 
\else
If $P$ is in  \ALpi and closed, then
$ \encoL P \emptyset  \wse P$. 
\fi
\end{theorem}

Since \Api{} can be encoded into \ALpi{}~\cite[Section 5.6]{SangiorgiWalker01}, combining the above result with this encoding  yields an encoding from \Api{} to \wapi{}.
Furthermore, as the synchronous \( \pi \)-calculus (without matching and mixed choice) can be encoded into \Api{}~\cite[Section 5.5]{SangiorgiWalker01}, closed processes of the ordinary  \( \pi\)-calculus can also be encoded into \wapi{}.

\section{Extended Calculus}
\label{sec:esyntax}
We extend the core calculus introduced in Section~\ref{sec:syntax} with tuple types, sum types and linear types.
This is done because adding these types allows us to construct categories with richer structure in the subsequent sections.
The extension is  standard~\cite{SangiorgiWalker01}.
Algebraic laws and proof techniques from the previous sections such as Sections~\ref{s:alaws} and~\ref{s:ls} also hold for this calculus.

The grammar of the processes now is
\begin{align*}
  \text{(Process)} \quad  P
  &\Coloneqq \nil \midd  P | Q\midd   \inp a x .{P}  \midd \out bv \midd \resB* a b  \ty P \midd ! \inp a x .P \\
  &\ \midd \pmTuple{x_1, \ldots, x_n} \val P \midd \binCase \val {x_1} {P_1} {x_2} {P_2} \\
  \text{(Values)} \quad \val &\Coloneqq x \mid \baseval \mid \tuple{\val_1, \ldots, \val_n} \mid \inl \val \mid \inr \val
\end{align*}

Constructs on the second line are the new ones.
The new construct  \( \pmTuple{x_1, \ldots, x_n} \val P \) is a destructor of a tuple.
We use tuples rather than pairs to keep our syntax close to that of the polyadic \( \pi \)-calculus; we write \( \inp a {x_1, \ldots, x_n} . P \) as an abbreviation of \( \inp a x . \pmTuple {x_1, \ldots, x_n} x P \).
The other new construct \( \binCase \val {x_1} {P_1} {x_2} {P_2} \) is
the pattern matching construct for sum types; here the name \( x_i \) is bound with scope \( P_i \).
The dynamics of these constructs are defined as usual (see
Appendix~\ref{app:esyntax}
).

The definition of values should be self-explanatory.
Points to note are, in the above grammar, \( x \) ranges over the set of names and \( \baseval \) ranges over basic values of the base types of our calculus, such as integer constants.\footnote{We may easily extend the calculus with \emph{value expressions}~\cite{SangiorgiWalker01} such as the successor operation \( \mathsf{succ}(n) \). Here, we shall not do so to keep the exposition simple.}

We now explain the type system, which is a variant of the linear type system~\cite{KobayashiPT99}.
For technical convenience, we classify types into two kinds~--- value type and channel type~--- which are mutually recursively defined.
Its grammar is given as follows:
\begin{align*}
  \text{(Channel type)} \quad \ty, \tyTwo &\Coloneqq \ich{\valty} \mid \och{\valty} \mid \lich{\valty} \mid \loch{\valty} \\
  \text{(Value type)} \quad \valty, \valtyTwo &\Coloneqq \btype \mid \ty \mid \valty_1 \times \cdots \times \valty_n \mid \valty_1 + \valty_2
\end{align*}

Here \( \btype \) is a metavariable over basic types such as the unit type used in Section~\ref{sec:syntax}.
The types \( \lich{\valty} \) and \( \loch{\valty} \) are types for linear (or \emph{affine} to be precise) input and output channels, respecively.

We briefly explain the modifications to the typing rules that are needed to incorporate linear types.
(The treatment of product and sum types should be evident.)
The rules \rn{Out} and \rn{In} in Figure~\ref{fig:typing} are modified as
\begin{gather*}
\textsc{Out}\inferrule{\ty \in \{ \och{\valty}, \loch{\valty} \} \quad \tyenv \vdash_v \val : \valty }{\tyenv, a : \ty \vdash \out a \val}\qquad
\textsc{In}\inferrule{\tyenv, a : \ich {\valty}, \seqq b : \valty \vdash P}{\tyenv, a : \ich {\valty} \vdash \inp a {\seqq b}. P} \qquad
\textsc{InL}\inferrule{\tyenv, b : \valty \vdash P}{\tyenv, a : \lich {\valty} \vdash \inp a b. P}
\end{gather*}
with a suitable notion of value typing \( \tyenv \vdash_v \val  : \valty \).
In \rn{InL}, the name \( a \) cannot appear free in \( P \).
The rule \rn{Par} is the same as before, except that \( \tyenv_1, \tyenv_2 \) should now be read as the combination of \( \tyenv_1 \) and \( \tyenv_2 \) that takes linearity into account; for example if \( a : \lich{\valty} \in \tyenv_1 \cap \tyenv_2 \) then \( \tyenv_1, \tyenv_2 \) should be undefined.
We also need to take care of the distinction between copyable and linear types in the rules.
For example, the rule \rn{RIn} is now
\[
  \textsc{RIn}\inferrule{\tyenv\ \text{copyable} \quad  \quad \iname{\tyenv} = \emptyset \quad \tyenv, \seqq b : \seqq \ty \vdash P}{\tyenv, a : \ich {\seqq \ty} \vdash  !\inp a {\seqq b}. P}
\]
The difference compared to the previous rule is that \( \tyenv \) may contain names that are not output channels, but are names with a copyable type such as a pair of integers.
See
Appendix~\ref{app:esyntax}
for the full list of the typing rules.

A type environment \( \tyenv \) is \emph{value closed} if \( \tyenv(a) \) is a channel type for all \( a \) in \( \tyenv \).
Similarly, a process is \emph{value closed} if it is typed under some value closed type environment.
In the rest of this paper, we shall use the word process to indicate value closed processes unless we explicitly say that there is a free name with a value type.

\renewcommand*{\category}[1]{\mathbf{#1}}
\newcommand*{\catC}{\category{C}}
\newcommand*{\catL}{\category{L}}
\newcommand*{\id}{\mathit{id}}
\newcommand*{\der}{\mathsf{der}}
\newcommand*{\prom}[1]{{#1}^\dagger}
\newcommand*{\emptyenv}{\bullet}
\newcommand*{\Proc}{\mathbf{Proc}}
\newcommand*{\ProcN}{\mathbf{Proc}^{-}}
\newcommand*{\ProcS}{\mathbf{Proc}_s}
\newcommand*{\ProcSL}{\mathbf{Proc}^l_s}
\newcommand*{\env}{\Gamma}
\newcommand*{\envTwo}{\Delta}
\newcommand*{\negenv}{\Gamma_{-}}
\newcommand*{\negenvTwo}{\Delta_{-}}
\newcommand*{\barnegenv}{\bar{\Gamma}_{-}}
\newcommand*{\barnegenvTwo}{\bar{\Delta}_{-}}
\newcommand*{\op}{\mathrm{op}}
\newcommand{\evm}{\mathrm{ev}}
\newcommand{\with}{\mathop{\&}}
\newcommand{\intr}[1]{\llbracket #1 \rrbracket}

\section{Some Categories of Processes}
\label{sec:category}

In this section, we show that, as claimed in the introduction, typed
processes of the extended calculus can be organised into a relative
Seely category. This will be done in two
steps. First, we will describe an ambient category $\Proc$ whose
objects are types (or type environments), and morphisms are processes~---
as in concurrent games~\cite{CastellanCRW17,CastellanCW19}, this category will turn out to be compact
closed, but lacks adequate structure to interpret programming languages.
To cope with this,
we then construct a subcategory of so-called \emph{negative} processes,
importing the familiar game semantics condition that Opponent should
always start. We show that negative processes can be organised as a
\emph{relative Seely category}, a weakening of Seely
categories~\cite{Clairambault24} sufficient to ensure that the Kleisli
category is cartesian closed, and thus a reasonable target for the
interpretation of higher-order programming languages.
In addition, we show that negative processes also have the sequoidal structure~\cite{Laird02}, which is used in game semantics to model stateful computation.

\subsection{A Compact Closed Category of Processes}

The objects of $\Proc$ are finite sequences of types, with the empty
sequence written $\emptyenv$. We use $\bar{\env}, \bar{\envTwo}$ as
metavariables for those; if $\env$ is a type environment we use
$\bar{\env}$ for the underlying sequence of types. A morphism from \(
\bar{\env} \) to \( \bar{\envTwo} \) is a (an equivalence class
of) triples \( (\env, \envTwo, P) \) so that $\env^\perp, \envTwo
\vdash P$, with names chosen so that $\n{\env} \cap \n{\envTwo} =
\emptyset$, and subject to the equivalence relation
$(\env, \envTwo, P) \sim (\env', \envTwo', P')$ when \( \bar{\env} = \bar{\env'} \), \( \bar{\envTwo} = \bar{\envTwo'} \), and moreover,
\[
  P \wbcTT {\env^\perp, \envTwo} {} P'\sub{\seq a}{\seq a'} \sub{\seq b}{\seq b'}
\]
where \( \seq a \) (resp.~\( \seq {a'} \)) are the names appearing in \(
\env \) (resp.~\( \env' \)) and \( \seq b \) (resp.~\( \seq{b'} \)) are
names appearing in \( \envTwo \) (resp.~\( \envTwo' \)).
In short, morphisms are processes modulo typed (weak) barbed congruence
and renaming. We write \( P \colon \bar{\env} \to \bar{\envTwo} \) for the
equivalence class \( [(\env, \envTwo, P)]_{\sim} \).

The identity over \( \ty_1, \ldots, \ty_n \) is the parallel composition
\begin{align*}
  b_1 : \ty_1^\perp, \ldots, b_n : \ty_n^\perp,  a_1 : \ty_1, \ldots, a_n : \ty_n \vdash \ulinks {a_1} {b_1} \mid \cdots \mid \ulinks {a_n} {b_n}.
\end{align*}
of wire processes.
The composition of \( P \colon \bar{\env} \to a_1 : \ty_1, \ldots, a_1 :
\ty_n \) to \(Q \colon b: \ty_1, \ldots, b_n \colon \ty_n \to
\bar{\envTwo} \) is (the equivalence class of):
\begin{align*}
  (\resBU {a_1} {b_1} \cdots \resBU {a_n} {b_n})(P \mid Q).
\end{align*}

\begin{remark}
In the definition of composition, we have slightly abused notations by
specifying the names $a_i$ and $b_i$ in $P$ and $Q$.
We are also (implicitly) assuming that \( P \) and \( Q \) do not share common names to avoid name clashes.
This is harmless
because morphisms are defined up to renaming; we shall make similar
notational shortcuts in the sequel.
\end{remark}

\begin{theorem}
  \label{thm:proc-is-category}
  \( \Proc \) is indeed a category.
\end{theorem}
\begin{proof}
  The associativity of composition holds because
  \begin{align*}
    \res{\seq b} (\res {\seq a}(P \mid Q) \mid R) \equiv \res{\seq a}( P \mid \res{\seq b}(Q \mid R))
  \end{align*}
  assuming \( \seq a \cap \fn R = \emptyset \) and \( \seq b \cap \fn P = \emptyset \) (which can be assumed without loss of generality).

  Wires are indeed identities because, by Theorem~\ref{t:wire-subst-law}, we have
  \begin{equation*}
    \resBP {b_1} \cdots \resBP {b_n}(P \mid \ulinks {b_1} {a_1} \mid \cdots \mid \ulinks {b_n} {a_n} ) \wbc P \sub{\seq a}{\seq b}. \tag*{\popQED}
  \end{equation*}
\end{proof}

The category \( \Proc \) has a symmetric (strict) monoidal structure.
The tensor product \( \bar{\env} \otimes \bar{\envTwo} \) is simply the
concatenation of the two sequences, i.e.~\( \bar{\env}, \bar{\envTwo} \).
The monoidal unit $1$ is the empty sequence $\emptyenv$.
We define the tensor of morphisms \( P \colon \bar{\env} \to
\bar{\envTwo} \) and \( Q \colon \bar{\env'} \to \bar{\envTwo'} \) by
\( P \mid Q \colon \bar{\env}, \bar{\env'} \to \bar{\envTwo},
\bar{\envTwo'} \) (names chosen so as to avoid collisions).
As this is a strict monoidal category, the structural morphisms of the monoidal category structure are simply identities.
The symmetries are %
 the parallel composition of wires performing the
adequate swap. Finally, for any object $\bar{\env}$ we have \(
\eta_{\bar{\env}} : \emptyenv \to \bar{\env} \otimes
\bar{\env}^\perp \) and \( \epsilon_{\bar{\env}} : \bar{\env}^\perp \otimes \bar{\env} \to
\emptyenv \) composed of the obvious parallel compositions of wires.

Altogether, this forms the data of:
\begin{restatable}{theorem}{ProcisCompactClosed}
\label{thm:proc-is-compact-closed}
$\Proc$ is a compact closed category.
\end{restatable}

A compact closed category is a degenerated model of Multiplicative
Linear Logic, where the tensor and par coincide. For semantic
purposes it is rather weak; however readers familiar with game
semantics may recognise familiar ground: basic categories of
unpolarised games tend to have this compact closed structure
\cite{joyal1977remarques,CastellanCRW17}. Also in line with game
semantics, $\Proc$ however lacks all of the structure needed to
interpret richer languages: for a start, it has neither products nor
coproducts, for the same reason as in game semantics
\cite{DBLP:journals/entcs/Mellies05}. In the nullary case, it is easy
to see that it does not have a terminal object: for every candidate
$\top$, a process \( P : \bar{\env} \to \top \)
can ignore $\top$ and perform arbitrary actions in $\bar{\env}$. The
usual game semantics answer to that is to force $P$ to first listen to
the right hand side via a condition called \emph{negativity} amounting
to always letting Opponent start~--- if $\top$ is empty, this forces $P$
to be the nil process. As the design principle of our process language
is that it lets us mimic game semantics, we follow this route here as well.

\subsection{The Category of Negative Processes}

We now aim to enforce the condition that `Opponent always starts'.

On objects, we first say that a type is \emph{negative} if it is of the
form $\ich{\valty}$ or $\lich{\valty}$. We say that a sequence of types
$\bar{\env} = \ty_1, \dots, \ty_n$ (or type environment) is
\emph{negative} iff $\ty_1, \dots, \ty_n$ are all negative. We use
$\barnegenv, \barnegenvTwo$ to range over sequences of negative
types, and $\negenv, \negenvTwo$ to range over negative type
environments. On morphisms, we set:
\begin{definition}\label{def:negproc}
  A process \( P \colon \! \negenv \to \negenvTwo \) is \emph{negative} if
\(   P \wbc^{\negenv^\perp, \negenvTwo}  Q \) for some \( Q \) such that  \( \emptyset \vDash Q \arr
\act  Q' \) implies \(  \act \) is an input action. %
 \end{definition}
 That is, \( P \) must be triggered first by an input action at a name in \( \negenvTwo \).
 Note that we may choose a \( P \) that is receptive by Corollary~\ref{c:ur}, so any input actions using the name in \( \negenvTwo \) is possible.
 We get:

\begin{theorem}
There is a sub symmetric monoidal category $\ProcN$ of $\Proc$ with
negative type sequences as objects, and negative processes (up to
equivalence) as morphisms.

Additionally, $\ProcN$ admits $\emptyenv$ as a terminal object, also
written $\top$.
\end{theorem}

Any identity process is negative. Moreover, negative
processes compose because if \( P \) and \( Q\) are negative, \( \res
{\seq x} (P \mid Q)\) cannot cause any interaction between \( P \) and
\( Q \). Moreover $\ProcN$ inherits from $\Proc$ its
symmetric (strict) monoidal structure. And $\emptyenv$ is
terminal: indeed, if $Q \colon \negenv \to \emptyenv$ satisfies the
condition of Definition \ref{def:negproc}, then it can only perform
negative actions, but none are available.
So it must be barbed congruent to $\nil$, as required.

However, $\ProcN$ still lacks the desired structure; it is not clear
how the arrow of two negative types, let alone the exponential of an
arbitrary negative type, could be defined in all generality within the
typing system of our extended calculus. Our solution is to restrict
these connectives to certain types dubbed \emph{strict}. This is not an
ad-hoc solution: we will show that we thus obtain a \emph{relative Seely
category} \cite{Clairambault24}. This is a weakening of Seely
categories matching the structure available in concurrent games, and yet
 sufficient to ensure that the Kleisli category is cartesian closed.

\subsection{The Strict Subcategory and Relative Closure}

In a relative Seely category, not all binary products are required:
we only ask for those between a subset of objects dubbed
\emph{strict}.

A type is \emph{strict} if it has the form $\lich \valty$ -- a type
sequence $\barnegenv$ is \emph{strict} if it is empty; or consists in a single
strict type. An important conceptual property of negative processes on strict
types is that they must be input processes (up-to barbed congruence):
\begin{lemma}
  \label{lem:negative-strict-proc}
  Let \( P \colon \negenv \to a: \lich{\valty} \) be a negative process.
  Then \( P \wbc \inp a x. P_0 \) for some process \( P_0 \).
\end{lemma}
\begin{proof}
  Since \( P \) is a negative process whose only free input name is \(
  a \), the only action \( P \) can do is an input at \( a \).
  Thanks to the type system, there is a unique  \( P' \) such that \( P \arr{\inp a x} P' \), for some \( x \).
  Hence, we have \( P \wbc \inp a x. P'\)
\end{proof}

Lemma~\ref{lem:negative-strict-proc}  holds thanks to the absence of races caused by multiple receivers.
Hence, in \api, the same reult holds even when \( a \) has a type \( \ich{\valty}\).
This is not the case in \Api{}, because negativity and the restriction on the type environment cannot rule out processes such as \( \inp a x.P \mid \inp a x . Q \).

Strict environments and negative processes over them form a
subcategory of \( \ProcN \). We write this subcategory as \( \ProcS \)
and the inclusion from \( \ProcS \) to \( \ProcN \) as \( J \).

\subparagraph*{Cartesian products.} $\ProcS$ inherits from $\ProcN$ its
terminal object $\emptyenv$, which is indeed strict. We define the
binary product, written \( \with \), between objects by $\emptyenv \with \barnegenv,
\barnegenv \with \emptyenv \defeq \barnegenv$ if one of the objects is
empty, and
\( \lich \valty \with  \lich \valtyTwo \defeq \lich{\valty +
\valtyTwo} \) otherwise. We only detail the pairing and projections for
a product of non-empty objects: the pairing $\langle P, Q \rangle \colon \barnegenv \to c : \lich{\valty + \valtyTwo}$
of \( P \colon \negenv \to a : \lich{\valty} \) and \( Q \colon \negenv
\to b : \lich{\valtyTwo} \) is defined as
\begin{equation*}
  \inp c x . \binCase x y {\resBP a(P \mid \out {\comp a}{y})} y {\resBP b(Q \mid \out {\comp b}{y})}.
\end{equation*}
The projection \( \pi_i \colon (b : \loch{\valty_1 + \valty_2}) \to (a
: \lich{\valty_i}) \) is \( \inp a x . \out b {\coprodTag_i(x)}\).
These indeed define a cartesian
product (See Appendix~\ref{app:category}).

\subparagraph*{Relative closure.} Likewise, in a relative Seely category
the linear arrow $A \multimap B$ is only required for $B$ strict. Given
 a negative environment
\( \barnegenvTwo = \ty_1, \ldots, \ty_n \), we set
$\barnegenvTwo \multimap \emptyenv \defeq \emptyenv$ along with
\(
\barnegenvTwo \multimap \lich{\valtyTwo} \defeq \lich{\ty_1^\perp \times
\ldots  \times \ty_n^\perp \times \valtyTwo}\).

The readers familiar with game semantics may recognize here the usual
arrow arena construction: $\barnegenvTwo$ is dualised, and set underneath
the root of $\lich{\valtyTwo}$. This extends to a \emph{relative} closed
structure, \emph{i.e.}:

\begin{restatable}{lemma}{closedStructure}
  \label{lem:closed-struture}
  There is a \( J \)-relative coadjunction~\cite{ArkorMcDermott2024,AltenkirchCU15}: \( - \otimes \negenvTwo
  \mathrel{{}_J\!\dashv} J(\negenvTwo \multimap -) \).
\end{restatable}
\begin{proof}[Proof Sketch]
It is sufficient to prove that for all $\barnegenv, \barnegenvTwo$ with
the latter strict, there is an object $\barnegenv \multimap
\barnegenvTwo$ and the \emph{evaluation morphism},
\( \evm_{\barnegenv, \barnegenvTwo} : (\barnegenv \multimap \barnegenvTwo) \otimes \barnegenv \to \barnegenvTwo \)
satisfying the usual universal property of an arrow object. We present
the data only for $\barnegenv \multimap \lich{\valty}$, the case of the
empty context being trivial.

For \( P \colon \negenv \otimes \negenvTwo \to b : \lich{T} \), its \emph{currying}
is
 \( \Lambda(P) \defeq c(\seq a, x).\resBP b (P \mid \out {\comp b} x) \).
This is invertible up to barbed congruence, with
\( \Lambda^{-1}(Q) \defeq \inp b x. \resBP c (\out {\comp c}{\seq a, x} \mid Q) \)
  for \( Q \colon \negenv \to c \colon \negenvTwo \multimap \lich{\valty}
\). As usual, the evaluation can be obtained as
$\evm_{\barnegenvTwo, \lich{\valty}} = \Lambda^{-1}(\id_{\barnegenvTwo
\multimap \lich{\valty}})$; this satisfies the equations required for
the universal property.
\end{proof}

\subsection{The Exponential Modality}

The next component is the
\emph{exponential modality}. The main difference with respect to
standard models of linear logic is that the bang is only defined on
strict objects, and sends a strict object to a non-strict one,
\emph{i.e.} $! \colon \Proc^{-}_s  \to \Proc^-$. Not being an
endofunctor it cannot be a comonad, but only a \emph{relative comonad}
in the sense of~\cite{AltenkirchCU15}, subject to additional conditions
adapted from Seely categories.

\subparagraph*{A relative comonad.} On strict objects, we set $! \emptyenv \defeq
\emptyenv$ and $! \lich{\valty} \defeq \ich{\valty}$, \emph{i.e.} we
simply transform a linear input into a non-linear input.

We do not need to define the functorial action of $!$, which
will be derived from the other components. The \emph{dereliction}
process
\(
\der_{\lich{\valty}} \defeq \inp a x. \out b x \colon (b :
\ich{\valty}) \to (a : \lich{\valty})
\)
is  a wire between types that differ only in linearity, and 
$\der_\emptyenv$ is %
 the terminal morphism. The \emph{promotion},
taking a process $P \colon (a : \ich{\valtyTwo}) \to (b :
\lich{\valty})$ to $\prom P \colon (a : \ich{\valtyTwo}) \to (c :
\ich{\valty})$ is defined as
\( \prom P \defeq !\inp c x. \resBP b(P \mid \out {\comp b}{x} ) \).

Since \( P \cong \inp b x . P_0 \) by
Lemma~\ref{lem:negative-strict-proc}, promotion can be seen as an
operator that turns a non-replicated input into a replicated one.

\begin{restatable}{lemma}{promotion}
  \label{lem:promotion}
  The following laws hold, i.e.~\( ! \) is a \( J \)-relative comonad.
  \begin{enumerate}
    \item \( \prom P ; \der = P \) \label{it:lem:promotion:left-unit}
    \item \( \prom{\der_{\lich{\valty}}} = \id_{\ich{\valty}}\) \label{it:lem:promotion:right-unit}
    \item \( \prom{(\prom P; Q)} = \prom P ; \prom Q\) \label{it:lem:promotion:assoc}
  \end{enumerate}
\end{restatable}
\begin{proof}
  As a showcase we
  show~\ref{it:lem:promotion:left-unit} (the proofs for the other laws are in Appendix~\ref{app:category}).
  Without loss of generality, we can assume that \( P = \inp b x . P_0 \) by Lemma~\ref{lem:negative-strict-proc}.
  In this case, \( \prom P \wbc ! \inp b x . P_0 \).
  Therefore,  by laws for internal communication and GC, we have
  \begin{eqnarray*}
    \prom P ; \der &=& \resBP b (!\inp b x .P_0 \mid \inp a x . \out {\comp b} x) \\
    &\sbisim& \inp a x . \resBP b (!\inp b x.P_0 \mid \out{\comp b} x) \\
    &\contr& \inp a x .P_0
  \end{eqnarray*}
  In turn, this is exactly
$P \sub a b$, concluding the proof.
\end{proof}

The proof is notably simple, in particular comparing it with corresponding
proofs in game semantics (see \emph{e.g.} \cite[Proposition
8.3.7]{Clairambault24}). We see this as a strong argument in favour of
our syntax, which encapsulates the intricate combinatorial aspects of
game semantics, emphasizing the logical reasons why the laws hold.

\subparagraph*{An exponential modality.}
We now show that \( ! \colon \ProcS \to \ProcN \) is a symmetric strong monoidal functor.
Note that Lemma~\ref{lem:promotion} makes \( ! \) a functor by \( !P \defeq \prom {(\der ;P)} \).
We define the mediating maps: the morphism \( m_0  \colon 1  \to !\top
\) is defined as $\id_\emptyenv$,
and the \emph{Seely isomorphism} \( m_{\lich{\valty}, \lich{\valtyTwo}}
\colon  (a : \ich \valty, b : \ich \valtyTwo) \to (c : \ich{(\valty +
\valtyTwo)}) \) is defined as
\begin{align*}
  &!\inp c x. \binCase x y {\out a y} z {\out b z}
\end{align*}
and trivial if one of the operands is the empty context.
It is easy to check that the mediating maps satisfy the coherence laws.
Each \(  m_{\lich{\valty}, \lich{\valtyTwo}}  \) is an isomorphism because its inverse is given as the process \( !\inp a y . \out c {\inl y} \mid !\inp b z . \out c {\inr z} \).

The mediating maps, moreover, are compatible with promotion as in Seely categories.
This concludes to show that the \( \ProcN\) together with its
subcategory \( \ProcS \) forms a relative Seely category:
\begin{restatable}{lemma}{promotionMediating}
  The following diagram commutes.
  \[
    \begin{tikzcd}[row sep = 11pt]
      !\lich{\valtyThree} \ar[d, "\prom {\tuple{\der, \der}}"] \ar[rr, "\prom {\tuple{P, Q}}"]& & !(\lich{\valtyTwo} \& \lich{\valty}) \ar[d, "m^{-1}"]\\
      !(\lich{\valtyThree} \ar[rd,"m^{-1}"]\& \lich{\valtyThree}) & & !\lich{\valtyTwo} \otimes !\lich{\valty} \\
      & ! \lich{\valtyThree} \otimes !\lich{\valtyThree} \ar[ru,"\prom P \otimes \prom Q"] &
    \end{tikzcd}
\]

  \end{restatable}

As is the case for the standard Seely categories, we can apply the Kleisli construction to a relative Seely category to obtain a cartesian closed category.
We write \( \ProcN_{!} \) for the Kleisli category of the (relative) comonad \( ! \colon \ProcS \to \ProcN \).
Its objects are those of \( \ProcS \) and homsets are defined by \( \ProcN_{!}(\lich{\valtyTwo}, \lich{\valty}) \defeq  \ProcN(\ich{\valtyTwo}, \lich{\valty})\).
The identities are derelictions and compositions are defined using promotion.
The finite product in \( \ProcN_{!} \) are defined as in \( \ProcS \)
and the arrow is defined by \( \lich{\valtyTwo} \Rightarrow
\lich{\valty} \defeq !\lich{\valtyTwo} \multimap \lich{\valty} =
\lich{\och{\valtyTwo} \times \valty}\); along with $\top \Rightarrow
\barnegenv = \barnegenv$ and $\barnegenv \Rightarrow \top = \top$.

\begin{theorem}
 $\ProcN$ is a relative Seely category. Consequently, the Kleisli
category $\ProcN_!$ is cartesian closed.
\end{theorem}

\subsection{Sequoidal Structure}
\label{sec:sequoid}
From processes we can also construct a sequoidal category~\cite{Laird02}, a categorical structure used to express state (and sequential behaviour) in game semantic models.
The construction follows that for HO games~\cite{Laird19}.
Together with the encoding in Section~\ref{s:expr}, this supports that \api{} is suited to modelling imperative features.

The \emph{sequoid} operator \( \oslash \) is a non-commutative operator over strict types and negative type sequences defined as
\begin{align*}
  \lich{\valtyTwo} \oslash (\ty_1, \ldots, \ty_n) \defeq 
	\lich{\valtyTwo \times \ty_1 \times \ldots \times \ty_n}.
\end{align*}

This is essentially the same as \( (\ty_1, \ldots, \ty_n) \multimap \lich{\tyTwo} \), except that \( \ty_i \) is not dualised.
Since \( \ty_i \) is a negative type, \emph{this construction exploits
the ability of \api{} to exchange input names.}\footnote{In game semantics
terms, the corresponding arenas are not alternating.}

We first recall the general definition of a sequoidal category, and then construct such a structure using processes.
\begin{definition}[\cite{Laird19}]
  A \emph{sequoidal category} is a tuple \( (\category C, \category L, \oslash, J, \omega) \) where \( (C, \otimes, I)\) is a symmetric monoidal category, \( (\category L, \oslash) \) is a \(\category C\)-action and \( (J, \omega) \) is a weak morphism of \( \category C\)-actions from \( (\category L, \oslash ) \) to \( (\category C, \otimes) \).
  Here a \emph{weak morphism} from \( (\category L, \oslash) \) to \(
  (\category C, \otimes) \) is a functor \( J \colon \category L \rightarrow
  \category C \) with a natural transformation \( \omega_{L, C} \colon
  J(L) \otimes C \to J(L \oslash C) \) that satisfies some coherence
  condition (for which we refer to~\cite{Laird19}).
\end{definition}

 To define a category of processes playing the role of \( \category L \), we consider a subcategory of \( \ProcS \) following~\cite{Laird19}.
\begin{definition}
  A negative process \( P \colon b: \lich{\valtyTwo} \to a : \lich{\valty}\) is \emph{linear} if, for any \( P \arr{\inp a x} P' \) and for all \( v \), we have \( P'\sub v x  \contr P'' \mid \out b w \) for some \( P'' \) and \( w \).
  We also consider any negative process \( P \colon b: \lich{\valtyTwo} \to \emptyenv \) or \( P \colon \emptyenv  \to a : \lich{\valty} \) to be linear.
\end{definition}
Intuitively, a linear process is a process that can always perform an output action using the unique free output name once the initial input guard is executed.
Since parallel composition plus hiding preserves linearity and wires are linear, we have a subcategory of \( \ProcS \) of linear processes, which we write as \( \ProcSL \).

\begin{theorem}
  There is a sequoidal category \( (\ProcN, \ProcSL, \oslash, J, \omega) \).
\end{theorem}
\begin{proof}[Proof Sketch]
  The functorial action of \( \oslash \) is defined by
  \begin{align*}
    P \oslash \negenvTwo \defeq \inp {c'} {x, \seq a} . \resBP b (\resBP c (P \mid \out {\comp c}{x}) \mid \inp {\comp b}{y} . \out {b'}{y, \seq a})
  \end{align*}
  for \( P \colon b : \lich{\valty} \to c : \lich{\valty} \) and \( \negenvTwo = a_1 : \ty_1, \ldots, a_n : \ty_n \).
  Note that this process is linear because once the input at \( c' \) happens, \( P \) will perform a communication using \( c \), which in turn allows an output at \( b \) that triggers the output at \( b' \).

Additionally, the family of processes \( \omega_{\lich{\valty}, (\ty_1, \ldots,
\ty_n) } \colon J(\lich{\valty}) \otimes (\ty_1, \ldots,\ty_n)) \to
J(\lich{\valty} \oslash (\ty_1, \ldots, \ty_n))\) given by
  \begin{align*}
    \inp a{x, \seq d} . (\out b x \mid  \links  {\seq d} {\seq c} ) \colon b : \lich{\valty}, \seq{c} : \seq{\ty} \to a : \lich{\valty \times \ty_1 \times \cdots \times \ty_n}.
  \end{align*}
forms a natural transformation as required.
\end{proof}
The difference between \( \inp a x. P \otimes \inp b y . Q = \inp a x.
P  \mid \inp b y . Q\) and \( (\inp a x. P \otimes \inp b y . Q);
\omega = \inp a {x, b} . (P \mid \inp b y . Q)\) highlights the behaviour of the sequoid.
In the former case, the environment can directly access \( \inp b y . Q \) whereas in
the latter it can be executed only after the input at
\( a \) (which also receives the name \( b \)) is triggered, thereby making the resource at \( b \) internal to the `object' at \( a \).
This object-oriented view of the sequoid also helps explain why sequoid is related to references: a reference may be modelled as an object with read and write methods.

We can further construct a richer structure called \emph{sequoidal CCC} from processes by considering the Kleisli category.
Roughly, a sequoidal CCC is a sequoidal category whose monoidal category is a cartesian closed category in which \( \oslash \) and \( \Rightarrow \) interact
(see~Appendix~\ref{app:sequoidal} or~\cite{Laird19} for the formal definition).
\begin{restatable}{theorem}{sequoidalCCC}
  The  tuple \( (\ProcN_{!}, \ProcSL, \oslash_{!},
  J_{!}) \) is a sequoidal CCC.

  Here \( \oslash_{!} \colon \ProcSL \times \ProcN_{!} \to \ProcSL \) is defined by \( \lich{\valtyTwo} \oslash_{!} \lich{\valty} \defeq  \lich{\valtyTwo} \oslash !\lich{\valty} \) and \( \lich{\valtyTwo} \oslash_{!} P \defeq \lich{\valtyTwo} \oslash \prom P \), \( J_{!} \) is an identity-on-object functor such that \( J_{!}(P) \defeq \der; P \).
\end{restatable}

\section{Interpretation of STLC}
\label{sec:encoding-lambda}
\newcommand*{\lamty}{\sigma}
\newcommand*{\lamtyTwo}{\tau}
\newcommand*{\lamterm}{t}
\newcommand*{\lamtermTwo}{s}
\newcommand*{\app}[2]{#1 \; #2}
\newcommand*{\ltopi}[2]{\llparenthesis #1 \rrparenthesis_{#2}}

We detail the interpretation of the simply-typed \(
\lambda \)-calculus in the cartesian closed category \( \ProcN_{!} \).
This induces a syntactic translation from \( \lambda \) to $\pi$
that validates \( \beta\eta \)-equivalence by construction.
We then compare this encoding with the encodings in the literature such as the seminal paper by Milner~\cite{Milner92} and the work by Hyland and Ong~\cite{HylandOng95} that gives an encoding based on game semantics.
As expected, the encoding we derive is essentially the one given by Hyland and Ong.
This provides evidence that the construction in Section~\ref{sec:category} faithfully follows that of game semantics.

\subparagraph*{The simply-typed $\lambda$-calculus.}
We fix the grammars
\begin{align*}
  \lamterm &\Coloneqq x \mid \lambda x. \lamterm \mid
\app \lamtermTwo \lamterm \qquad  \lamty \Coloneqq o \mid \lamtyTwo \to \lamty
\end{align*}
for types and terms, where $o$ is a fixed base type. We consider the
terms equipped with the standard typing rules, that we shall not repeat
here. Note that every type \( \lamty \) has a unique representation of
the form \( \lamty_1 \to \cdots \to \lamty_n \to  o \) where \( n \ge 0
\). 

\subparagraph*{Interpretation in $\ProcN_{!}$.}
The simply-typed $\lambda$-calculus has a sound interpretation in any
cartesian closed category, following a well-established recipe
\cite{lambekscott}. The only degree of liberty is
the image set for the base type~--- here we choose the minimal
$\ltopi{o}{} \defeq \lich{\unittype}$ (which corresponds, in game
semantics terms, to an arena with a single move). 
From there, the definition of the interpretation follows purely
mechanically. The interpretation of the arrow type is $\ltopi{\lamtyTwo
\to \lamty}{} \defeq !\lamtyTwo \multimap \lamty$, altogether giving an
interpretation of every type of the simply-typed $\lambda$-calculus as
a strict type of our $\pi$-calculus. 

For iterated arrows, this amounts to:
\begin{align*}
  \ltopi{\lamty_1 \to \cdots \to \lamty_n \to o}{} = \lich{!
{\ltopi {\lamty_1}{}}^{\bot} \times \cdots \times !{\ltopi
{\lamty_n}{}}^{\bot} \times \unittype}
\end{align*}
which is essentially how types are translated in the
translation from PCF to the \( \pi \)-calculus by Hyland and
Ong~\cite{HylandOng95}. We note that this is different from variants of Milner's encoding that has been commonly used in the literature~\cite{Milner92,SangiorgiWalker01}.
Typically, arguments of a function are received one by one even if the target of the translation is the polyadic \( \pi \)-calculus.\footnote{The original encoding by Milner used the monadic \( \pi \)-calculus.}

The interpretation
of a simply-typed context $\Gamma : x_1 : \lamty_1, \ldots, x_n :
\lamty_n$ is the iterated product $\ltopi{\Gamma}{} =
\ltopi{\lamty_1}{} \with \cdots \with \ltopi{\lamty_n}{}$.
For each $\Gamma \vdash \lamterm : \lamtyTwo$, the interpretation of terms in
the Kleisli category $\ProcN_{!}$ yields \( \intr{\lamterm}{} \in \ProcN_{!}(\ltopi{\Gamma}{}, \ltopi{\lamtyTwo}{}) \)
\emph{i.e.} a negative process $x : \ich{(\valty_1 + \cdots +
\valty_n)}, p : \ltopi{\lamtyTwo}{} \vdash \ltopi{\lamterm}{}$ where $\Gamma =
x_1 : \lamty_1, \ldots, x_n : \lamty_n$ with $\ltopi{\lamty_i}{} =
\lich{\valty_i}$. This process $\ltopi{\lamterm}{}$ is computed inductively on
$\lamterm$, in a way that we may wish to regard as a
syntactic translation, which we set to describe.

\subparagraph*{The interpretation as an encoding.} Because of the way
contexts are interpreted in a cartesian closed category, the induced
encoding uses the cartesian product, and hence sum types. We prefer to
avoid this, because $\pi$-calculi often do not have sum types; and we
wish to relate to previous encodings that do not involve sums. So we
present a syntactic encoding instead sending $\Gamma \vdash \lamterm : \lamtyTwo$ to
\[
\ltopi{\lamterm}{} \in \ProcN(!\ltopi{\lamty_1}{} \otimes \cdots \otimes
!\ltopi{\lamty_1}{}, \ltopi{\lamtyTwo}{})\,,
\]
which is related to $\intr{\lamterm}$ via the canonical isomorphism 
\[
s_\Gamma : ! \ltopi{\lamty_1}{} \otimes \cdots \otimes ! \ltopi{\lamty_n}{} \cong
! (\ltopi{\lamty_1}{} \with \cdots \with \ltopi{\lamty_n}{})\,,
\]
obtained via the Seely isomorphism 
$!\tyTwo \otimes !\ty \cong !(\tyTwo \& \ty)$, one of the components of
a relative Seely category.

An encoding
$\ltopi{\lamterm}{} \in \ProcN(!\ltopi{\lamty_1}{} \otimes \cdots \otimes
!\ltopi{\lamty_1}{}, \ltopi{\lamtyTwo}{})$ 
has free output  names
corresponding to \( x_i \), and one  free linear input name; the latter
is often called a \emph{continuation name}.
In the translation, we override \( x_i \) and also use them for the output names whereas the continuation name \( p \) is passed as a parameter of the encoding of terms, which we write as \( \ltopi{-} p \).
The encoding is given as below.
\begin{align*}
  \ltopi x p &\defeq \inp p {\seq y} . \out x {\seq y} \\
  \ltopi{\lambda x_1 . \lamterm} p &\defeq \inp p {x_1, \ldots, x_n, q} . \resBP r(\out {\comp r} {x_2, \ldots, x_n, q} \mid \ltopi \lamterm r) \\
  \ltopi{\app \lamtermTwo \lamterm} p &\defeq \inp p {x_2, \ldots, x_n, q} . \resBP r \resB{x_1}{y}(\out{\comp r}{x_1, \ldots, x_n, q} \mid \ltopi \lamtermTwo r \mid !{\ltopi \lamterm y})
\end{align*}

The encoding is derived from the interpretation (modulo the choice of
representatives for equivalence classes).
For example, the encoding of \( \lambda \)-abstraction simply uses the
bijection of Lemma~\ref{lem:closed-struture}.\footnote{Here variables
of a tuple type are expanded so as to make the encoding similar to that of
Hyland and Ong, which uses the polyadic \( \pi \)-calculus.} 

This encoding follows the categorical interpretation, so the following
theorem is a routine verification by induction on $\lamterm$:

\begin{theorem}
For any $\Gamma \vdash \lamterm : \lamty$, $\ltopi{\lamterm}{} \wbc \intr{\lamterm}
\circ s_\Gamma$.
\end{theorem}

By the soundness theorem for cartesian closed categories, it
automatically follows that our encoding is sound for
$\beta\eta$-equivalence. 

\subparagraph*{Observations.}
Again, this is essentially the translation given by Hyland and Ong~\cite{HylandOng95}.
One difference is that our translation uses free outputs whereas they used bound outputs with dynamic wires, i.e.~\( \bout a c : \transf{{\links {\comp c} b}} \).
Since the two are barbed congruent, this is just a difference in the choice of the representative.
Free outputs give a more economical representation whereas only using bound outputs gives a tighter connection with game semantics in the sense that each output action corresponds to a move of a strategy.
Note that the equivalence \( \out a b \cong \bout a c : \transf{{\links {\comp c} b}} \) was not known (even for O-names) when~\cite{HylandOng95} was published.
Another difference lies in the way \( \beta\eta \)-equivalence is validated.
The validity of Hyland and Ong's encoding is proved via a non-trivial correspondence between processes and strategies and the fact that the category of games forms a CCC.\footnote{This is partly because the motivation of~\cite{HylandOng95} was to show this correspondence.}
Directly checking the \( \beta\eta \)-laws at the level of processes (i.e.~constructing a CCC as in Section~\ref{sec:category}) is simpler.

While we use the \( \lambda \)-calculus as an illustrative example, we note that the categorical interpretation is a general methodology for deriving process encodings of various (higher-order) programming languages (non-deterministic, concurrent, mutable state, etc.).
The categorical structure and its compositionality can be used to
derive a correct by construction encoding of standard data types (e.g.
arrow type, product type). This allows us to focus on the encoding of
additional constructs, whose correctness can be checked separately.
This is in stark  contrast with other process encodings that are usually designed by operational means.

\section{Related Work}

A number of offspring  of the $\pi$-calculus  require that the
input end of a channel should be  permanently owned by a single
process.
Sometimes this feature is enforced by the type systems, as in the case of
\iflong
\emph{linearly} or
\fi
 \emph{uniformly  receptive} names
\cite{Sangiorgi97}.
\iflong
,
where inputs are either linear or replicated and hence behave
 as `function' names.
\fi
In the Join calculus~\cite{FournetGonthier96}, the feature
is syntactic, as  all inputs  are unique and  replicated; the need for modelling state is
achieved by the introduction of a `join' construct.
\iflong
 that combines two or more inputs. 
\fi
In the $\pi_1$-calculus \cite{Amadio97,AmadioBL03}, stateful processes
are modelled using recursion, in place of replication and join, to allow nesting of
inputs. Similar approaches are followed 
in other $\pi$-calculus  languages addressing mobility in distributed systems, see e.g.,  \cite{Castellani01}.
In these calculi, the owner of an input name is persistent since
input capabilities may not be communicated.
\iflong
 only 
output names may
be exchanged. 
\fi This property is sometimes called
\emph{locality}, and is  studied in the Asynchronous Localised $\pi$-calculus
\cite{MeSa98}.  Our proof technique based on a translation into a
subset of internal processes is inspired by  proof techniques in \cite{MeSa98}.

Another class of \( \pi \)-calculi related to \api{} is the process calculi developed through the Curry–Howard correspondence with linear logic~\cite{BellinScott94,CairesPT16,Wadler14}.
In these calculi, typically an input name is either used linearly or
has a unique replicated occurrence.
Hence these calculi satisfy the 1-input property.
In these linear-logic-inspired calculi, `wire substitution law' and the commutativity of prefixes have also been considered.
In particular, wires are often defined as a primitive construct with the substitution law serving as their defining property (e.g.~in~\cite{Wadler14,KokkeMP19,YoshidaCS20}).
However, the validity of these laws is often supported only by informal arguments based on proof-theoretic adequacy, lacking a formal justification from a behavioural perspective.
A notable exception is~\cite{PerezCPT14}, in which the behavioural validity of commutativity laws for \( \pi\text{DILL} \) is proved using typed context bisimilarity.
 One reason the behavioural theory of calculi such as Wadler's CP~\cite{Wadler14} remains underexplored is the difficulty of defining a standard LTS, due to discrepancies between the syntax of CP and the standard \( \pi \)-calculus.
To remedy this issue, Kokke et al.~\cite{KokkeMP19} proposed a calculus inspired by the theory of hyper sequents.
While their calculus has a syntax close to that of the ordinary \( \pi \)-calculus and is equipped with an LTS, input and output actions are delayed, unlike in \api{}.

The behavioural equivalence we have used in \api is as flexible as possible in the
treatment of input and output ends of channels: closing contexts may connect names in an
arbitrary manner.  This is reminiscent of \emph{open bisimulation} for $\pi$-calculi
\cite{Sangiori93}.  In both cases, behavioural equivalence must take into account the
possibility that syntactically different names may give rise to interactions, depending on
the context in which they are used.  This treatment of name identities yields sharp
substitutivity and congruence properties, and has in turn analogies with symbolic
bisimulation (e.g., \cite{BorealeN96}). Indeed, open bisimilarity has been advocated in
practical implementations of bisimulation checkers for $\pi$-calculi and in
representations in logical frameworks; see e.g.,
\cite{VictorM94,BaeldeGMNT07,TiuM10,TiuD10,LiuTH25}.  A further analogy between
bisimilarity in \api and open bisimilarity in $\pi$-calculi appears in the dynamic
creation of permanent constraints on names during the bisimulation game: in the case of
\api, the constraints represent connections between input and output names; in the case of
open bisimilarity the constraints represent distinctions among names.  While seemingly
dual in the two cases, the constraints are handled similarly in both definitions and proofs.

The idea of using \( \pi \)-calculus processes as syntactic representations of strategies of game semantics dates back to the early days of game semantics~\cite{HylandOng95,FioreHonda98,HondaYoshida99}.
In the study of process representations of strategies, the redundancy of copycat strategies (i.e.~dynamic wires) was also pointed out by Fiore and Honda~\cite{FioreHonda98}.
To address this issue, they introduced the notion of optimised
strategies to represent the free outputs of \( \pi
\)-calculus, which can be understood as an application of
Lemma~\ref{l:wire}.
In a similar vein, the connection between operational game semantics and Internal \( \pi \)-calculus has been studied~\cite{JaberSangiorgi22}.
Since these works studied sequential languages, only restricted (essentially sequential) fragments of \( \pi \)-calculus were employed.
The relationship between concurrent game semantics and the \( \pi \)-calculus is explored in~\cite{YoshidaCS20} via a session-typed variant of the \( \pi \)-calculus. In that calculus, a typing rule inspired by differential linear logic is introduced, which is similar to the In rule of \api{}.

The processes-as-strategies paradigm has also enabled the transfer of game semantic conditions~--- such as sequentiality, well-bracketing, visibility, and innocence~--- to the \( \pi \)-calculus by means of type systems~\cite{BergerHY01,HirschkoffPS21,HirschkoffQS25}.
This line of work is orthogonal to the 1-input property, and the two approaches may be combined.

On the other hand, game models of \( \pi \)-calculi have also been investigated.
Laird~\cite{Laird05} introduced a game model for the Asynchronous \( \pi \)-calculus (with I/O separation), which was later extended to a truly concurrent model~\cite{SakayoriTsukada17}.
Since these game models require copycat strategies to behave as identities, they are not sound with respect to barbed congruence: input wire law does not hold for barbed congruence in ordinary \Api{}.
By contrast, by considering a session typed \( \pi \)-calculus rather than \Api{}, a causal game model that is fully abstract for barbed congruence is introduced in~\cite{CastellanYoshida19}.

Beyond game semantic models, we should also mention other denotational models of the \( \pi\)-calculus, including those based on functor categories~\cite{Stark96,FioreMS02,Hennessy02,CattaniSW97}, sheaves~\cite{EberhartHS15} and Geometry of Interaction~\cite{JagadeesanJagadeesan95}.
These works differ in focus from ours: they aim to interpret the \( \pi \)-calculus in specific categories, whereas our goal is to organize processes themselves into categories.
On a more technical note, the models based on functor categories place emphasis on the nominal aspects of the \( \pi \)-calculus.
By contrast, our results are achieved by excluding such nominal aspects, by forcing communications to happen along names that have been explicitly connected when created.

\bibliography{main}

\appendix
\section{Additional material for Section~\ref{s:rs} }
\label{a:b}

\begin{definition}[Structural congruence]
\label{d:sc}
\emph{Structural congruence} is the smallest congruence on  processes
induced by the following rules: 
	\begin{mathpar}
		P|\nil\equiv P
		\and
		P|Q\equiv Q|P
		\and
		P|(Q|R)\equiv (P|Q)|R
\\
		P|\resBPP{a}b Q \equiv \resBPP{a}b ({P|Q}) \text{ if }a,\comp a \notin\fn{P}
		\and
\resBPP{a}b\resBPP{c}d		{P}\equiv \resBPP{c}d \resBPP a{b}P
		\and
		\resBPP{a}b \nil\equiv\nil
	\end{mathpar}
      \end{definition}
       There is no unfolding of replication, as usual in $\pi$-calculi, as this would violate the
1-input property.  We can however derive an unfolding law for
replication, as by Lemma~\ref{l:otherlaws}(\ref{law:unf}).

\begin{definition}[Reduction]
\label{d:rr}
The \emph{reduction relation},  $P \longrightarrow P' $, is  defined on processes by the
following rules: 
  \begin{mathpar}
  \inferrule[Red]{}{\resB a b (\inp a x . P | \out b v | Q ) \longrightarrow
\resB a b (  P \sub v x  | Q )
}
\\
  \inferrule[RedRep]{}{\resB a b (!\inp a x . P | \out b v | Q ) \longrightarrow
\resB a b (!\inp a x . P |  P \sub v x  | Q )
}
\\
  \inferrule[Par]{P \longrightarrow P'}{
P | Q \longrightarrow P'|Q}
 \qquad
  \inferrule[Res]{P \longrightarrow P'}{
\resBP a P  \longrightarrow \resBP a P'}
 \qquad
  \inferrule[Equiv]{P \equiv P' \andalso P' \longrightarrow P'' \andalso P'' \equiv P'''}{
P  \longrightarrow P'''}
\end{mathpar} 
\end{definition}

\begin{theorem}[Subject Reduction]
\label{t:sr}
If $\ok \Gamma P$ and $P \longrightarrow P' $ then also 
$\ok \Gamma {P'}$.
\end{theorem}

\section{Additional material for Section~\ref{s:ls}}
\label{a:ls}

\subsection{Background on \texorpdfstring{\( \pi \)}{pi}-calculi}
\label{aa:backpi}

We recall the syntax of the  (simply-typed)  Asynchronous $\pi$-calculus, \Api:

\begin{definition}[\Api]
\label{d:Api}
\[ \begin{array}{rcl}
P & \Coloneqq &  \nil \midd  P | Q\midd   \inp a x .{P}
\midd \out bv \midd %
\res a   P
\midd ! \inp a x .P \\
 v & \Coloneqq & a \midd \unitvalue
   \end{array}
  \]
\end{definition} 
Below is  the standard LTS for the calculus (in the \emph{ground}
style); as usual, $\bout ac $ stands for $\res c \out a c$.

\begin{center}
\begin{tabular}{rlcrl}
{\trans{ inp}}:& $ \inp  a b .     P \arr{\inp a b }P$ & 
&
{\trans{ out}}:& $ \out a b     \arr{\out a b}\nil$ \\% [\mysp]
{\trans{ rep}}:& $ ! \inp  a b .     P \arr{\inp a b }P | ! \inp  a b .     P$ & 
&
{\trans{ par}}:& $\displaystyle{   P \arr\mu   P' \over   P \mid Q   \arr\mu
P' \mid Q } $ if $\bn \mu \cap \fn Q = \emptyset $   \\ %
{\trans{ com}}:& $\displaystyle{ 
P \arr{\inp a b }P' \andalso Q \arr{\out a c} Q' \over
P | Q \arr\tau P' \sub cb | Q'} 
$& 
&
{\trans{ close}}:& $\displaystyle{ 
P \arr{\inp a b }P' \andalso Q \arr{\bout a c} Q'\over
P | Q \arr\tau \res c (P' \sub cb | Q')} 
$
  if $c \cap \fn {P} = \emptyset $
\\
{\trans{ res}}:& $\displaystyle{ P \arr{\mu}P' \over
 \res a     P   \arr{ \mu} \res a P'  } $ $ a \not \in \n\mu$
& %
&
{\trans{ open}}:& $\displaystyle{ P \arr{\out ab }P' \over
 \res b     P   \arr { \bout a b  }  P'  } $ 
\end{tabular} 
\end{center} 

We recall the ordinary definitions of strong and weak (synchronous) bisimilarity,
and of the expansion relation,  on an arbitrary LTS; thus in the remainder
of this section all definitions and results are on a generic LTS, which 
could be  the LTS for
\Api, as well as that for other languages, including  \api. As usual, following
the Barendregt convention, we assume freshness of possible bound names  in actions. 

\begin{definition}[Strong bisimulation]
\label{d:bisimulation}
 A symmetric relation $\RS$ on the processes of an LTS
is a   {\em strong bisimulation} if  $P\RR  Q$ implies that, 
whenever $P \arr\mu P'$, there is $Q'$ such that 
$Q \arr\mu Q'$ and $P' \RR Q'$.

Processes $P$ and $Q$ are  \emph{strongly  bisimilar},
written $P \sbS Q$, if $P  \RR Q$ for some
strong bisimulation $\R$.
\end{definition} 
As usual, %
$ \Longrightarrow $ is the reflexive and transitive closure of $ \arr\tau $, and
$\Arr\mu \defi \Longrightarrow \arr\mu \Longrightarrow $. 
Further, $\Arcap\mu $ is $ \Longrightarrow $ when $\mu=\tau$ and $\Arr\mu$ otherwise. 
 
\emph{Weak bisimulation} and  \emph{weak bisimilarity}, written $\wbS$, are defined as
their strong counterpart in Definition~\ref{d:bisimulation}, except that  the answer
transition in the bisimulation game is allowed to be $Q \Arcap \mu Q'$.
\emph{Weak simulation} and  \emph{weak similarity} are the `one-way'
analogue; that is, in  weak simulation does not have the requirement
of symmetry. We write $\wse$ for (weak)  \emph{simulation
  equivalence}, or \emph{mutual simulation}, which holds on two
processes if  similarity between them holds in both directions.

A useful auxiliary relation in proof techniques is the \emph{expansion relation}, $\expa$.

\begin{definition}[Expansion]
\label{d:expa}
 A relation $\RS$ on the processes of an LTS
 is an {\em  expansion} if, whenever
 $P\RR Q$, %

\begin{enumerate}

\item   $P \arr\mu P'$ implies that there is $Q'$ with $Q \arcap \mu
  Q'$
 and $P' \RR Q'$;

\item 
    $Q \arr\mu Q'$   implies that there is $P'$ with $P \Arr \mu
 P'$ and $P' 
\RR Q'$.

\end{enumerate}
 $P$  {\em expands } $Q$, written
$Q  \expa P$, 
if $P \RR Q$,  for some expansion $\R$. 
\end{definition}

The definition of \emph{weak bisimulation up-to $\wbS$ and $\expa$}
is obtained from that of weak bisimulation by replacing 
 the requirement 
$P' \RR Q'$ (on the derivatives the challenge and response transitions), with 
  $
 P' \wbS \, \RR\, \expa Q'$. 

\begin{theorem}[Soundness of bisimulation up-to $\wbS$ and $\expa$~\cite{Sangiorgi96b}]
\label{t:uptoWE}
If $\R$ is  a weak bisimulation up-to $\wbS$ and $\expa$, then ${\R} \subseteq {\wbS}$.
\end{theorem} 

\subsection{Transporting results from  ordinary \texorpdfstring{\( \pi \)}{pi}-calculi onto \texorpdfstring{\api}{AWpi}}
\label{aa:trans}

\subsubsection{Composite \texorpdfstring{\api}{AWpi} processes}
                             
In order to relate \api to Asynchronous  $\pi$-calculus and transplant results from the
theory of the latter onto the former, we first revisit the LTS for \api 
in Figure~\ref{fig:LTS}.   
In such an LTS, when  $\ltsTv { \delta } P {\mu}  {P'}$,  the connection set under which the behaviour
of $P'$ is to be further examined should remain $\delta$, unless $\mu$ is a bound output, say $\bout b
c$, and $c$ is an O-name; in this case,  the derivative process
is allowed to use both $c$ and its companion, and  the
connection set  should thus become $\delta, \breve c$. 
It is convenient to make  this aspect  explicit,
using an alternative 
 presentation of the LTS for \api on
\emph{composite processes}, that is, pairs  
$\compP P \delta$
of a process and connection set.
Transitions on composite processes are defined via the
following rule:
\[
\infR{\ltsTv \delta P {\mu } {P'}   \andalso 
\delta' = \left\{\begin{array}{cl}
 \delta, \breve c  & 
\mbox{ if $\mu$ is a bound output  $\bout
b c = (\resB {\comp c} c ) \out b c$ (i.e., $c$ is an O-name)} \\
 \delta  & 
\mbox{ otherwise} \\
\end{array} \right.
}{
\ltsC \delta P \mu {\delta'}{P'}
}
\] 
Structural congruence  is extended to composite processes following the definition on
plain processes; however we also allow garbage collection of names in the connection set that do not
occur free in the process: 
\[ 
\infR { P \equiv P' }{\compP P \delta  \equiv \compP {P'} \delta}
\andalso\andalso 
\infR {  }{\compP P{ \delta, \breve a}  \equiv \compP {P} \delta} \;\mbox{$\breve a \cap
 \fn P = \emptyset $}
 \]

On composite processes, weak transitions, 
$\LtsC \delta P \mu {\delta'}{P'}$,
 are defined in the usual manner (i.e., 
via relational composition).

\subsubsection{From composite \texorpdfstring{\api}{AWpi} processes to \texorpdfstring{\Api}{Api}}

There is a straightforward encoding, written $\encoA{\cdot}  $,  from  composite \api processes onto the Asynchronous
$\pi$-calculus (\Api), in which,
in  every restricted subterm
$\resB a b P $,  names 
$a,b$ are 
are  mapped onto the same name, 
 say the O-name $b$ and, similarly, for every pair of names  $a \sepRES b \in \delta $,
any (free) occurrence of   $a$ is replaced by $b$:
\begin{align*}
\encoA{ \compP {\inp a c . P} \delta} & =
\left\{ \begin{array}{ll}
\inp b c . \encoA{ \compP P \delta } & \mbox{if $a\sepRES b  \in \delta $, for some $b$  } \\[\tkpP]
\inp a c . \encoA{ \compP P \delta } & \mbox{otherwise} 
\end{array} \right.
\\
\encoA{ \compP {\out  c a} \delta} & =
\left\{ \begin{array}{ll}
\out c b  & \mbox{if $a\sepRES b  \in \delta $, for some $b$ } \\[\tkpP]
\out c a  & \mbox{otherwise} 
\end{array} \right.   \\[\tkp]
\encoA{ \compP {!\inp a c . P} \delta} & =
\left\{ \begin{array}{ll}
!\inp b c . \encoA{ \compP P \delta } & \mbox{if $a\sepRES b  \in \delta $, for some $b$  } \\[\tkpP]
!\inp a c . \encoA{ \compP P \delta } & \mbox{otherwise} 
\end{array} \right.
\\
\encoA{ \compP {\resB a b P} \delta} & =
\res b \encoA{ \compP { P} {\delta, a \sepRES  b }} \\
\encoA{ \compP { P| Q} \delta}  &=
\encoA{ \compP { P} {\delta} } | 
\encoA{ \compP { Q} {\delta} } \\
\encoA{ \compP { \nil} \delta} &=
\nil
\end{align*}
  In the same manner, the encoding is
extended to (composite)
actions,
by replacing, in 
 an action $\mu$, 
for every $a \sepRES b \in
  \delta $,  
 free occurrences of 
$a$ 
with $b$, and, in a bound output in which a pair of connected
  names $a \sepRES b$ is exported, name $a$ with $b$:
\begin{align*}
  \encoA{ \compP {\inp a c } \delta} & =
  \begin{cases}
\inp b c & \mbox{if $a\sepRES b  \in \delta $, for some $b$  } \\
\inp a c   & \mbox{otherwise} 
 \end{cases} \\
\encoA{ \compP {\out  c a} \delta} & =
\begin{cases}
\out c b  & \mbox{if $a\sepRES b  \in \delta $, for some $b$ } \\
\out c a  & \mbox{otherwise}
\end{cases} \\
\encoA{ \compP { \resB a b \out c a  } \delta} & =
\res b \out  cb \qquad
\encoA{ \compP { \resB a b \out c b  } \delta}  =
\res b \out  cb \qquad
\encoA{ \compP { \tau  } \delta}  = \tau
\end{align*}

\begin{lemma}
\label{l:encoA}
If $\ltsC \delta P  \mu{\delta'} {P'}$, then $\encoA{ \compP P \delta}
  \arr {\encoA{\compP\mu \delta}} \encoA{\compP {P'}{\delta'}}
$. 
And the converse, on the actions from $\encoA{ \compP P \delta}$.
\end{lemma} 

\begin{proof}
Straightforward transition induction.
\end{proof} 

\begin{corollary}
\label{c:compBis}
Let $\bowtie \;  \in \{ \sbS,\wbS, \expa\}$.
Then 
  $ \encoA{ \compP P \emptyset } \bowtie  \encoA{ \compP Q
    \emptyset }$ 
iff 
${ \compP P \emptyset } \bowtie { \compP Q \emptyset }$.
\end{corollary} 
 In the above translation from \api to \Api, outputs
at success names, carrying the unit value, are mapped onto
themselves. Given this, and since 
the only actions that 
   closed \api and \Api processes 
can perform  are either $\tau$ actions or outputs at
success names, 
we can convert the operational correspondence of Lemma~\ref{l:encoA} into a bisimilarity
result on closed processes.  

\begin{lemma}
\label{l:api_pi}
If $P$ is  closed, then 
$ P  \sbS \encoA{  \compP P \emptyset}$.  
\end{lemma} 
\begin{proof}
  
  We consider the relation
  \begin{align*}
    {\R} \defeq \left\{ (P, \encoA{  \compP P \emptyset}) \; \midbar \; \text{\( P \) is closed} \right\}
  \end{align*}
  and show that this is a strong bisimulation.

  We first consider the challenge moves from \( P \).
  Since \( P \) is closed, the only action \( P \) can do is an internal action or an ouput at a success name.
  We proceed by a case analysis.

  If \( P \arr{\tau} P' \), that is, \( \emptyset \vdash P \arr \tau P'\), then we also have \( \compP P \emptyset \arr \tau \compP {P'} \emptyset \) by the definition of the transitions on composite processes.
  Hence, by Lemma~\ref{l:encoA}, we have \( \encoA{\compP P \emptyset} \arr \tau \encoA{\compP {P'} \emptyset} \) because  \( \encoA{ \compP { \tau  } \delta}  = \tau \).
  We have \( P' \R   \encoA{\compP {P'} \emptyset} \) because a \( P' \) is closed.

  If \( P \arr{\outC a} P' \),  then we have \( \compP P \emptyset \arr {\outC a}\compP {P'} \emptyset \) because the connection set only changes when a bound output is performed.
  Again, by Lemma~\ref{l:encoA} and the definition of \( \encoA {\compP \act \delta} \),  we have \( \encoA{\compP P \emptyset} \arr {\outC a}\ \encoA{\compP {P'} \emptyset} \).
  Since \( P' \) is closed,  we have \( P' \R   \encoA{\compP {P'} \emptyset} \)  as desired.

  The case where \(  \encoA{  \compP P \emptyset} \) makes the challenge is proved similarly because Lemma~\ref{l:encoA} relates the transitions in both directions.
\end{proof}

From Lemma~\ref{l:api_pi}, the compositionality of the translation $\encoA{\cdot}$,
and the soundness of strong and weak bisimilarity ($\sbS$ and $\wbS$) for barbed congruence in
$\pi$-calculi,  we derive:

\begin{corollary}
\label{c:trans_laws}
For $\ok{\Delta}{P,Q}$ with $\n \delta \subseteq \n \Delta $, if
$\encoA{ \compP P \delta} \wbS 
\encoA{ \compP Q \delta} $ then  
$P \wbcTT {\Delta}\delta  Q$.
\end{corollary} 

A similar corollary holds for \emph{asynchronous} bisimilarity
\cite{AmCaSa96},  in
place of the synchronous one.

\subsubsection{Transferring results about transitions and substitutions}
Using Lemma~\ref{l:encoA}, 
we  can import a number of results from \Api relating silent and visible
transitions, and relating
transitions with substitutions.
These results are useful when reasoning on processes.

\iflong
\fi

\begin{lemma}
\label{lem:api-zero}
Suppose $P$ %
$
{\compP P \delta} 
{\arr \tau}{\compP {P'} \delta} $, and this is a synchronisation between two 
names $a,b$ with $a \sepRES b \in\delta $.
Then either 
\begin{enumerate}
\item 
$
{\compP P \delta} 
\arr{\out{b}{c}}\arr{\inp{a}{x}}
{\compP {P''} \delta} 
$,
for some $c,P'$ with 
${\compP {P'} \delta} \equiv {\compP {P''\sub{c}{x}} \delta} $, or  
\item 
${\compP P \delta} 
\arr{\bout{b}{c}}\arr{\inp{a}{x}}
{\compP {P''} {\delta'}} 
$,
for some $\breve c,P'', \delta'$ with 
${\compP {P'} \delta} \equiv {\compP { \resBP c P''\sub{c}{x}} {\delta'}} $.
\end{enumerate}
\end{lemma}

Lemma~\ref{lem:api-one-comp} is a standard result for asynchronous
calculi, relating visible and silent  transitions.

\begin{lemma}
\label{lem:api-one-comp}
Suppose  $a\sepRES b \in \delta $.
\begin{enumerate}
\item
If 
$
{\compP P \delta} 
\arr{\out{b}{c}}\arr{\inp{a}{x}}
{\compP {P'} \delta} 
$,
then also 
$
{\compP P \delta} 
{\arr\tau}\equiv  
{\compP {P'\sub{c}{x}} \delta} 
$.
\item
If
${\compP P \delta} 
\arr{\bout{b}{c}}\arr{\inp{a}{x}}
{\compP {P'} {\delta'}} 
$,
then also 
$
{\compP P \delta} 
{\arr\tau}
\equiv 
{\compP{
\resBP c  P'\sub{c}{x}}{\delta}}
$.
\end{enumerate}
\end{lemma}

Application of
a substitution to a process does not diminish its
capabilities for action. 
A substitution $ \sigma $  \emph{agrees with} a connection set $\delta$
if $a \sepRES b \in \delta $ implies that  also 
 $\sigma(a) \sepRES \sigma(b) \in \delta $ holds.
(We recall that a substitution acts on names of the same type.)  

\begin{lemma}
\label{lem:sub1}
If $\sigma$ agrees with $\delta$ and 
 $ %
{\compP P \delta} 
{\arr\mu}
{\compP {P'} {\delta'}} 
$, 
then also  
 $ %
{\compP {P\sigma} \delta} 
{\arr{\mu \sigma }}
{\compP {P' \sigma } {\delta'}} 
$. 
\end{lemma}

The following (limited) form of converse, with an empty connection set,  will be sufficient  for our purposes.

\begin{lemma}
\label{lem:api-three}
If %
$
{\compP {P\sigma} \emptyset } 
{\arr{\mu  }}
{\compP {P'  } {\delta}} 
$
then 
$
{\compP {P} \emptyset } 
{\arr{\mu' }}
{\compP {P'' } {\delta}} 
$,
for some $P'', \mu' $    with $\mu'\sigma=\mu$
and $P''\sigma = P'$.
\end{lemma}

Lemma~\ref{lem:Api-one-comp-ex} is the `weak' counterpart of Lemma~\ref{lem:api-one-comp}.

\begin{lemma}
\label{lem:Api-one-comp-ex}
Suppose $a\sepRES b \in \delta $.
\begin{enumerate}
\item
If 
$
\compP P { \delta } \Arr {\out{b}{c}} 
\Arr{\inp{a}{x}}  \compP {P'}{ \delta } $
 then also 
$\LtsCs \delta
P{\tau}\delta{ P'\sub{c}{x}}$.
\item
If $
\compP P { \delta } \Arr {\bout{b}{c}} 
\Arr{\inp{a}{x}}  \compP {P'}{ \delta' }$
 then also 
$\LtsCs \delta
P{\tau}{\delta}{
 \resBP c (P'\sub{c}{x})}
$.
\end{enumerate}
\end{lemma}

As a final example of transfer of results concerning transitions, we mention the Harmony Lemma, relating
reductions in the reduction semantics and $\tau$-transitions in the LTS of a process
calculus \cite{SangiorgiWalker01}. Indeed, the correspondence between transitions in the LTS for
\api and in  that of \Api can also be made on their reduction semantics 
(Sections~\ref{s:rs} and Appendix~\ref{a:b}). We omit the details. 
We then derive the following result, expressed on   composite processes.

\begin{lemma}[Harmony Lemma, in \api]
\label{l:hl}  For any \api process $P$ and connection set $\tilx$, we have: 
\begin{itemize}
\item
if $ \res \tilx P \longrightarrow \res \tilx P' $ then also 
$\compP P  \tilx \arr\tau \equiv \compP {P'} {\tilx}$
\item conversely, if 
$\compP P \tilx \arr\tau \compP {P'} {\tilx}$ then also 
 $\res \tilx P \longrightarrow \res \tilx P' $.
\end{itemize}
 \end{lemma} 

\subsubsection{Transferring results about behavioural relations}
On composite processes, the 
 definitions of 
 (strong and weak)   bisimilarity ($\sbS$ and $\wbS$), and 
 of the  expansion relation ($\expa$) are as in Section~\ref{aa:backpi}. 
Exploiting Corollary~\ref{c:compBis} 
 and  the analogous results in  \Api, including
the closure of behavioural relations    under substitutions, 
we can  derive 
the congruence properties of the bisimilarities and 
the precongruence property of the expansion relation in \api.
We only report the result for an initial empty connection set, which is simpler and
sufficient for our needs in the paper. 

\begin{theorem}
\label{t:congS}
Let $\bowtie \;  \in \{ \sbS,\wbS, \expa\}$. If 
$ \compP P \emptyset   \bowtie  \compP Q \emptyset$, then for any \api context  $\qct$
also $ \compP {\ct P} \emptyset   \bowtie  \compP {\ct Q} \emptyset$.
\end{theorem}

We then derive soundness of the relations with respect to barbed
congruence.
\begin{corollary}
\label{c:congS}
Let $\bowtie \;  \in \{ \sbS,\wbS, \expa\}$, and suppose $\ok \Delta {P,Q}$.
 If 
$ \compP P \emptyset   \bowtie  \compP Q \emptyset$, then
also 
$ P   \wbcTT \Delta{}  Q$.
\end{corollary}

Similarly we can  derive: 

\begin{lemma}
\label{l:ddp}
Let $\bowtie \in \{ \sbS,\wbS, \expa\}$.  
Suppose $\delta \subseteq \delta'$. 
Then  $\compP P \delta \bowtie \compP Q \delta$ implies 
$\compP P{\delta'} \bowtie \compP Q{\delta'}$.
\end{lemma}

Using Corollaries~\ref{c:compBis}  and ~\ref{c:congS},
we  can transfer well-known laws  from $\pi$-calculi to \api. 
Some examples laws, that we will use later, are the following ones. 

\begin{corollary}
  \label{c:laws_A}
\noindent
\begin{enumerate}
\item  
\label{i:laws_PC}
$\resBP a (! \inp a x .  P | (Q_1 | Q_2)) 
 \wbcTT \Delta{}
\resBP a (! \inp a x .  P | Q_1)
 |
\resBP a (! \inp a x .  P | Q_2)
$ ;
\item  
\label{i:laws_INP}
$\resBP a (! \inp a x .  P | \inp b y . Q) 
 \wbcTT \Delta{}
\inp b y . \resBP a (! \inp a x .  P |  Q) 
$ ; 
\item  
\label{i:laws_INPrep}
$\resBP a (! \inp a x .  P |!  \inp b y . Q) 
 \wbcTT \Delta{}
! \inp b y . \resBP a (! \inp a x .  P |  Q) 
$ ; 
\item  
\label{i:laws_COM}
$\resBP a (
 \inp a x .  P | \out{\comp a } v) 
 \wbcTT \Delta{}
\resBP a ( P \sub v x )
$.
\end{enumerate} 
\end{corollary} 
The first three laws are often called  `replication theorems'.
To derive law \reff{i:laws_PC} in the corollary from  $\pi$-calculus laws, we have however
first to apply Lemma~\ref{l:wire} (for O-names) so to make sure that name $\compa$ is never exported within
$P$, $Q_1$ and $Q_2$.

\section{Internal processes}
\label{a:inter}

\subsection{Transformation into  internal processes and basic properties of internal processes}
We  define the translation that maps \api processes onto
internal \api processes.
The translation, $\transf \cdot$, uses wires.
We recall that
  $\ulinks  a b $ stands for either
$\links  a b $ or
$\links  ba $, depending on the types for $a$ and $b$.
The translation is a homomorphism on all operators except for output
particles.
\[
\begin{array}{lll}
\transf{P_{1}|P_{2}} \defi \transf{P_{1}}|\transf{P_{2}}
\hspace*{.5cm}
&
\transf{{\inp a x. P}} \defi {\inp a x. \transf{P}}
&
\hspace*{.5cm}
\transf{! {\inp a x .P}} \defi ! {\inp a x. \transf{P}}\\[\tkp]
\transf{{\resBP a}P} \defi {\resBP a}\transf{P}
 &
 \multicolumn{2}{c}{
\transf{{\out a b}} \defi
\bout a c : \transf{{\ulinks {\comp c} b}}    \andalso \mbox{where $c$ and $\comp c$ are fresh}
}
\end{array}
\]

Using induction on types one can show that the encoding is
well-defined (cf., the recursive application of the encoding in
clause for output terminates), and preserves typing (i.e., if $\ok
\Gamma P$ then also $\ok \Gamma {\transf P}$).

The encoding maps \api processes onto  the calculus \apiE, whose syntax is as follows:
\[
P::= \nil  \midd {\inp a x. P} \midd  \transf{{\out a b}} \midd
   P | P \midd   {\resB a b P} \midd ! {\inp a x .P}
\]

Lemma~\ref{l:closedTR} shows that  \apiE is indeed  a set of internal
processes, and is  invariant under transition.

\begin{lemma}
\label{l:closedTR}
If  $P$ is  an \apiE process and $\ltsTv \delta P  \mu
{P_{1}}$ then $\mu$ may not be a free output, and
$P_{1} \in $ \apiE.
\end{lemma}
\begin{proof}
   By structural induction.
 The  interesting case is when
 \[ P =  \transf{\out a b} =
\bout  a c :  \transf{{\ulinks {\comp c} b}}
\]
In this case,  $\mu = {\bout  a c}$ and $P_1 =
\transf{{\ulinks {\comp c} b}}$,  which, by definition, is in \apiE.
\end{proof}

The three lemmas below bring up key properties of \apiE processes.
First,  the translation maps an \api process into a barbed congruent
\apiE process.

\begin{lemma}[Lemma~\ref{l:wire_trans} of the main text]
For any $P \in$ \api  and any $ \Delta $, if $\ok P \Delta $ then
$P \wbcTT \Delta{}  {\transf P}$.
\end{lemma}

\begin{proof}
Follows from the correctness of the wire law,
Lemma~\ref{l:wire}.
\end{proof}

\subsection{Properties of transitions for internal  processes}

Two important properties for transitions on internal processes
 concern bound outputs. Lemma~\ref{l:mm} highlights
a structural property of \apiE processes  capable of performing a
bound output.

\begin{lemma}
\label{l:mm}
If $P$ is a  \apiE process and $\ltsTv \delta P {\bout a c } {P_1}$, then we have:
\begin{enumerate}
\item $P\equiv \res\tilz  ( \resBP{c} (\out a c | \ulinks {\comp c }b) | P_2) $,
and  $P_1\equiv \resb\tilz  ( \ulinks {\comp c}b) | P_2) $,
for
  some $\tilz, b,P_2$ with $\{ a,c,\comp c\} \cap \tilz = \emptyset $ and $c,\comp c \not \in \fn
  {P_2}$.
\item $P \wbS
 \resBP{c} (  \transf{ \out a c} | P_1) $
\end{enumerate}
\end{lemma}

The other relevant lemma about bound outputs for internal processes
is    Lemma~\ref{l:du} in the
 main text. Its
  proof is a straightforward
transition induction.
The lemma is  important for defining proof techniques about
process behaviours.
One of the reasons is that, as hinted at
in the comments to the LTS of  Figure~\ref{fig:LTS},
when an arbitrary \api process performs  a bound output of an O-name, say
$\ltsTv { \delta  } P {\bout b c }  {P'}$, the connection set under
which the behaviour of the derivative $P'$ should be examined is an
extension of the original set $\delta$, namely $\delta, \breve c$,
because $P'$  may retain both the exported name $c$ and its
companion.
Lemma~\ref{l:du} ensures us that, in contrast, an extension of $\delta
$ is never  necessary in \apiE. Therefore, when examining the
behaviour of an \apiE process (the transitions that it can perform and,
recursively, those of its derivatives), the initial connection set
never changes.
The statements and proofs in this section rely on this property.
In the remainder of the section, unless otherwise stated, all
processes are in \apiE.

We write
$\LtsSS \asetD {Q} {}  {\aset}{ Q'}$ if there is $n \geq 1$ and
 $Q_1, \ldots, Q_n$ with $Q_1 = Q$, $Q_n = Q'$ and for each $i<n$,
 $\ltsSS \asetD {Q_i} {\tau}  {\aset}{ Q_{i+1}}$.
Then
 $\LtsSS \asetD Q {\mu}  {\aset}{ Q'}$ holds if
$\LtsSS \asetD {Q} {}  {\aset}{ Q_1}$,
 $\ltsSS \asetD {Q_1} {\mu}  {\aset}{ Q_2}$, and
$\LtsSS \asetD {Q_2} {}  {\aset}{ Q'}$, for some $Q_1,Q_2$ (here again we are exploiting Lemma~\ref{l:du} and the fact that $\delta
$ need not change through transitions in \apiE).
Below we list a few simple results concerning transitions for internal processes.

\subparagraph*{Notation}
For tuples  of I-names  $\tila$ and O-names $\tilb$,  of the same
length and, componentwise, of dual types,
we write
 $ \tila \sepRES \tilb$  to indicate the componentwise connection
 among the names of the tuples.

\begin{lemma}
\label{l:deltaABtau}
Suppose $
\ltsTv { \delta, \tila \sepRES \tilap}  {
P  }\tau {P'}$, and $\tilb \cap \n {\delta, \tila \sepRES \tilap} =
\emptyset $. Then also
 $
 \ltsTv { \delta, \tila \sepRES \tilb}  {
 P \sub {\tilb}\tilap  }\tau {P'\sub {\tilb}\tilap  }$
\end{lemma}

\begin{corollary}
\label{c:deltaABTau}
Suppose $
\LtsTv { \delta, \tila \sepRES \tilap}  {
P  }{} {P'}$, and $\tilb \cap \n {\delta, \tila \sepRES \tilap} =
\emptyset $. Then also
 $
 \LtsTv { \delta, \tila \sepRES \tilb}  {
 P \sub {\tilb}\tilap  }{} {P'\sub {\tilb}\tilap  }$
\end{corollary}

\begin{corollary}
\label{c:deltaABMU}
Suppose $
\LtsTv { \delta, \tila \sepRES \tilap}  {
P  }{\mu} {P'}$, and $\tilb \cap \n {\delta, \tila \sepRES \tilap} =
\emptyset $. Then also
 $
 \LtsTv { \delta, \tila \sepRES \tilb}  {
 P \sub {\tilb}\tilap  }{\mu} {P'\sub {\tilb}\tilap  }$
\end{corollary}

\begin{lemma}
\label{l:deltaABres}
Suppose $
\LtsTv { \delta, \tila \sepRES \tilap}  {
P  }{} {P'}$. Then also:
\begin{enumerate}
\item
$
\LtsTv { \delta}  {
\res{ \tila \sepRES \tilap}P  }{} {\res{ \tila \sepRES \tilap}P'}$;

\item
 $
\LtsTv { \delta, \tila \sepRES \tilap}  {
\res \tilz P  }{} {\res \tilz P'}$,
 for  $\tilz \cap \n {\delta, \tila \sepRES \tilap} =
\emptyset $.
\end{enumerate}
\end{lemma}

The following lemma transport Lemma~\ref{lem:Api-one-comp-ex} to the LTS
for plain (as opposed to composite) internal processes.
\begin{lemma}
\label{lem:Api-one}
Suppose $a\sepRES b \in \delta $.
\begin{enumerate}
\item
If
$
\ltsBEG {\delta} P  \Arr {\out{b}{c}} P_1$ and
$\ltsBEG {\delta} P_1
\Arr{\inp{a}{x}}   {P'} $
 then also
$
\LtsSS \delta
P{\tau}\delta{ P'\sub{c}{x}}$.
\item
If
$
\ltsBEG {\delta} P  \Arr {\bout{b}{c}} P_1$ and
$\ltsBEG {\delta} P_1
\Arr{\inp{a}{x}}   {P'} $
 then also
$\LtsSS \delta
P{\tau}\delta{ \resBP c (P'\sub{c}{x})}$.
\end{enumerate}
\end{lemma}

\subsection{Some results on  internal bisimulation}

\begin{lemma}
\label{l:is_instance}
If  $a \not \in \n \delta $ and $c$ is fresh,
and $P \wbT \delta Q$, then also
$P \sub c a \wbT \delta Q \sub c a$.
\end{lemma}

We will  need some  `up-to' proof of technique for  internal
bisimulation:

\begin{itemize}
\item
\emph{internal
bisimulation up-to $\wbS$ and  $\expa$}:  it is the analogue of the
`weak bisimulation up-to $\wbS$ and $\expa$'  of Theorem~\ref{t:uptoWE};
thus its definition is the same as  that of internal bisimulation
(Definition~\ref{d:hb-bisimulation}) except that, in the conclusion of
the clauses, the appearance of
$\sor{\R}{\tilx} $, where $\tilx $ is a connection set, is replaced by
$\contr \; \sor{\R}{\tilx} \;\,  \wbS $; that is, on the side of the
 process performing the challenge, one
can make use of  $\contr$, whereas on the side of the process
 answering the challenge, one can make use of $\wbS$.

\item
\emph{internal
bisimulation up-to $\wb$};
here the definition is the same as  that of internal bisimulation
with, in the conclusion of
the clauses, the appearance of
$\sor{\R}{\tilx} \;\,  \wb $, in place of
$\sor{\R}{\tilx} $.
\end{itemize}

\begin{lemma}
\label{l:uptowb}
Suppose  $\R$ is either an internal
bisimulation up-to $\wbS$ and  $\expa$ or
 an internal  bisimulation   up to
$ \wb { }$, suppose $(\delta  , P, Q)\in \R$.
Then $P \wbT \delta Q$.
\end{lemma}

\subsection{Substitutivity properties  of  internal bisimulation}
We now tackle the two main and most delicate substitutivity properties of internal
bisimulation, for the restriction and parallel composition operators (Lemmas~\ref{l:res}
and~\ref{l:par} in the main text).

\subparagraph*{Notation}
As in the LTS of Figure~\ref{fig:LTS},
with transitions such as $\ltsTv { \delta } P {\mu}  {P'}$,
the parameter $\delta $ is only used in the rules to infer
$\tau$-actions, in proofs we sometimes abbreviate such a transition
with $P \arr\mu P'$ when  $\mu$ is a visible action.

The assertion of Lemma~\ref{l:res}   follows from the lemma below by taking
$\tila = \tilap =\tilb = \emptyset $ and
$\tilz = z\sepRES z'$.

\begin{lemma}
\label{l:crux}
If  $P \wbT{\delta, \tila \sepRES \tilap} Q$
and $\tilap \cap \tilb = \emptyset $, then
\[   \res \tilab \res \tilz (P \sigmaapb )
\wbT{\delta}
 \res \tilab
\res \tilz
 (Q \sigmaapb )     \]
\end{lemma}

\begin{proof}
We consider the relation $\R$ with all triples  of the form:
\[ ( \delta ,  \res \tilab \res \tilz (P \sigmaapb ), \res \tilab
\res \tilz
 (Q \sigmaapb )     \]
with
 $P \wbT{\delta, \tila \sepRES \tilap} Q$
and $\tilap \cap \tilb = \emptyset $

and show  that this is an internal bisimulation
up-to $\wbS$ and $\expa$.
Set $\sigma = \sub \tilb \tilap$, and $\tilx =
\tilab,  \tilz$.

Consider the challenge moves from
$\res \tilx (P \sigma )$.
The interesting case is an interaction between two
names $a,b$ with $a \sepRES b \in \tilx$.

Thus suppose
$
\ltsTv \delta {
\res \tilx (P \sigma )}\tau {P_\star}$, with
$P\sigma  \arr{\bout b c} \arr{\inp a x} P_\#$ and
$P_\star = \res \tilx \resBP c  P_\# \sub c x $.

We first consider the case $a \sepRES b \in \tilab$. Then we will consider
the case  $a \sepRES b \in \tilz$.

\begin{description}
\item[Case $a \sepRES b \in \tilab$.]

Let $a' $ be the
companion of $a$ in $\tila \sepRES \tilap$. Moreover we have $\sigma
(a') = b$.
We have   either
$P\arr {\bout{a'} c} P_1 \arr{\inp a x } P_2$, or
$P\arr {\bout{b} c} P_1 \arr{\inp a x } P_2$
 with
$P_\#  = P_2 \sigma $ and therefore
$P_\star  = \res \tilx \resBP c ( P_2 \sigma\sub c x  )$.

We consider the two cases.
\begin{enumerate}
\item
If $P\arr {\bout{a'} c} P_1 \arr{\inp a x } P_2$,
since $a\sepRES a' \in \delta $, we also have
$\ltsTv  \delta P {\tau}  {\resBP c (P_2 \sub c x )}$.

Then
$\LtsTv  \delta Q {}   {Q'}$ with
${\resBP c (P_2 \sub c x )}
 \wbT{\delta, \tila \sepRES \tilap} Q'$

It also holds
$\LtsTv  \delta{ \res\tilx Q\sigma }{}   {\res\tilx Q' \sigma }$
and we are done, as
\[
P_\star =
\res \tilx( (\resBP c  P_2\sub c x  )  \sigma) \; \: \R_\delta
\; \:\res\tilx (Q' \sigma )
\]

\item
If
$P\arr {\bout{b} c} P_1 \arr{\inp a x } P_2$,
since  $P \wbT{\delta, \tila \sepRES \tilap} Q$, and $b\not \in \n{\delta, \tila \sepRES \tilap}$,
there is
$$\ltsBEG {\delta, \tila \sepRES \tilap} Q\Arr {\bout{b} c} Q_1$$
 with
  $P_1 \wbT{\delta, \tila \sepRES \tilap} Q_1$.
From Corollary~\ref{c:deltaABMU},
We also have
\begin{equation}
\label{e:ltsBEG}
 \ltsBEG {\delta, \tila \sepRES \tilb}
Q\sigma \Arr {\bout{b} c} Q_1 \sigma .
\end{equation}
Given the challenge $\ltsBEG {\delta, \tila \sepRES \tilap}P_1 \arr{\inp a x } P_2$, we have two subcases
for $Q_1$:
\begin{enumerate}
\item $\ltsBEG {\delta, \tila \sepRES \tilap}Q_1 \Arr{\inp a x } Q_2$ and
  $P_2 \wbT{\delta, \tila \sepRES \tilap} Q_2$.
Then, since $c$ is fresh for $P_1,Q_1$ (Lemma~\ref{l:du}), by Lemma~\ref{l:is_instance} also
  $P_2 \sub c x  \wbT{\delta, \tila \sepRES \tilap} Q_2 \sub c x$.

Now, from  Corollary~\ref{c:deltaABMU} we have
$ \ltsBEG {\delta, \tila \sepRES \tilb}
Q_1 \sigma \Arr{\inp a x} Q_2 \sigma $.
From this and \reff{e:ltsBEG}, using Lemma~\ref{lem:Api-one}, we derive
$ \ltsBEG {\delta, \tila \sepRES \tilb}
Q \sigma \Longrightarrow  Q_2 \sigma \sub c x $.
Finally, also
$ \ltsBEG {\delta}
\res \tilx
 (Q \sigma ) \Longrightarrow
\res \tilx \resBP c
 (Q_2 \sigma \sub c x  )
     $, using Lemma~\ref{l:deltaABres}.
We can thus conclude
\[ ( \delta ,  \res \tilx \resBP c(P_2 \sub c x \sigma ), \res \tilx
\resBP c (Q_2\sub c x \sigma )    \in \R \]

\item
The answer from $Q_1$ does not involve an input action.
We recall that we have
 $a  \in \tila$, that $a'$  is its companion according to
$\tila \sepRES \tilap$, and that $\sigma (a') = b$.

We have
$
\ltsBEG {\delta, \tila \sepRES \tilap}
Q_1 \Longrightarrow Q_2 $ and
  $P_2 \wbT{\delta, \tila \sepRES \tilap}
\outbis {a'} x | Q_2  $.
Again, using injective substitutions we also have
\begin{equation}
\label{e:P2cxU}
P_2 \sub c x  \wbT{\delta, \tila \sepRES \tilap}
(\outbis {a'} x | Q_2) \sub c x
=
\outbis {a'} c | Q_2
\end{equation}

Since
$
\ltsBEG {\delta, \tila \sepRES \tilap}
Q\Arr {\bout{b} c} Q_1$, we have (Lemma~\ref{l:mm}(2))
$Q \Longrightarrow  \wbS \resBP c ( \outbis {b} c | Q_1)$.
Hence also
\[ \ltsBEG {\delta, \tila \sepRES \tilap}
Q \Longrightarrow \wbS
 \resBP c ( \outbis {b} c | Q_2)
\]
and therefore, using Lemma~\ref{l:deltaABres} and Corollary~\ref{c:deltaABTau},
\[ \ltsBEG {\delta}
\res\tilx (Q\sigma)  \Longrightarrow \wbS
\res\tilx  (\resBP c ( \outbis {b} c | Q_2) ) \sigma
=
\res\tilx  \resBP c(  ( \outbis {a'} c | Q_2 )  ) \sigma
\]
From \reff{e:P2cxU}
 we can conclude
\[ ( \delta ,
 \res \tilx \resBP c ( P_2 \sub c x \sigma  ),
\res\tilx  \resBP c(  ( \outbis {a'} c | Q_2 ) \sigma  )
\in \R \]
and we are done, up-to $\wbS$,
as
$P_\star  = \res \tilx \resBP c ( P_2 \sub c x \sigma )$.
\end{enumerate}
\end{enumerate}
\medskip
\par
\item[Case $a \sepRES b \in \tilz$.]

This is similar to the previous case (in fact, simpler).
The novelties are in subcase (b)  below.

We have
$P\arr {\bout{b} c} P_1 \arr{\inp a x } P_2$
 with
$P_\#  = P_2 \sigma $ and therefore
$P_\star  = \res \tilx \resBP c ( P_2 \sigma\sub c x  )$.

Since  $P \wbT{\delta, \tila \sepRES \tilap} Q$, and $b\not \in \n{\delta, \tila \sepRES \tilap}$,
there is
$
\ltsBEG {\delta, \tila \sepRES \tilap}
Q\Arr {\bout{b} c} Q_1$ with
  $P_1 \wbT{\delta, \tila \sepRES \tilap} Q_1$.
We also have
$
\ltsBEG {\delta, \tila \sepRES \tilb}
Q\sigma \Arr {\bout{b} c} Q_1 \sigma $.

Given the challenge $
\ltsBEG {\delta, \tila \sepRES \tilap}
P_1 \arr{\inp a x } P_2$, we have two subcases
for $Q_1$:
\begin{enumerate}
\item $
\ltsBEG {\delta, \tila \sepRES \tilap}
Q_1 \Arr{\inp a x } Q_2$ and
  $P_2 \wbT{\delta, \tila \sepRES \tilap} Q_2$.
Then, since $c$ is fresh for $P_1,Q_1$ (Lemma~\ref{l:du}), by Lemma~\ref{l:is_instance} also
  $P_2 \sub c x  \wbT{\delta, \tila \sepRES \tilap} Q_2 \sub c x$.

Moreover,
reasoning as in the previous case,
 $
\ltsBEG {\delta, \tila \sepRES \tilb}
Q_1 \sigma \Arr{\inp a x} Q_2 \sigma $, and also
$
\ltsBEG {\delta}
\res \tilx
 (Q \sigma ) \Longrightarrow
\res \tilx \resBP c
 (Q_2 \sigma \sub c x  )
     $.
We can thus conclude
\[ ( \delta ,  \res \tilx \resBP c(P_2 \sub c x \sigma ), \res \tilx
\resBP c (Q_2\sub c x \sigma )    \in \R \]

\item
The answer from $Q_1$ does not involve an input action.
We recall that we have
 $a \not \in \tila$.

We have
$
\ltsBEG {\delta, \tila \sepRES \tilap}
Q_1 \Longrightarrow Q_2 $ and
  $P_2 \wbT{\delta, \tila \sepRES \tilap, a \sepRES a''}
\outbis {a''} x | Q_2  $, where $a''$ is fresh.
Again, we can apply an injective substitution and derive
\begin{equation}
\label{e:P2cx}
P_2 \sub c x  \wbT{\delta, \tila \sepRES \tilap, a \sepRES a''}
(\outbis {a''} x | Q_2) \sub c x .
=
\outbis {a''} c | Q_2
\end{equation}

Since
$
\ltsBEG {\delta, \tila \sepRES \tilap}
Q\Arr {\bout{b} c} Q_1$,  employing Lemma~\ref{l:mm}(2),
we have
$
\ltsBEG {\delta, \tila \sepRES \tilap}
Q \Longrightarrow  \wbS \resBP c ( \outbis {b} c | Q_1)$.

Hence also
\[
\ltsBEG {\delta, \tila \sepRES \tilap}
Q \Longrightarrow \wbS
 \resBP c ( \outbis {b} c | Q_2)
\]
and therefore
\[
\ltsBEG {\delta}
\res\tilx (Q\sigma)  \Longrightarrow \wbS
\res\tilx  (\resBP c ( \outbis {b} c | Q_2) ) \sigma
=
\res\tilx  \resBP c(  ( \outbis {a''} c | Q_2 ) \sigma \sub b{a''} )
\, .
\]
From \reff{e:P2cx} and
$P_\star  = \res \tilx \resBP c ( P_2 \sub c x \sigma )
=
   \res \tilx \resBP c ( P_2 \sub c x \sigma \sub b {a''} )
$ (for  $a'' $ does not appear in $P_2$) we can conclude

\[ ( \delta ,
 \res \tilx \resBP c ( P_2 \sub c x \sigma \sub b {a''} ),
\res\tilx  \resBP c(  ( \outbis {a''} c | Q_2 ) \sigma \sub b{a''} )
\in \R \, . \]

\end{enumerate}
\end{description}
 \end{proof}

\begin{corollary}
\label{c:cruc}
If $P \wbT \emptyset  Q$,  then also
$\resB a b P \wbT \emptyset \resB a b Q$
\end{corollary}

We now consider the substitutivity for parallel composition.

\begin{lemma}[Lemma~\ref{l:par} in the main text]
  Suppose that \( P \wbTT \Delta{\delta}  Q \), that $ \ok \Gamma R$, and that
$\iname \Delta \cap \iname \Gamma  = \emptyset $.
Then also
 \( P \mid R \wbTT{\Delta , \Gamma }{ \delta }  Q \mid R\).
\end{lemma}

\begin{proof}
We take the set $\R$ of all (well-typed) triples of the form
\[
(
 \delta ;
\res {\tily} \res {\tilx} ( P \mid R)  ;
\res {\tily}  \res {\tilx} ( Q \mid R)  )
\]
where, for some $\Delta $ and $\Gamma $:
\begin{itemize}
\item
 \( P \wbTT \Delta{ \delta, \tilx }  Q \) ;
\item $\n \delta , \tilx \subseteq \n \Delta $ ;
\item  for each $y \sepRES y' \in \tily$, the set $\{y,y'\} \cap \n \Delta $ is a
  singleton (exactly one element);

\item  $ \ok \Gamma R$ and
$\iname \Delta \cap \iname \Gamma  = \emptyset $;
\end{itemize}
and show that
 $\R$ is a
internal bisimulation up-to $\wbS$ (and $\expa$, though $\expa$ is not
used here).
We abbreviate $
\res \tily \res \tilx $ as $\res \til z$, and $ \delta , \tilx  $ as  $ \deltap $.
We proceed by a case analysis on the possible challenge actions coming
from
$\res {\tilz} ( P \mid R)  $.
Similarly to what was done in the proof of Lemma~\ref{l:crux}, we often
use Lemma~\ref{l:deltaABres}, to transport
 reductions of a process  to a
restricted processes with a smaller connection set as a parameter; we
will not explicitly refer to the lemma.

\begin{enumerate}
\item  Action from $R$ or $P$ alone: this case is easy.

The most interesting case is a $\tau$ action from $P$ stemming from an interaction
along names in $\tilx$.  We can then infer an answer from $Q$ using the fact that $\tilx$
appears  in  $\deltap$.

\item
Interaction between $P$ and $R$ along names in $\tilx$.

In this case the input must originate from $P$ (as $P$ owns the input end).
Thus let $a,\compa$ be the names involved, and
 suppose   $P \arr {\inp a c} P'$, and $R \arr {\bout {\compa } c} R'$, and
$
\ltsBEG {\delta}
\res {\tilz} ( P \mid R) \arr\tau
\res {\tilz} \resBP c ( P'\mid R')
\equiv
\resBP c \res {\tilz}  ( P'\mid R')
$.

Since  \( P \wbT  \deltap Q  \),
we have two cases.

\begin{enumerate}
\item
 $
\ltsBEG {\delta'}
Q \Arr {\inp a c} Q'$ with
$P' \wbT{\deltap} Q'$.

We have
$
\ltsBEG {\delta}
\res {\tilz} ( Q \mid R) \Arr\tau  \equiv
\resBP c \res {\tilz}( Q'\mid R')$
and we are done.

\item
 $
\ltsBEG {\delta'}
Q \Longrightarrow Q'$ with
\begin{equation}
\label{e:deltap}
P' \wbT{\deltap} Q' | \outbis \compa c
\end{equation}
  (as $a\sepRES \compa \in \deltap$).
We have
\[
\ltsBEG {\delta}
\res {\tilz} ( Q \mid R) \Longrightarrow
\res {\tilz}( Q' \mid R) \defi Q_1
 \]
Now, since  $R$ is a \apiE process, and
$R \arr {\bout{ \compa } c} R'$, by Lemma~\ref{l:mm}(2),
we have $R \wbS
 \resBP{c} (  \transf{ \out \compa c} | R')$, hence
\[Q_1 \wbS
\res {\tilz} \resBP{c} ( Q' |
  \transf{ \out \compa c} | R') \equiv
\resBP{c} \res {\tilz}  ( Q' |
  \transf{ \out \compa c} | R')
 .\]
\end{enumerate}
and we are done, up-to $\wbS$, since, by \reff{e:deltap}, processes
$\resBP c \res {\tilz}  ( P'\mid R')$ and $\resBP{c} \res {\tilz}  ( Q' |
  \transf{ \out \compa c} | R')  $ are related.

\item
Interaction between $P$ and $R$ along names in $\tily$ in which $P$ makes the input.
Let $a,\compa$ the names involved, and
 suppose   $P \arr {\inp a c} P'$, and $
R \arr {\bout {\compa } c} R'$, and
$
\ltsBEG {\delta}
\res {\tilz} ( P \mid R) \arr\tau
\res {\tilz} \resBP c ( P'\mid R')
$. We know that $\compa$ does not appear in $P,Q$.
Again, as  \( P \wbT{  \deltap } Q \),
we have two cases.

\begin{enumerate}
\item
 $
\ltsBEG {\delta'}
Q \Arr {\inp a c} Q'$ with
$P' \wbT{\deltap} Q'$.

We have
$
\ltsBEG {\delta}
\res {\tilz} ( Q \mid R) \Arr\tau  \equiv
\resBP c \res {\tilz}( Q'\mid R')$
and we are done, as before.

\item
 $
\ltsBEG {\delta'}
Q \Longrightarrow Q'$ with
$P' \wbT{\deltap, a\sepRES\compa} Q' | \outbis \compa c $ (we can use exactly $\compa$
here because it is fresh for $P$ and $Q$ and bisimilarity is preserved by injective
substitutions). We proceed as in the case (2.b) above.
We have
\[
\ltsBEG {\delta}
\res {\tilz} ( Q \mid R) \Longrightarrow
\res {\tilz}( Q' \mid R) \defi Q_1
 \]
Now, since  $R$ is a \apiE process, and
$R \arr {\bout{ \compa } c} R'$, by Lemma~\ref{l:mm}(2),
we have $R \wbS
 \resBP{c} (  \transf{ \out \compa c} | R')$, hence
\[Q_1 \wbS
\res {\tilz} \resBP{c} ( Q' |
  \transf{ \out \compa c} | R')
 .\]
and we are done (note that in this case $\breve c$ contributes to the `$\tilx$' part,
rather  than the `$\tily$' as in case (2.b) above; this is the only difference with
respect to that case).

\end{enumerate}

\item
Interaction between $P$ and $R$ along names in $\tily$ in which $P$ makes the output.
Let $a,\compa$ the names involved; we have
$P \arr {\bout \compa c} P'$, and $R \arr {\inp { a  } c} R'$, and
$
\ltsBEG {\delta}
\res {\tilz} ( P \mid R) \arr\tau
\res {\tilz} \resBP c ( P' \mid R') \equiv
\resBP c \res {\tilz}  ( P' \mid R')
$.

Since  \( P \wbT{  \deltap} Q \)
and $ a\sepRES \compa  \not \in  \deltap$,
it holds that $
\ltsBEG {\delta'}
Q \Arr {\bout \compa c} Q'$ with
$P' \wbT{\deltap } Q'$.
Hence also
\[
\ltsBEG {\delta}
\res {\tilz} ( Q \mid R) \Arr\tau
\res {\tilz} \resBP c ( Q' \mid R') \equiv
\resBP c
\res {\tilz}  ( Q' \mid R')
 \]
and we are done.
\end{enumerate}
\end{proof}

\section{Completeness for barbed congruence}
\label{a:completeness}

In the definition below, the typing environment for the tested processes is not omitted,
as we did in the Definition~\ref{d:hb-bisimulation}  of the bisimilarity, 
because we will need it in the following completeness  proof.
Once more, in the definition we exploit the property of the subset of internal
processes in \apiE in Lemma~\ref{l:du}.

\begin{definition}[Stratification of bisimilarity]
\label{d:strati}
$ $ 
\begin{itemize}
\item
 $ \wbTTn 
 0 \Delta \asetD  \defi \{ (P,Q) \st \ok \Delta {P,Q}  \}$

\item for $n > 0$, we have
$P 
 \wbTTn 
n \Delta \asetD Q$ if 
 the following clauses  hold:
 \begin{enumerate}

    \item If 
$\ltsSS \asetD P \tau  {\aset}{ P'}$, then there is $Q'$ such that
$\LtsSS \asetD Q {}  {\aset}{ Q'}$
 and 
$P' 
 \wbTTn 
{n-1} \Delta \asetD Q'$;

    \item
if  $
\ltsSS \asetD P { {\bout b c}}   {\aset}{ P'}$
and   $b \not \in \n \asetD$, 
         then there is $Q'$ such that 
$\LtsSS \asetD Q {{\bout b c}}  {\aset}{ Q'}$
 and
$P' 
 \wbTTn 
{n-1}{ \Delta'} \asetD Q'$, 
where $ \Delta' $ is the extension of $ \Delta $ in which  
 the appropriate  type for $\comp c$ is added;

     \item If 
$\ltsSS \asetD P {a(c)}  {\aset'}{ P'}$, 
and $ \Delta' $ is the extension of $ \Delta $ in which  
 the appropriate  type for $ c$ is added,
   then  there is $ Q'$ such that:
         \begin{enumerate}
           \item either 
$\LtsSS \asetD Q { a( c)}  {\aset'}{ Q'}$
 and
$P' 
 \wbTTn 
{n-1}{ \Delta'} \asetD Q'$, 
           \item or 
$\LtsSS \asetD Q {}  {\aset}{ Q'}$
 and, either
\begin{enumerate}
\item
 $a\sepRES a' \in  \asetD$, for some $a'$,  and
$P' 
 \wbTTn 
{n-1}{ \Delta'} \asetD 
(Q' | {\outbis  {a'}  c})
$, 
 or

\item   
$a \not \in \n \asetD$ and
$P' 
 \wbTTn 
{n-1}{ \Delta''} {\asetD, a\sepRES a' }
(Q' | {\outbis  {a'}  c})
$, 
 for $a' $ fresh, and where 
 $ \Delta'' $ is the extension of $ \Delta' $ in which  
 the appropriate  type for $a'$ is added.
\end{enumerate}
          \end{enumerate}
 \end{enumerate}
\item 
 $ \wbTTn 
 \omega \Delta \asetD  \defi 
\bigcap_n   \wbTTn 
n \Delta \asetD
$.
\end{itemize}
\end{definition}

\begin{theorem}
\label{t:omega_strat}
Let $P$ and $Q$ be two image-finite regular \apiE-processes. Then 
 ${ \wbTTn 
 \omega \Delta \asetD}  = { \wbTT 
 \Delta \asetD}$.
\end{theorem}

We will also use the following lemma.

\begin{lemma}
\label{l:notfreedelta}
Suppose 
$
\ltsBEG {\delta, a\sepRES b } P \arr\tau P'$, and $a$ or $b$  not free in $P$. 
Then also  $
\ltsBEG {\delta } P \arr\tau P'$
\end{lemma} 

\Completeness*
 \begin{proof}
We exploit the stratification 
$\{\wbTTn 
 n \Delta \asetD \}_n$
of the labeled bisimilarity 
$\wbTT  \Delta\delta $, on the natural numbers, defined above.
We define a collection of tests $\RLc n \LL$
depending on the integer $n$ and the finite set $\LL$.

For defining the tests $
\RLc n \LL 
$ we 
use the internal choice operator $\oplus$, introduced in Example~\ref{ex:internal-choice},
as a derived operator and  defined as follows:
\begin{center}
   $P_{1} \oplus \cdots \oplus P_{n} \defi
   {\resB a{a'}}(
\resB{b_1}{b_1'}(\out{a}{b_1'}| b_1 .P_{1})
| \cdots |
\resB{b_n}{b_n'}(\out{a}{b_n'}| b_n .P_{1})
|
a(x) .  \outC x )$ \\
   with $a,a',b_1,b_1' \ldots  b_n,b_n'\not \in \fn{P_{1},\ldots ,P_{n}}$.
\end{center}
In the above definition, for each $i$ (we recall
that $\simS$ stands for   
ordinary strong bisimilarity and $\contr$ for  the expansion relation,
Section~\ref{a:ls}) we have:    
\[ 
   P_{1} \oplus \cdots \oplus P_{n}  \arr\tau \arr\tau \; \simS P_i   \]
and 
\[ 
   P_{1} \oplus \cdots \oplus P_{n}  \arr\tau  \contr  \sim P_i   \]
In the transitions below, the appearances  of $\simS$ and $\contr$ are due to
 transitions involving internal choice as above. 

We assume that 
 the following are all distinct  success names:  
\[ \{ b_{n},b'_{n} \st n \geq 0  \} \cup
 \{ c_{n}^{c} \st n \geq 0, \:c  \; \mbox{ not a success name}
 \} \]

 The test $
\RLc n \LL 
$ is defined by induction on $n$ as follows. 
In the definition we associate 
a success name  $c_{n}^{a}$
to 
 each pair  of a    name $a $ in $\LL$ and of  $n$, to reason about   actions along
$a$ at  stage $n$.  (There is some abuse of notation, to
facilitate the reading: the abuse is due to the appearance of names in the syntax of other
names, which conflicts with the possibility of $\alpha$-converting names; however this is harmless, as  the only property we need is
that the  success names
 so picked  are all distinct.)

 Intuitively, $A$ is the set of names that the tester observer, or context, 
is supposed to use. 
\begin{itemize}
\item
For $n = 0$, we set:
$ 
\RLc 0 \LL 
\defi \overline{b_{0}} \oplus
                    \overline{b'_{0}}$;
\item
For   $n > 0$ , we set:
 \begin{tabbing}
 $
\RLc n \LL 
 \defi$
     \= $\overline{b_{n}} \oplus \overline{b'_{n}} $ \\
     \> $\oplus ( \overline{c_{n}^{\tau}} \oplus
\RLc {n-1} \LL   
  )$\\
     \> $\oplus 
(
\{ \overline{c_{n}^{{a}}} \oplus
            {\bout b {a'}} :\RLc {n-1} {\LL\cup \{ \comp{a'}\}}  \st
\begin{array}[t]{l}
b \in {\LL} \mbox{ and is an O-name} \})  
\end{array}
$
 \\
     \> $\oplus
(
 \{ 
 \overline{c_{n}^{a}} \oplus
           {\inp a x. { 
           \RLc {n-1}{\LL \cup \{ x \}}}}
)  
\st
a \in {\LL} \mbox{ and is an I-name} \})
$
 \end{tabbing}
\end{itemize} 

A connection set   $\tilx$ and  $\RLc {n}\LL$ (defined as above) are \emph{complete for $ \Delta $
and $ \delta $} if the following conditions hold: 
\begin{itemize}
\item[(a)]
$\delta \subseteq 
\tilx$;
\item[(b)]
for all  $a \in (\iname \Delta \setminus \iname \delta)$ there is $b$ that is fresh for $ \Delta $ with
$ a\sepRES b \in \tilx$;

\item[(c)]
for all  
 $b \in (\oname \Delta \setminus \oname \delta
)$ there is $a$ that is fresh for $ \Delta $ with
$ a\sepRES b \in \tilx$.

\item[(d)]
$ %
A 
= 
 \{ b \st a \sepRES b \in \tilx, \mbox{ for some $a \in \iname \Delta $} \} \cup 
 \{ a \st a \sepRES b \in \tilx, \mbox{ for some $a \not \in \iname \Delta $} \} 
$
\end{itemize} 
Intuitively,  $\tilx$ is closing for $ \Delta $ and includes $\delta$
(formally, it is the context $\res \tilx (\contexthole|\RLc n \LL )) $
that is closing for $\Delta $).  
And $A$ has the companions of the I-names in $ \Delta $ (such companion may also be in $
\Delta $, when the pair is in $\delta$); and the companions of the 
O-names in $\Delta$ that are not in $\delta$.

We prove by induction on $n$ that
for  $\tilx$  and $\RLc n \LL $ complete  
 for $ \Delta $
and $ \delta $ we have: 
whenever 
 $\ok \Delta {P,Q}$, 
\begin{center}
${\res {\tilx}(P|
\RLc n \LL 
)} \wbb {\res {\tilx}(Q|
\RLc n \LL 
)}$ implies
$ P \wbTTn   n \Delta \asetD  Q$.
\end{center}

 The proof is by induction on $n$.  The case $n=0$ is trivial
 because
$\wbTTn  0 \Delta   \asetD$
 has all pairs of
     processes.  If $n>0$ we suppose  that
for all $\tilx, \LL$ as above it holds that 
         \begin{center}
         ${\res {\tilx}}(P|\RLcp) \wbb
 {\res{\tilx}}(Q|\RLcp)$
 \end{center}
and consider
  $\ltsTnew \delta  P 
 { \mu} {A'} {P'}$.
  We
 proceed by case analysis on the action $\mu$ to show
 that $Q$ can match it.
 \begin{enumerate}
  \item $\mu = \tau$.  Then 
   \[
\ltsBEG {\emptyset }
{\res {{{\tilx}}}}(P|\RLcp) {\Arr \tau} \simS
   {\res {{{\tilx}}}}(P|(\overline{c_{n}^{\tau}} \oplus
\RLc{n-1}\LL))
   \]
   To match this reduction (up to barbed bisimulation) we must have:
   \[ 
\ltsBEG {\emptyset }
{\res {{{\tilx}}}}(Q|\RLcp) {\Arr \tau}
\contr   {\res {{{\tilx}}}}(Q_{1}|(\overline{c_{n}^{\tau}} \oplus 
\RLc{n-1}\LL
))
   \]
with $
\ltsBEG {\tilx}
Q \Longrightarrow Q_1$, and hence also 
 $
\ltsBEG {\delta}
Q \Longrightarrow Q_1$ (Lemma~\ref{l:notfreedelta}). 
   We make a further reduction:
   \[ 
{\res {{{\tilx}}}}(P|(\overline{c_{n}^{\tau}} \oplus
\RLc {n-1} \LL   ) 
      {\Arr \tau}
\simS      {\res {{{\tilx}}}}(P'|
\RLc {n-1} \LL    ) 
   \]
   Again this has to be matched by (note that we cannot run the process
   $
\RLc {n-1} \LL $ without losing a commitment $\overline{b_{n-1}}$
   or $\overline{b'_{n-1}}$):
   \[
\ltsBEG {\emptyset }
 {\res {{{\tilx}}}}(Q_{1}|(\overline{c_{n}^{\tau}} \oplus 
\RLc {n-1} \LL 
))
      {\Arr \tau} \contr
        {\res {{{\tilx}}}}(Q'|
\RLc {n-1} \LL 
))
   \]
   We have $
\ltsBEG {\delta}
Q {\Arr \tau} Q_{1} {\Arr \tau} Q'$ and the conclusion
   follows
   by applying the inductive hypothesis.

    \item $\mu = {\bout {b } {c}}$, and there is $a$ with 
$a \sepRES b \in \tilx$ and 
 $a \in \LL $. Moreover we  assume 
that  $c$ is an O-name. 
We may also suppose $c \not \in \fn{Q}$. 
   Then
   $$ \begin{array}{rcl}
\ltsBEG {\emptyset }    {\res {{\tilx}}}(P|
\RLcp) & {\Arr \tau} \simS %
     {\res {{\tilx}}}(P|
\overline{c_{n}^{a}} \oplus
          {\inp a x. { 
           \RLc {n-1}{\LL \cup \{ x \}}}
}))
 \end{array} $$
   This has to be matched by:
   $$ \begin{array}{rcl}
\ltsBEG {\emptyset }     {\res {{\tilx}}}(Q|
\RLcp) & {\Arr \tau} \contr %
     {\res {{\tilx}}}(Q_{1}|
(
\overline{c_{n}^{a}} \oplus
          {\inp a x. { 
           \RLc {n-1}{\LL \cup \{ x \}}}
}))
\end{array}$$
with $\ltsBEG {\delta} Q \Longrightarrow Q_1$.
   We take some further steps:
   $$ \begin{array}{rcl}
\ltsBEG {\emptyset }     {\res {{\tilx}}}(P|
(
\overline{c_{n}^{a}} \oplus
          {\inp a x. { 
           \RLc {n-1}{\LL \cup \{ x \}}}
}))
      & {\Arr \tau} \simS %
     (\res {\tilx} \resBP
 {c})
(P'|
\RLc {n-1}{\LL \cup \{ c \}} 
)
\end{array}$$
   This has to  be  matched thus:
   $$ \begin{array}{rcl}
\ltsBEG {\emptyset }     {\res {{\tilx}}}(Q_{1}|
(
\overline{c_{n}^{a}} \oplus
          {\inp a x. { 
           \RLc {n-1}{\LL \cup \{ x \}}}
}))
      & {\Arr \tau}  \contr %
(     \res {\tilx}\resBP c)
(
Q'| \RLc{n-1}{\LL
         \cup \{c \}} 
)
\end{array}$$
   This means that $\ltsBEG {\delta}Q {\Arr \tau} Q_{1} {\Arr {\bout {b} c  }} 
   Q'$.  Using the 
inductive hypothesis
we conclude  that 
     $
   P'  \wbTTn{n-1}{\Delta'} { \delta }
 Q'$, as $\tilx, \breve c$ and 
$\LL
         \cup \{c \}$ are complete for $ \Delta' $ and $\delta$, where 
$ \Delta' $ is the extension of $ \Delta $ obtained from the transition $\bout b c$, thus
 adding the type for $\comp c$  to $ \Delta $.

    \item The same case as before, that is, 
$\mu = {\bout {b} {c}}$, but  $c$ is an I-name. 
We reason as in the previous  case.

  \item $\mu = {\inp { \comp a   } {{a'}}}$, 
and there is $a$ with $\comp a \sepRES a  \in \tilx $
and $ a \in \LL$;   we assume
    $a',\comp{a'}$ fresh.  
We let
$  B = \{\comp{a'}\}$ .  
We have
   $$ 
\ltsBEG {\emptyset }
    {\res {{\tilx}}}(P|
\RLcp)  {\Arr \tau} \simS
   {\res {{\tilx}}}(P| (
 \overline{c_{n}^{{a}}} \oplus
            {\bout a {a'}} :\RLc {n-1} {\LL\cup B } 
))
    $$
   This has to be matched by
    $$ 
\ltsBEG {\emptyset }
 \res {\tilx}(Q|
\RLcp
) 
     {\Arr \tau}  \contr
   \res {\tilx}(Q_1| (
 \overline{c_{n}^{{a}}} \oplus
            {\bout a {a'}} :\RLc {n-1} {\LL\cup B} 
))
    $$
(as before, with $
\ltsBEG {\delta  }   
Q \Longrightarrow Q_1$).
   We make a further reduction: 
   $$ \begin{array}{rcl}
\ltsBEG {\emptyset }   
   {\res {{\tilx}}}(P| (
 \overline{c_{n}^{{a}}} \oplus
            {\bout a {a'}} :\RLc {n-1} {\LL\cup B} 
))
 & {\Arr \tau} \simS& %

        \res {{\tilx}\,\resBP{a'}}(P'| 
\RLc {n-1} {\LL\cup B} 
)
    \end{array} $$
       \noindent
    This is matched by:
    \[
\ltsBEG {\emptyset }      \res {\tilx}(Q_1| (
 \overline{c_{n}^{{a}}} \oplus
            {\bout a {a'}} :\RLc {n-1} {\LL\cup B} 
))
            {\Arr \tau} Q''
    \]
    We have three possibilities:
    \begin{enumerate}
     \item $
\ltsBEG {\delta}
Q_{1} {\Arr \tau} Q'$ and 
$\comp a \sepRES a \not\in \delta $ (here $a$ is fresh for $P,Q$):
$$
\begin{array}{rcl}
Q'' &\contr&
     \res \tilx 
(Q'|{\bout a {a'}}:\RLc{n-1}{\LL \cup B})
 \\
&=&
     \res \tilx 
(Q'|\resBP {a'}({\out a {a'}}|\RLc{n-1}{\LL \cup B}))
\\
&\wbb &
     \res \tilx 
(Q'|\resBP {a'}(\transf{{\out a {a'}}}|\RLc{n-1}{\LL \cup B}))
\\
&\scong&
     \res \tilx\, \resBP{a'} 
(Q'|\transf{{\out a {a'}}}|\RLc{n-1}{\LL \cup B})
\end{array}
 $$
where the use of $\wbb$ comes from Lemma~\ref{l:mm}(2) and the soundness of $\wbS$ for
barbed congruence (and hence also barbed bisimilarity).
We can therefore conclude 
that $
\ltsBEG {\delta}
Q {\Arr \tau} Q_{1} {\Arr \tau}
     Q'$ 
and
     $P'   \wbTTn {n-1}{\Delta'}{\delta, \compa \sepRES a} Q'|
\transf{{\out a {a'}}} 
$, where $\Delta'$ is the extension of $\Delta$ with a type for $a'$,  using the inductive hypothesis.
     \item $
\ltsBEG {\delta}
Q_{1} {\Arr \tau} Q'$ and 
$\comp a \sepRES a \in \delta $. This is the same as the previous case, except that
$\delta$ need not be extended.

     \item $
\ltsBEG {\delta}
Q_{1} {\Arr {\inp {\compa} {a'} }} Q'$ and
      $Q'' \contr \res \tilx \, \resBP{{a'}}
(Q' | \RLc{n-1}{{\LL} \cup  B})$. 

We can then conclude 
$
\ltsBEG {\delta}
Q {\Arr \tau} Q_{1} {\Arr {\inp {\compa} {a'}}} Q'$,
      and hence
      $P'  
   \wbTTn {n-1}{\Delta'}{\delta}
 Q'$ using 
 the inductive hypothesis, reasoning as in  previous cases.
    \end{enumerate}
\end{enumerate}
\end{proof}

\section{Additional material for Section~\ref{s:alaws} }
\label{a:laws}

The main goal of this section is to prove
Theorem~\ref{t:wire-subst-law}.  
The proofs for the two cases of the theorem 
(when the substituted name  is an O-name  or an I-name)
 are quite different, and are separately
discussed below, as Lemma~\ref{lem:wire-subst-law_O} and
Theorem~\ref{thm:wire-subst-law_I}, respectively.

\subsection{Wire substitution law for O-name}
We first discuss the case for O-names.
The following special case of the substitution law is a valid law of the asynchronous \( \pi \)-calculus with I/O types~\cite[Lemma 10.3.1]{SangiorgiWalker01}.
Hence, we may also transplant this result to \api{}.
\begin{lemma}[Lemma~\ref{l:wire} for O-names]
  \label{l:wire-O}
  \( \out a b \wbcTT \Delta{} \bout a c : \links {\comp c} b  \).
\end{lemma}
Now we are ready to prove the general case.
\begin{lemma}[Theorem~\ref{t:wire-subst-law} for O-names]
  \label{lem:wire-subst-law_O}
$
\resBP b  ( \links {\comp b }a| P ) 
 \wbcTT \Delta{}
 P \sub a{b}
$.
\end{lemma}
\begin{proof}
We can assume, using Lemma~\ref{l:wire-O}, that name $b$ is never exported in $P$.
We can then proceed by structural induction on $P$,  using  algebraic reasoning.
We use 
the laws derived from $\pi$-calculus, such
as those in  Corollary~\ref{c:laws_A}, both in the inductive part  of the proof
(e.g., laws~\ref{i:laws_PC} and~\ref{i:laws_INP} of 
Corollary~\ref{c:laws_A}), and in the base case
(law~\ref{i:laws_COM} of the corollary).
\end{proof}
We  also state a variant of this lemma that uses the expansion
relation, as this will later be used together with the `up-to' expansion technique.
\begin{lemma}
  \label{lem:wire-sust-contr-O}
  Suppose that the O-name \( b \) only appears in  subject position 
in \( P\) (that is, $P$ does not contain outputs in which $b$ is the
value emitted).
  Then we have
  \( \resBP b  ( \links {\comp b }a| P ) \contr P \sub a{b} \).
\end{lemma}

\subsection{Wire substitution law for I-names}
For the proof of the substitution law of I-names, a law that does not
hold in \Api{} or \ALpi{}, we rely on internal bisimulation
(Definition~\ref{d:hb-bisimulation}).
To apply internal bisimulation we first need  to turn
 free outputs into bound outputs.
This, for O-name, is Lemma~\ref{l:wire-O} above. 
We now consider the case of I-names.
The proof directly deals with barbed congruence and make use of
Lemmas~\ref{lem:wire-subst-law_O} and~\ref{lem:wire-sust-contr-O}.

\begin{lemma}[Lemma~\ref{l:wire} for I-names]
  \label{l:wire-I}
  \( \out a b \wbcTT \Delta{} \bout a c : \links  b {\comp c}   \).
\end{lemma}
\begin{proof}
  \newcommand*{\topiIo}[1]{\llbracket #1 \rrbracket^{o}}
  \newcommand*{\relR}{\mathrel{\R}}
  \newcommand*{\proc}{P}
  \newcommand*{\proctwo}{Q}
  Let \( \qct \) be an \api{} closing context.
  Our goal is to prove that \( \ct {\out a b} \wbb \ct {\res c (\out a c \mid \links b {\comp c})}\).
  It suffices to prove that \( {\topiIo \qct} [\out a b] \wbb {\topiIo \qct} [\res c (\out a c \mid \links b {\comp c})] \), where \( \topiIo{-} \) is the half of the encoding \( \transf{-} \) that only turns free outputs sending O-names to bound outputs.
  This is because, by Lemma~\ref{lem:wire-subst-law_O}, \( \ct{\out a b} \wbc {\topiIo \qct} [\out a b] \wbb {\topiIo \qct} [\res c (\out a c \mid \links b {\comp c})] \wbc \ct{\res c (\out a c \mid \links b {\comp c})} \) and \( {\wbc} \subseteq {\wbb} \).
  The reason for applying \( \topiIo{-} \) is to use Lemma~\ref{lem:wire-sust-contr-O} with the up-to expansion technique.

  We show that \( \topiIo \qct [\out a b] \wbb \topiIo \qct [\res c
  (\out a c \mid \links b {\comp c})] \) by proving that the relation
  \( {\relR} \cup \wbb \) is a barbed bisimulation up-to
  expansion, where \( \relR \) is defined as pairs of the form
  \begin{align*}
    (\ct {\out a b \mid \out {\comp b} {x_1} \mid \cdots \mid \out {\comp b} {x_n}}, \ct {\resBP c(\out a c \mid \links b {\comp c} \mid \out {\comp c} {x_1} \mid \cdots \mid \out {\comp c} {x_n}})  )
  \end{align*}
  such that
    \begin{align*}
      &\text{\( \qct \) is closing and  does not have any free outputs exporting O-names.}
  \end{align*}
  Note that since processes relate by \( \relR \) are closed processes they can be considered as processes of the traditional asynchronous \( \pi \)-calculus; that is, we can consider the composite process \( \compP P \emptyset \).
  In the proof, we may implicitly consider \( \compP P \emptyset \) and apply standard proof techniques of \Api{}.

  Suppose that \( \proc_1 \relR \proc_2 \); the case where \(\proc_1 \wbb \proc_2 \) is trivial.
  Processes \( \proc_1 \) and \( \proc_2 \) are of the form \( \ct{\out a b \mid \out {\comp b} {x_1} \mid \cdots \mid \out {\comp b} {x_n}} \) and \( \ct{\resBP c(\out a c \mid \links b {\comp c} \mid \out {\comp c} {x_1} \mid \cdots \mid \out {\comp c} {x_n}}) \), respectively.

  We first consider the case where \( \ct{\out a b \mid \out {\comp b} {x_1} \mid \cdots \mid \out {\comp b} {x_n}} \) makes the challenge.
  It is clear that the barbs can be matched, so we only consider the case when the above process makes a reduction.
  There are only two cases to consider
  \begin{enumerate}
    \item The interaction occurs in \( \qct \). That is, \( \qct \red \qctp \).
    \item The interaction happens between \( \qct \) and the process inside the context using the name \( a \).
  \end{enumerate}
  It is worth emphasizing that interaction between \( \qct \) and the process inside the context using the names \( b, \comp b \) does not happen due to the 1-input property.
  The proof for the first case is trivial since we can simply match the transition with \( \qct \red \qct'\); note that \( \qctp \) is a context without free outputs sending O-names.

  Now we consider the second case.
  In this case, the context \( \qct \) must be of the form (up-to structural congruence)
  \begin{align*}
    \resBP a \resBP b  \resb{\seq z}(\inp a x . \proc \mid \proctwo \mid \contexthole)
  \end{align*}
  (or a context of a similar shape where the input at \( a \) is replicated), and the reduction must be of the form
   \begin{align*}
     &(\proc_1 \defeq ) \resBP a \resBP b  \resb{\seq z}(\inp a x .\proc \mid \proctwo \mid \out a b \mid \out {\comp b} {x_1} \mid \cdots \mid \out {\comp b} {x_n}) \\
     &\red  \resBP a \resBP b  \resb{\seq z}((\proc \sub b x \mid \proctwo \mid \out {\comp b} {x_1} \mid \cdots \mid \out {\comp b} {x_n}) (\defeq \proc_1')
  \end{align*}
  For the matching transition, we consider
  \begin{align*}
  &(\proc_2 \defeq ) \resBP a \resBP b  \resb{\seq z}(\inp a x .\proc \mid \proctwo \mid \resBP c(\out a c \mid \links b {\comp c} \mid \out {\comp c} {x_1} \mid \cdots \mid \out {\comp c} {x_n}] ) \\
  &\red  \resBP a \resBP b  \resb{\seq z}\resBP c(\proc \sub c x \mid \proctwo \mid \links b {\comp c} \mid \out {\comp c} {x_1} \mid \cdots \mid \out {\comp c} {x_n}) (\defeq \proc_2')
  \end{align*}
  Using replication theorems and the substitution law for O-names, we get
  \begin{align*}
    \proc_2'
    &\sbisim \resBP a \resBP c \resb{\seq z}( \resBP b (\proc \sub c x \mid \links b {\comp c}) \mid \resBP b(\proctwo \mid \links b {\comp c}) \mid \resBP b(\out {\comp c} {x_1} \mid \cdots \mid \out {\comp c} {x_n} \mid \links b {\comp c})) \tag{replication theorem} \\
    &\contr  \resBP a \resBP c \resb{\seq z}(\proc \sub c x \sub {\comp c} {\comp b} \mid \proctwo \sub {\comp c} {\comp b} \mid \out {\comp c} {x_1} \mid \cdots \mid \out {\comp c} {x_n}) \tag{Lemma~\ref{lem:wire-sust-contr-O}} \\
    &\equiv  \resBP a \resBP b \resb{\seq z}(\proc \sub b x \mid \proctwo \mid \out {\comp b} {x_1} \mid \cdots \mid \out {\comp b} {x_n}) \tag{\( \alpha\)-conversion } \\
    &= \proc_1'
  \end{align*}
  The premise of Lemma~\ref{lem:wire-sust-contr-O}, i.e.~\( b \) only appears in output subject position, is satisfied because \( \qct \) cannot (1) contain input capability of \( b \) and (2) does not have \( \comp b \) in an object position because it has been eliminated.
  To summarize, we showed that for a reduction \( \proc_1\red \proc_1' \) there exists \( \proc'_2 \) such that \( \proc_2 \red \proc_2' \contr \proc_1' \), and because \( \proc'_1 \wbb \proc'_1 \), we are done with this case.
  (The case where the input at \( a \) is replicated can be proved in a similar manner.)

  We now consider the case where \( \ct{\resBP c(\out a c \mid \links b {\comp c} \mid \out {\comp c} {x_1} \mid \cdots \mid \out {\comp c} {x_n})} \) makes the challenge.
  For this case, there are three cases to consider.
  \begin{enumerate}
    \item The interaction occurs in \( \qct \). That is, \( \qct \red \qctp \).
    \item The interaction happens between \( \qct \) and the process inside the context using the name \( a \).
    \item The interaction happens between \( \qct \) and the process inside the context using the name \( b \).\
  \end{enumerate}
  The proof for the first two cases is essentially identical to what we have done above.
  Therefore we only consider the third case.
  In this case, the context must be of the form
  \begin{align*}
  \resBP b \resb{\seq z}(\proctwo \mid \out {\comp b} {x_{n + 1}} \mid \contexthole)
  \end{align*}
  and the reduction must be of the form
   \begin{align*}
     &\resBP b \resb{\seq z}(\proctwo \mid \out {\comp b} {x_{n + 1}} \mid \resBP c(\out a c \mid \links b {\comp c} \mid \out {\comp c} {x_1} \mid \cdots \mid \out {\comp c} {x_n}) \\
     &\red \resBP b \resb{\seq z}(\proctwo \mid \resBP c(\out a c \mid \links b {\comp c} \mid \out {\comp c} {x_1} \mid \cdots \mid \out {\comp c} {x_n} \mid \out {\comp c} {x_{n + 1}}).
  \end{align*}
  Let \( \qctTwo \defeq \resBP b \resb {\seq z}(\proctwo \mid \contexthole )\).
  Note that \( \qctTwo \) also does not have any free outputs that sends output capability.
  The process after the reduction is \( \ctTwo { \resBP c(\out a c \mid \links b {\comp c} \mid \out {\comp c} {x_1} \mid \cdots \mid \out {\comp c} {x_n} \mid \out {\comp c} {x_{n + 1}})}\).
  We can match this transition with a zero-step reduction because \( \ct{\out a b \mid \out {\comp b} {x_1} \mid \cdots \mid \out {\comp b} {x_n}} \scong \ctTwo{\out a b \mid \out {\comp b} {x_1} \mid \cdots \mid \out {\comp b} {x_n} \mid \out {\comp b} {x_{n+1}}} \).
  The two processes after the (possibly zero-step) reduction are related by \( \relR \), which concludes this case.
\end{proof}

We will also use the following lemma.

\begin{lemma}
\label{l:aapIB}
If $P \wbT{\delta,
    a\sepRES \compa} Q$ and $a$ or $\compa$ are not free in $P$ and
in   $Q$, then also  $P \wbT{\delta} Q$.
\end{lemma}

\begin{theorem}[Theorem~\ref{t:wire-subst-law} for I-names]
  \label{thm:wire-subst-law_I}
$\resBPP b{\compb } ( \links a{\compb }| P ) 
 \wbcTT   \Delta \emptyset
 P \sub {a}b$,  with $\compb $  fresh for $P$. 
\end{theorem}

\begin{proof}
We first 
 translate the statement  to internal processes, via
Lemma~\ref{l:wire_trans}, and then we  apply internal bisimilarity and
\iflong
Corollary~\ref{c:sound}.
 \else
Theorem~\ref{t:sound}.
\fi
Thus the result to be proved becomes: 
\[ 
\resBPP b{\compb } ( \transf{ \links a{\compb }} | P ) 
 \wbTT   \Delta \emptyset
 P \sub {a}b \; , \mbox{ with $\compb $  fresh for $P$ and $P \in$ \apiE.}
\] 
For readability, we omit $ \Delta $ below. 
We proceed by induction on the order of the type of names $a,\compb $ in the wire. 
Precisely, what we assume in the induction is the
general statement of Theorem~\ref{t:wire-subst-law}, tailored to 
internal bisimilarity and internal processes, 
namely 
\begin{equation}
\label{e:ipT}
\resBP c  (\transf{ \ulinks {\comp c }d}| Q ) 
 \wb 
 Q \sub d{c }  \mbox{ with $\comp c$  fresh for $Q$ and $Q \in $ \apiE. }
\end{equation} 
(which, for O-names, has been proved  in Lemma~\ref{lem:wire-subst-law_O} above).

Since the proofs for the base case and the inductive case empoly a similar case analysis, we shall prove them together, explicitly indicating where they differ.
Here, the base case is where $a,\compb $  may only carry unit values and the inductive case is where $a,\compb $   carry I-names or O-names.
For this, 
 we prove a   stronger statement, namely: 
$$\resBPP b{\compb } ( \links a{\compb }| P ) 
 \wbTT \Delta {a\sepRES \compa }
 P \sub {a,\compa }{b,\compb } \mbox{  where $a$  is  fresh for $P$.}$$ 
(the result is stronger because the basic result follows from it in
the case when $\compa, \compb $
do not appear in $P$, and using Lemma~\ref{l:aapIB}).

We consider the relation $\R$  with all elements of the form 
\[
( 
a \sepRES \compa , 
\resBPP b{\compb } ( 
 \transf{ \links a{\compb }} | P ),  
 P \sub {a,\compa }{b,\compb } )
\]
where $a$   fresh for $P$  and $P 
\in \apiE$, 
and show that this is an internal  bisimulation   up to 
$ \wb { }$ (Lemma~\ref{l:uptowb}). 
Let $$
\begin{array}{rcl}
Q_1 &\defi& 
\resBPP b{\compb } ( 
 \transf{ \links a{\compb }} | P ) \\
Q_2 &\defi&  P \sub {a,\compa }{b,\compb } \\
\delta  &\defi& a\sepRES \compa
\end{array} 
$$

The case of an action from $Q_1 $ or $Q_2$ emanating from an action from $P$ 
that does not involve names $\compa ,b,\compb $ is easy, appealing to the  results about invariance for
transitions with respect to substitutions. This case includes visible actions, as well as
silent actions originating from an interaction within $P$ between names $b$ and $\compb $.
The remaining cases are: an interaction in $Q_1$ between an input at $a$ and an output at
$\compa $; an interaction in $Q_2 $ originated from an output at $\compa $ and an input at $b$ in
$P$; an input at $a$ from $Q_1$; an input at $a$ from $Q_2$.
We consider these cases below. 
We  exploit Lemma~\ref{l:mm}(1), which tells us that, if $P$ can make an output at $\compa $,
then  
\begin{equation}
\label{e:boP}
P\equiv \res\tilz  ( \resBP{c} (\out {\compa } c | \ulinks {\comp c }b) | P_1) 
\end{equation} 
with $c$ fresh for $P$.
Or, in the case where \( a \), \( \comp b \) carry the unit value, we have
\[
P\equiv \res\tilz  (\outC {\compa }  | P_1)
\]

\begin{itemize}
\item  
Interaction in $Q_1$ between an input at $a$ and an output at
$\compa $.

We have 
\[ 
Q_1 \equiv 
\resBPP b{\compb }  ( 
 \transf{ \links a{\compb }} | 
\res\tilz   \resBP{c} (
\out {\compa } c | \ulinks {\comp c }b| P_1)) 
\]
and therefore the interaction is 
\[
\ltsBEG {\delta }
 Q_1 \arr \tau 
\equiv 
\resBPP b{\compb } 
( \transf{ \links a{\compb }} |
\res\tilz  ( \resBP{c}  ( 
    \transf{\out {\compb } c}| 
 \ulinks {\comp c }b)| P_1)) \defi Q_1'
\]
Process $Q_2 $ need not move, as 
\[
\begin{array}{rcl}
 Q_2  &\equiv & 
\res\tilz  ( \resBP{c}  ( 
    {\out {\compa } c}| 
 \ulinks {\comp c }a)| P_1 \sub {a,\compa }{b,\compb })
\\
& \wb& 
\res\tilz  ( \resBP{c}  ( 
    \transf{\out {\compa } c}| 
 \ulinks {\comp c }a)| P_1 \sub {a,\compa }{b,\compb })
 = Q_1' 
\sub {a,\compa }{b,\compb } 
\end{array}
     \]
\item 
 Interaction in $Q_2 $ originated from an output at $\compa $ and an input at $b$ in
$P$.

For the base case, we have
\[
P \equiv
\res\tilz  (
\outC {\compa }  \mid \inpC b . P_2 |
P_1) \defi P'
\]
and the interaction is
\[
\ltsBEG {\delta }Q_2  \arr\tau \equiv
\res\tilz  (P_2 |  P_1)
 \sub {a,\compa }{b,\compb } \defi Q_2'
\]
The matching transition for \( Q_1 \) can be given as follows
\begin{align*}
\ltsBEG {\delta }
Q_1 &\equiv \resBPP b{\compb }  (!\inpC a . \outC{\comp b} |  P') \\
&\arr\tau
\resBPP b{\compb }  \res\tilz( !\inpC a . \outC{\comp b} |
 \outC{\comp b} | \inpC b. P_2 | P_1)  \\
&\arr\tau
  \resBPP b{\compb }  \res\tilz( !\inpC a . \outC{\comp b} |  P_2 | P_1).
\end{align*}

Similarly, for the step case, we have
\[ 
P \equiv 
\res\tilz  ( \resBP{c}  ( 
\out {\compa } c | \ulinks {\comp c }b)| \inp by. P_2 | 
P_1) \defi P'
\]
and  the interaction is 
\[ 
\ltsBEG {\delta }Q_2  \arr\tau \equiv 
\res\tilz   \resBP{c}  ( 
   \ulinks {\comp c }b| P_2 \sub c y ) |  P_1)
 \sub {a,\compa }{b,\compb } \defi Q_2'
\]
On $Q_1$'s side we have: 
\[
\begin{array}{rcl}
\ltsBEG {\delta }
Q_1& \equiv & \resBPP b{\compb }   ( 
 \transf{ \links a{\compb }} |  P') \\
& \arr\tau& 
\resBPP b{\compb }  \res\tilz   \resBP{c} (
 \transf{ \links a{\compb }} |
 \transf{ \out{\compb } c} | 
\inp by. P_2 | 
 \ulinks {\comp c }b| P_1)  \\
& \equiv& 
\resBPP b{\compb }  \res\tilz   \resBP{c} (
 \transf{ \links a{\compb }} |
 \resBP{d}( \out{\compb } d |   \ulinks {\comp d }c)
| \inp by. P_2 | 
 \ulinks {\comp c }b| P_1)  \\
& \arr\tau& 
\resBPP b{\compb }  \res\tilz   \resBP{c} (
 \transf{ \links a{\compb }} |
 \resBP{d}(    \ulinks {\comp d }c
| P_2  \sub  d y ) |  
 \ulinks {\comp c }b|  P_1)   \\
& \equiv& 
\resBPP b{\compb }  (
 \transf{ \links a{\compb }} |
 \res\tilz   \resBP{c} (
 \resBP{d}(    \ulinks {\comp d }c |  P_2  \sub  d y ) |  
 \ulinks {\comp c }b|  P_1)   \\
\end{array} 
\]
We can now appeal to the induction hypothesis 
\reff{e:ipT} and derive (as $d\not\in\fn{P_2}$)
\[
 \resBP{d}(    \ulinks {\comp d }c |  P_2  \sub  d y ) 
\wb 
  P_2  \sub  d y  \sub cd  =    P_2  \sub  c y 
 \]
We can then continue above, exploiting the substitutivity properties of $\wb
 $: 
\[
\begin{array}{rcl}
&
\wb
& 
\resBPP b{\compb }  (
 \transf{ \links a{\compb }} |
 \res\tilz   \resBP{c} (
P_2  \sub  c y |
 \ulinks {\comp c }b|  P_1)    \defi Q_1'
\end{array}
\]
and now we are done, as 
$Q'_1 \; \R_\delta \; Q_2'$.
\item 
Input at $a$ from $Q_1$, say $a(c)$.

The derivative is  $
Q'_1 \equiv 
\resBPP b{\compb } ( 
 \transf{ \links a{\compb }} |   (\transf{ \out{\compb } c} |
P)) \defi Q_1''
$.
In this case, $Q_2$ need not move, using clause 3.b.i of
Definition~\ref{d:hb-bisimulation}, as 
  $Q_1''\; \R_\delta \; 
(\transf{ \out{\compb }c} |
P)   \sub {a,\compa }{b,\compb }
$.

\item 
    Input at $a$ from $Q_2$. For this case, we split the base and the inductive case.

The base case is where the input action is simply \( a \).
It must be the case that \( P \equiv \res  \tilz ( \inpC b  . P_1 | P_2) \), and the derivative for $Q_2$ is \( \res  \tilz (  P_1  | P_2) \sub {a,\compa  }{b,\comp b} \).
We have an obvious matching transition for \( Q_1 \):
\begin{align*}
  \ltsBEG {\delta } Q_1
  &\equiv   \res  \tilz \resBPP b{\compb } (! \inpC a . \outC{\comp b} | \inpC b  . P_1 | P_2)\\
  & \arr{\inpC a} \res  \tilz \resBPP b{\compb } (! \inpC a . \outC{\comp b} | \outC {\comp b} | \inpC b  . P_1 | P_2)\\
  &\arr\tau \res  \tilz \resBPP b{\compb } (! \inpC a . \outC{\comp b} |  P_1 | P_2).
\end{align*}

The inductive case is the case where \(a \) carries a name, that is, the case where the action is $a(c)$.
This means that $P \equiv \res  \tilz ( \inp b y . P_1 | P_2)$, and the derivative for
$Q_2$ is 
 $\res  \tilz (  P_1 \sub c y  | P_2) \sub {a,\compa  }{b,\compb }$. 
Process $Q_1$ can answer as follows: 
\[ 
\begin{array}{rcl}
\ltsBEG {\delta }
Q_1 & \equiv & 
 \res  \tilz \resBPP b{\compb } ( 
 \transf{ \links a{\compb }} | 
 \inp b y . P_1 | P_2)
\\
& \arr{\inp a c} & 
 \res  \tilz \resBPP b{\compb } ( 
 \transf{ \links a{\compb }} |   \transf{\out{\compb }c} |
 \inp b y . P_1 | P_2) \\
&\equiv & 
 \res  \tilz \resBPP b{\compb } ( 
 \transf{ \links a{\compb }} | 
 \resBP{d}(   
\out{\compb }  d |
 \ulinks {\comp d }c)
|  \inp b y . P_1 | P_2) \\
&\arr\tau & 
 \res  \tilz \resBPP b{\compb } ( 
 \transf{ \links a{\compb }} | 
 \resBP{d}(   
 \ulinks {\comp d }c
|  P_1 \sub dy) | P_2) \\
&\wb  & 
 \resBPP b{\compb } 
( \res  \tilz
( 
 \transf{ \links a{\compb }} | 
 P_1 \sub cy 
| P_2) )
\end{array} \]
where the equality 
$ \resBP{d}(   
 \ulinks {\comp d }c
|  P_1 \sub dy) 
\wb 
 P_1 \sub cy 
$ is again 
derived from  the inductive 
 hypothesis 
\reff{e:ipT}. 
 \end{itemize} 
\end{proof}

\section{Auxiliary Results for Section~\ref{s:expr}}
\label{a:expr}

\subsection{Encoding of Types}
We first  describe the translation of types.
In simply-typed \ALpi, the grammar of value types is:
\[  T \Coloneqq \Och \ty \mid \unittype   \]
and a typed  restriction is therefore 
$(\res {a: \Bch \ty }) P $, where  $ \Bch \ty$ is the type of a name that could be used
both in input and in output.
When translating such a process, types of the names $\comp a$, $\comp b$,  
and $\comp c$  that are thus introduced in Figure~\ref{fig:fromALpi} of Section~\ref{s:expr} are, respectively,
$\Ich \ty$, $ \Ich{\Och \ty}$, and $ \Ich{(\Ich \ty \times  \Ich{\Och
    \ty})}$ (where $T_1 \times T_2$ indicates a pair of types); and those of the companion names $a,b$ and $c$ are
their dual.
Finally, the name $d$ introduced in the encoding of replication
has type $\Ich \unittype $.

\subsection{Correctness}
In the light of the  correspondence between \api and 
 ordinary asynchronous
$\pi$-calculus (\Api), in Section~\ref{ss:trans},
  we first  show that the  transformation applied in the
encoding $\encoL{\cdot}{}$ of Figure~\ref{fig:fromALpi}
 is valid within \Api, exploiting the
theory of \Api (in fact, we use the theory of the ordinary
$\pi$-calculus). Thus we consider the following encoding $\encoLpi{\cdot}{}$, \emph{from \Api onto
itself}, where function $f$ is a finite function on names as for the encoding
$\encoL{\cdot}{}$:
\[
\begin{array}{rclcrcl}
\multicolumn{7}{c}{ 
\encoLpi{\res a P } f  \defeq
(\res{a,  b,  c })
( 
\out c {a, b}  |
! \inp {c} {a,b}. 
\inp {a}x.\inp {b}y.
(\out c{a,b} |
\out yx )| 
 \encoLpi P {f,a\mapsto b}
)
}\\[10pt]
\encoLpi{ \inp a y . P  } f  & \defeq &  \bout  {f_a} z: \inp z y. \encoLpi P f & &
\encoLpi{!  \inp a y . P  } f  & \defeq &
\res d (\outC {d } | ! d . (\outC { d} | \encoLpi { \inp a y .P} f  )
\\[\tkp]
\encoLpi{ \out  a  b  } f  & \defeq & \out a b  & &
\encoLpi{ P|Q  } f  & \defeq & \encoLpi{ P  } f | \encoLpi{ Q  } f
 \end{array} 
\] 
To reason on the correctness of such an encoding, we employ simulation
equivalence, $\wse$ (Section~\ref{aa:backpi}), rather than
bisimilarity. The reason is that 
 in the encoding a single step (an interaction at
a name $a$, say), is broken into multiple steps (3 steps: a value
emitted at
$a$  is consumed; a process makes a request
for that value; the value is  finally received). These
steps may sometimes yield `partial commitments', as
 the value to be used and the
process that  receives it are identified in separate steps. It is
known in concurrency that partial commitments 
may make bisimilarity  too demanding, and therefore coarser relations are
employed.

We begin by removing the name $d$ introduced to handle the
replication. We can do so using the law
\[
\res d ( \outC d  | !  d . (\outC d | P)
\wse  
! \tau . P  \hskip 2cm   \mbox{ for $d \not \in \fn P$}
 \]
Then we can remove the $\tau$, using the law
\[
\tau . P \wse P 
 \]
We can also  remove the communication of values at the name $c$
introduced in the replication for the encoding of a restriction
server, since these values are always the same and, as we are in \Api,
we are not tight to the 1-input property. 
After such transformations, we derive the following
 simplified  encoding $ \encoLpiD{\cdot}{}$:
\[ 
\begin{array}{rclcrcl}
\multicolumn{7}{c}{ 
\encoLpiD{\res a P } f  \defeq
(\res{a,  b,c })
( \outC c | 
\begin{array}[t]{l}
! c. \inp {a}x.\inp {b}y.( \outC c  | \out yx) 
 |  \encoLpiD P {f,a\mapsto b} )
\end{array}
}\\[10pt]
\encoLpiD{ \inp a y . P  } f  & \defeq &  \bout  {f_a} z: \inp z y. \encoLpiD P f & &
\encoLpiD{!  \inp a y . P  } f  & \defeq &
 ! \encoLpiD { \inp a y .P} f  
\\[\tkp]
\encoLpiD{ \out  a  b  } f  & \defeq & \out a b  & &
\encoLpiD{ P|Q  } f  & \defeq & \encoLpiD{ P  } f | \encoLpiD{ Q  } f
 \end{array} 
\] 
For a $\pi$-calculus process $Q$, we write $\Qab Q$ for the process
resulting from $Q$ by replacing all occurrences of an input at $a$,
say $\inp a y$,  with  $\bout b z:\inp z y$, for $z$ fresh. 
Using  this notation, the final step is the following result (still
within \Api). 
 (As in \api, so in other $\pi$-calculi a
\emph{closed process} is one whose only free names are success names.)

\begin{lemma}[in \Api]
\label{l:aux_pi}
For all \Api closed process $Q$  with $b,c$ fresh for $Q$, we have:
\[
(\res{a} ) Q \wse 
(\res{a,b,c} ) 
( \outC c | !c. a (x). \inp b y . (\outC c | \out y x) |
\Qab   Q
  )
 \]
\end{lemma} 

\begin{proof}
To show that 
$
(\res{a} ) Q $
 can be  simulated by 
$(\res{a,b,c} ) 
( \outC c | !c. a (x). \inp b y . (\outC c | \out y x) |
\Qab  Q )$
it is sufficient to use the relation with all pairs of this form,
modulo structural congruence. 
For the converse 
we take the relation with pairs of the following forms: 
\[ 
\begin{array}{l}
\big( \; 
(\res{a,b,c} ) 
( \outC c | !c. a (x). \inp b y . (\outC c | \out y x) |
\Qab   Q),  (\res{a} ) Q \; \big) \\
\big( \; 
(\res{a,b,c} ) 
(   a (x). \inp b y . (\outC c | \out y x)
|
 !c. a (x). \inp b y . (\outC c | \out y x) |
\Qab  Q ),  (\res{a} ) Q \; \big) \\
\big( \; 
(\res{a,b,c} ) 
(    \inp b y . (\outC c | \out y v)
|
 !c. a (x). \inp b y . (\outC c | \out y x) |
\Qab  Q  ),  (\res{a} )( \out a v |  Q)  \; \big) 
\end{array} 
\]
and show that this a similarity up to expansion (that is, on the
left-hand side, the challenger side, we can make use of the expansion
relation $\contr$).
\end{proof}

Putting together all results above, we have:

\begin{corollary}
\label{c:expr_pi}
For each  closed \ALpi process $P$, it holds that $P \wse \encoLpi P \emptyset $.
\end{corollary}

Finally, putting together Corollary~\ref{c:expr_pi} and
Lemma~\ref{l:api_pi}, we 
derive Theorem~\ref{t:expr_pi} in the main   text; 
we recall its assertion:
\begin{center}
Suppose  $P$ is a closed \ALpi process. Then
$ \encoL P \emptyset  \wse P$. 
\end{center}

\begin{proof}
We have:
\[ \encoL P \emptyset \sbS \encoA{\: \encoL P \emptyset \: } = 
\encoLpi P \emptyset \wse P 
  \]
\end{proof}

\iflong
A further corollary is the full abstraction of the encoding for barbed
bisimulation.
\begin{corollary}
\label{c:bb_fa}
For closed   \ALpi processes $P,Q$, it holds that
$P \wbb Q$ iff 
$ \encoL P \emptyset  \wbb \encoL Q \emptyset$. 
\end{corollary} 

\fi

\clearpage
\section{Supplementary Materials for Section~\ref{sec:esyntax}}
\label{app:esyntax}
This section provides the details of the extended calculus that we have omitted from Section~\ref{sec:esyntax}.

\subsection{Dynamics}
In Section~\ref{sec:esyntax}, we did not explain how the dynamics of the new constructs are defined.
As expected, we extend the reduction relation with the following rules.
\begin{gather*}
{\pmTuple {x_1, \ldots, x_n} {\tuple {\val_1, \ldots, \val_n}} P}  \red {P \sub {\val_1, \ldots, \val_n}{x_1, \ldots, x_n}} \\[1ex]
{\binCase {\inl \val} {x_1} {P_1} {x_2} {P_2}} \red {P_1 \sub \val {x_1}} \\[1ex]
{\binCase {\inr \val} {x_1} {P_1} {x_2} {P_2}}  \red {P_2 \sub \val {x_2}}
\end{gather*}

The LTS can also be extended accordingly.
We just need to add the \( \tau \)-transitions that correspond to the new reductions.
\begin{gather*}
\trans{\rn{Tup}}\infR { }{\ltsTv \delta {\pmTuple {x_1, \ldots, x_n} {\tuple {\val_1, \ldots, \val_n}} P}  {\tau} {P \sub {\val_1, \ldots, \val_n}{x_1, \ldots, x_n}}}\\[1em]
\trans{\rn{SumL}}\infR { }{\ltsTv \delta {\binCase {\inl \val} {x_1} {P_1} {x_2} {P_2}}  {\tau} {P_1 \sub \val {x_1}}} \\[1em]
\trans{\rn{SumR}}\infR { }{\ltsTv \delta {\binCase {\inr \val} {x_1} {P_1} {x_2} {P_2}}  {\tau} {P_2 \sub \val {x_2}}}
\end{gather*}

\subsection{Typing}
\begin{figure}[tp]
  \begin{flushleft}
    \fbox{\( \tyenv \vdash_v \val : \valty \)}
  \end{flushleft}
\begin{mathpar}
   \inferrule[Const]{\tyenv \ \text{discardable} \quad \baseval \text{ is a constant of type } \btype}{\tyenv \vdash \baseval : \btype}\qquad
   \inferrule[Var]{\tyenv \ \text{discardable}}{\tyenv, x : \valty \vdash_v x : \valty} \\
   \inferrule[Tup]{\tyenv_1 \vdash_v \val_1 : \valty_1 \; \cdots \; \tyenv_n \vdash_v \val_n : \valty_n}{\biguplus_{1 \le i \le n} \tyenv_i \vdash \tuple{\val_1, \ldots, \val_n}: \valty_1 \times \cdots \times \valty_n} \quad
   \inferrule[In1]{\tyenv \vdash_v \val : \valty_1}{\tyenv \vdash_v \inl \val : \valty_1 + \valty_2}\qquad
   \inferrule[In2]{\tyenv \vdash_v \val : \valty_2}{\tyenv \vdash_v \inr \val : \valty_1 + \valty_2}
\end{mathpar}
\bigskip
  \begin{flushleft}
    \fbox{\( \tyenv \vdash P \)}
  \end{flushleft}
  \begin{mathpar}
    \inferrule[Weak]{\valty \text{ discardable } \\ \tyenv \vdash P}{\tyenv, a : \valty \vdash P} \qquad\qquad
    \inferrule[Nil]{ }{\emptyenv \vdash \nil}\qquad\qquad
    \inferrule[Res]{\tyenv, a : \ty, b : \dual \ty \vdash P}{\tyenv \vdash \resB* a b \ty P} \\\
    \inferrule[Par]{\tyenv_1 \vdash P_1 \\ \tyenv_2 \vdash P_2}{\tyenv_1 \uplus \tyenv_2 \vdash P_1 \mid P_2} \qquad\qquad\qquad
    \inferrule[Out]{\ty \in \{ \och{\valty}, \loch{\valty} \} \quad \tyenv \vdash_v \val : \valty }{\tyenv, a : \ty \vdash \out a \val} \\
    \inferrule[In]{\tyenv, a : \ich {\valty}, \seqq b : \valty \vdash P}{\tyenv, a : \ich {\valty} \vdash \inp a {\seqq b}. P} \qquad
    \inferrule[InL]{\tyenv, b : \valty \vdash P}{\tyenv, a : \lich {\valty} \vdash \inp a b. P} \qquad
    \inferrule[RIn]{\tyenv\ \text{copyable} \\ \tyenv, \seqq b : \seqq \ty \vdash P}{\tyenv, a : \ich {\seqq \ty} \vdash  !\inp a {\seqq b}. P} \\
    \inferrule[With]{\tyenv_1, x_1 : \valty_1, \ldots ,x_n : \valty_n \vdash P \\ \tyenv_2 \vdash \val : \valty_1 \times \cdots \times \valty_n}{\tyenv_1 \uplus \tyenv_2 \vdash \pmTuple {x_1, \ldots, x_n} \val P  } \\
    \inferrule[Case]{\tyenv, x_i : \valty_i \vdash P_i \quad i = 1, 2\\ \tyenv' \vdash \val : \valty_1 + \valty_2}{\tyenv \uplus \tyenv' \vdash \binCase \val {x_1} {P_1} {x_2} {P_2}}
\end{mathpar}
\caption{Typing rules for the extended syntax}
\label{fig:typing-extended}
\end{figure}

Here we present the full type system of the extended calculus.

We first introduce some auxiliary definitions.
The set of discardable types and copyable types, which are subsets of value types, are defined as follows:
\newcommand*{\dty}{\underline{D}}
\newcommand*{\cty}{\underline{C}}
\begin{align*}
  \text{(Discardable type)} \quad \dty &\Coloneqq \btype \mid \ich{\valty} \mid \och{\valty} \mid \lich{\valty} \mid \loch{\valty} \mid \dty_1 \times \cdots \times \dty_n \mid \dty_1 + \dty_2 \\
  \text{(Copyable type)} \quad \cty &\Coloneqq \btype \mid \och{\valty} \mid   \cty_1 \times \cdots \times \cty_n \mid \cty_1 + \cty_2 \\
\end{align*}
Note that (non-linear) input channel types are discardable, but not copyable.
This reflects the 1-input property of our calculus.
We also note that \( \lich{\valty} \) and  \( \loch{\valty} \) are discardable, meaning that these are actualy \emph{affine} types.
We say that a type environment \( \tyenv \) is discardable (resp.~copyable) if, for every \( x \in \dom \tyenv \), \( \tyenv(x) \) is discardable (resp.~copyable).

We define the combination of types \( \valtyTwo \uplus \valty \) by
\begin{align*}
  \valtyTwo \uplus \valty \defeq
  \begin{cases}
  \valty & \text{if \( \valtyTwo = \valty \) and \( \valty \) is copyable} \\
  \text{undefined} & \text{otherwise}
  \end{cases}.
\end{align*}
The operator \( \uplus \) is extended to act on type environments in a point-wise manner:
\begin{align*}
  (\tyenv_1 \uplus \tyenv_2)(x) \defeq
  \begin{cases}
    \tyenv_1(x) \uplus \tyenv_2(x) & \text{if \(x \in \dom{\tyenv_1} \cap \dom{\tyenv_2}\)} \\
    \tyenv_1(x) & \text{if \( x \in \dom{\tyenv_1}\), but \( x \notin \dom{\tyenv_2}\)} \\
    \tyenv_2(x) & \text{if \( x \in \dom{\tyenv_2}\), but \( x \notin \dom{\tyenv_1}\)} \\
    \text{undefined} & \text{otherwise}
  \end{cases}
\end{align*}
We write \( \tyenv, x : \valty \) to mean \( \tyenv \uplus (x : \valty)\).

The full list of typing rules are given in Figure~\ref{fig:typing-extended}.

\section{Supplementary Materials for Section~\ref{sec:category}}
\label{app:category}
This section gives the proofs missing in Section~\ref{sec:category}.
In other words, we prove that \( \Proc \), \( \ProcN \), \( \ProcS \) have the structure we explained in Section~~\ref{sec:category}.

\subsection{Compact Closed Structure of \texorpdfstring{\( \Proc \)}{Proc}}
Here we prove the claim that \( \Proc \) is a compact closed category.
We have already shown that \( \Proc \) is a category in Theorem~\ref{thm:proc-is-category}.
The next thing we show is that \( \Proc \) is a symmetric monoidal category.

\subparagraph*{Notation.}
We write \( \ulinks {\seq a } {\seq b} \) for parallel compositions of (unoriented) wires \( \ulinks {a_1}{b_1} \mid \cdots \mid \ulinks {a_n}{b_n}  \) for \( \seq a = a_1, \ldots, a_n\) and \( \seq b = b_1, \ldots, b_n \) where each \( a_i \) and \( b_i \) have the dual type of each other; \( \links {\seq a} {\seq b} \) is defined similarly.

\begin{lemma}
  \label{lem:proc-is-SMC}
  \( \Proc \) is a strict symmetric monoidal category.
\end{lemma}
\begin{proof}
  As we already explained, the tensor product \( \bar \env \otimes \bar \envTwo = \bar \env, \bar \envTwo \), the concatenation of the two sequence, and the monoidal unit \( 1 \) is \( \emptyenv \).
  Since string concatenation is associative and empty sequence is the unit of concatenation, the structural morphisms for the monoidal category structure are simply identities.

  What we need to show is that \( \otimes \) is a bifunctor.
  Recall that we defined the tensor product of \( P \colon  \tyenv \to \tyenvTwo \) and  \( Q \colon   \tyenv' \to \tyenvTwo' \) as the parallel composition \( P \mid  Q \colon  \tyenv, \tyenv' \to \tyenvTwo, \tyenvTwo' \), where names are chosen to avoid collision.
  We have the following laws for the strict symmetric monoidal categories.
  The preservation of identities:
  \begin{align*}
    \id_{\bar \tyenv} \otimes \id_{\bar \tyenvTwo}
    &= \ulinks {\seq a}{\seq b} \mid \ulinks {\seq c} {\seq d} \\
    &= \ulinks {\seq a, \seq c}{\seq b, \seq d} \\
    &= \id_{\bar \tyenv, \bar \tyenvTwo}
  \end{align*}
  interchange law:
  \begin{align*}
    (P; Q) \otimes (R; S)
    &= \resBP {\seq a}(P \mid Q) \mid \resBP{\seq b}(R \mid S) \\
    &\equiv \resBP{\seq a}\resBP{\seq b}((P \mid R) | (Q \mid S)) \\
    &= (P \otimes R); (Q \otimes S)
  \end{align*}
  strict associativity of the tensor products:
  \begin{align*}
    (P \otimes Q) \mid R = (P \mid Q) \mid R \equiv P \mid (Q \mid R) = P \otimes (Q \otimes R)
  \end{align*}
  and the unit law law of tensor products:
  \begin{align*}
    \id_{\emptyenv} \otimes P = 0 \mid P  \equiv P \equiv P \mid 0 = P \otimes \id_{\emptyenv}
  \end{align*}
  as desired. Note that they all follow from the structural congruence.

  The symmetry \( \sigma_{\bar \tyenv, \bar \tyenvTwo} : \bar \tyenv, \bar \tyenvTwo \to \bar \tyenvTwo, \bar \tyenv \) is defined by
  \begin{equation*}
    \sigma_{\bar \tyenv, \bar \tyenvTwo} \defeq \ulinks{\seq a}{\seq b'} \mid \ulinks{\seq b}{\seq a'} \colon ( \tyenv, \tyenvTwo)  \to (\tyenvTwo', \tyenv')
  \end{equation*}
  where \( \seq a \) (resp.~\( \seq a' \)) are names in \( \tyenv \) (resp.~\( \tyenv' \)) and \( \seq b \) (resp.~\( \seq b' \)) are names in \( \tyenvTwo \) (resp.~\( \tyenvTwo' \)).
  It is easy to check that the symmetry satisfies its laws.
  The naturality follows from Theorem~\ref{t:wire-subst-law} and the other two laws \( (\sigma_{\bar \tyenv, \bar \tyenvTwo} \otimes \id_{\bar \Xi} ); (\id_{\bar \tyenvTwo} \otimes \sigma_{\bar \tyenv, \bar \Xi}) = \sigma_{\bar \tyenv, \bar \tyenvTwo \otimes \bar \Xi}\) and \( \sigma_{\bar \tyenv, \bar \tyenvTwo}; \sigma_{\bar \tyenvTwo, \tyenv} = \id_{\bar \tyenv \otimes \bar \tyenvTwo}\) can be proved by the transitivity of wires (cf. the proof for Theorem~\ref{thm:proc-is-compact-closed}).
\end{proof}

We now show that \( \Proc \) is a compact closed category.
\ProcisCompactClosed*
\begin{proof}
  Since we have shown that \( \Proc \) is a symmetric monoidal category (Lemma~\ref{lem:proc-is-SMC}), it suffices to show the triangle identities for the compact closed structure.
  Let \( \bar{\env} = \ty_1, \cdots  \ty_n \).
  Then the unit \( \eta_{\bar{\env}} : \emptyenv \to \bar{\env} \otimes \bar{\env}^\perp \) and counit \(\epsilon_{\bar{\env}} : \bar{\env}^\perp \otimes \bar{\env} \to \emptyenv \) are defined by
  \begin{align*}
    \eta_{\bar{\env}} &\defeq \ulinks {\seq a} {\seq b} :  \emptyenv \to (a_1 : \ty_1, \ldots, a_n : \ty_n, b_1 : \dual{\ty_1}, \ldots, b_n : \dual{\ty_n}) \\
    \epsilon_{\bar{\env}} &\defeq  \ulinks {\seq b} {\seq a} : (b_1 : \dual{\ty_1}, \ldots, b_n : \dual{\ty_n}, a_1 : \ty_1, \ldots, a_n : \ty_n) \to \emptyenv
  \end{align*}
  We have \( (\id_{\bar{\env}} \otimes \eta_{\bar{\env}});(\epsilon_{\bar \env} \otimes \id_{\bar \env}) = \id_{\bar \env} \) because
  \begin{align*}
    (\id_{\bar{\env}} \otimes \eta_{\bar{\env}});(\epsilon_{\bar \env} \otimes \id_{\bar \env})
    &= \resBP {\seq b} \resBP {\seq c} \resBP {\seq d}( (\ulinks {\seq a} {\seq b} \mid \ulinks{\seq c}{\seq d}) \mid (\ulinks{\seq {\comp b}}{\seq {\comp c}} \mid \ulinks{\seq {\comp d}}{\seq e})) \\
    &\equiv \resBP {\seq c} (\resBP {\seq b}(\ulinks {\seq a} {\seq b} \mid \ulinks{\seq {\comp b}}{\seq {\comp c}} )  \mid \resBP {\seq d}(\ulinks{\seq c}{\seq d} \mid \ulinks{\seq {\comp d}}{\seq e})) \\
    &\contr \resBP {\seq c}(\ulinks {\seq a} {\seq c}  \mid  \ulinks{\seq {\comp c}}{\seq e}) \\
    &\contr \ulinks {\seq a}{\seq e} \\
    &= \id_{\bar \env}
  \end{align*}
  The other triangle equality is proved in a similar manner.
\end{proof}
\subsection{Cartesian Structure of \texorpdfstring{\( \ProcN \)}{Proc-}}
We have claimed that \( \ProcN \) is a cartesian category by providing the data such as the projection maps.
Here we show that these data indeed satisfies the laws for the carteisan structure.

\begin{lemma}
  The type \( \lich{S + T} \) is indeed the product of \( \lich{S} \) and \( \lich{T} \).
\end{lemma}
\begin{proof}
  We check the \( \beta \) law, \( \eta \) law and the functoriality.

  We first check the \( \beta \) law \( \pi_1; \tuple{P, Q} = P \). (The law for the second projection is proved similarly.)
  We just need to observe that the following equivalence holds:
  \begin{align*}
    &\res{b, c}(\inp a x. \out b {\inl{x}} \mid \inp c x . \binCase x y {\resBP d(P \mid \out {\comp d}{y})} y {\resBP e(Q \mid \out {\comp e}{y})}) \\
    &\sbisim \inp a x. \res{b, c}(\out b {\inl x} \mid \inp c x . \binCase x y {\resBP d(P \mid \out {\comp d}{y})} y {\resBP e(Q \mid \out {\comp e}{y})}) \\
    &\contr \inp a x. \binCase {\inl x} y {\resBP d(P \mid \out {\comp d}{y})} y {\resBP e(Q \mid \out {\comp e}{y})} \\
    &\contr \inp a x. \resBP d(P \mid \out {\comp d}{x}) \\
    &\wbc \inp a x. \resBP d(\inp d x. P_0 \mid \out {\comp d}{x}) \tag{by Lemma~\ref{lem:negative-strict-proc}}\\
    &\contr \inp a x . P_0 \\
    &= P \sub a d.
  \end{align*}

  We next check the \( \eta \) law \( \tuple{\pi_1, \pi_2} = \id \).
  This law holds because
  \begin{align*}
    &\inp a x . \binCase x y {\resBP b(\inp b y . \out d {\inl y} \mid \out {\comp b}{y})} y {\resBP c(\inp c y . \out d {\inr y} \mid \out {\comp c}{y})} \\
    &\contr \inp a x . \binCase x y {\out d {\inl y}} y {\out d {\inr{y}}} \\
    &\equiv \inp a x. \out d x.
  \end{align*}

  We conclude by showing the functoriality \( R; \tuple{P, Q} = \tuple{R; P, R; Q} \).
  For \( P \colon  a : \lich{S} \to b : \lich{T_1} \),  \( Q \colon  a : \lich{S} \to c : \lich{T_2} \),  and  \( R \colon \negenv \to a' : \lich{S} \) we have
  \begin{align*}
    &\resBP a (R \mid \inp d x . \binCase x y {\resBP b(P \mid \out {\comp b}{y})} y {\resBP c(Q \mid \out {\comp c}{y})}) \\
    &\wbc \resBP a (\inp{\comp a}{y}.R_0 \mid \inp d x . \binCase x y {\resBP b(P \mid \out {\comp b}{y})} y {\resBP c(Q \mid \out {\comp c}{y})}) \tag{by Lemma~\ref{lem:negative-strict-proc}} \\
    &\sbisim \inp d x . \left(
      \binCaseAligned x y {\resBP b(  \resBP a(\inp{\comp a}{y}.R_0 \mid P) \mid \out {\comp b}{y})} y {  \resBP c(  \resBP a(\inp{\comp a}{y}.R_0 \mid Q) \mid \out {\comp c}{y})}
      \right)\\
    &= \tuple{R; P, R; Q}
  \end{align*}
\end{proof}

\subsection{On the Relative Closure}
\begin{lemma}
  The mapping \( \negenvTwo \multimap - \) is an endofunctor over \( \ProcS \).
\end{lemma}
\begin{proof}
  The identity is preserved because we have
  \begin{align*}
    &\inp {c'} {\seq a, x}.\resBP b(\resBP c(\inp {c}{z}. \out b {z} \mid \out {\comp c}{x}) \mid \inp{\comp b}{y}. \out{b'}{\seq a, y}) \\
    &\contr \inp {c'} {\seq a, x}.\resBP b(\out b {x} \mid \inp{\comp b}{y}. \out{b'}{\seq a, y})\\
    &\contr \inp {c'} {\seq a, x}.\out{b'}{\seq a, x}.
  \end{align*}

  Now we prove that composition is preserved.
  Suppose that \( P \colon b : \lich{\tyTwo} \to c : \lich{\ty} \) and \( Q \colon d : \lich{\ty} \to e : \lich{\tyThree} \) for some names \( b \), \( c \), \( d \) and \( e\); and types \( \tyTwo \), \( \ty \) and \( \tyThree \).
  Then we have
  \begin{align*}
    & \res{c', d'} (\negenvTwo \multimap P \mid  \negenvTwo \multimap Q) \\
    &=\res{c', d'}\left(
    \begin{aligned}
    &\inp {c'} {\seq a, x}.\resBP b(\resBP c(P \mid \out {\comp c}{x}) \mid \inp{\comp b}{y}. \out{b'}{\seq a, y})
      \mid \\
    &\inp {e'} {\seq a, x}.\resBP d(\resBP e(Q \mid \out {\comp e}{x}) \mid \inp{\comp d}{y}. \out{d'}{\seq a, y})
    \end{aligned}
      \right)\\
    &\sbisim \inp {e'} {\seq a, x}.\resBP d\left(
      \begin{aligned}
        &\resBP e(Q \mid \out {\comp e}{x}) \mid \\
        &\inp{\comp d}{y}. \res {c', d'} (\out{d'}{\seq a, y} \mid \inp {c'} {\seq a, x}.\resBP b(\resBP c(P \mid \out {\comp c}{x}) \mid \inp{\comp b}{y}. \out{b'}{\seq a, y})
      \end{aligned}
      \right)\\
    &\contr \inp {e'} {\seq a, x}.\resBP d\left(
      \begin{aligned}
        &\resBP e (Q \mid \out {\comp e}{x}) \mid \\
        &\inp{\comp d}{y}.\resBP b(\resBP c(P \mid \out {\comp c}{y}) \mid \inp{\comp b}{y}. \out{b'}{\seq a, y})
      \end{aligned}
      \right)\\
      &\wbc \inp {e'} {\seq a, x}.\resBP d \left(
      \begin{aligned}
        &\resBP e (Q \mid \out {\comp e}{x}) \mid \\
        &\inp{\comp d}{y}.\resBP b(\resBP c(\inp c y .P_0 \mid \out {\comp c}{y}) \mid \inp{\comp b}{y}. \out{b'}{\seq a, y})
      \end{aligned}
        \right) \tag{by Lemma~\ref{lem:negative-strict-proc}} \\
       &\contr \inp {e'} {\seq a, x}.\resBP d\left(
        \resBP e(Q \mid \out {\comp e}{x}) \mid  \inp{\comp d}{y}.\resBP b(P_0  \mid \inp{\comp b}{y}. \out{b'}{\seq a, y})
         \right) \\
      &\sbisim \inp {e'} {\seq a, x}.\resBP d\left(
        \resBP e (Q \mid \out {\comp e}{x}) \mid  \resBP b( \inp{\comp d}{y}.P_0  \mid \inp{\comp b}{y}. \out{b'}{\seq a, y})
        \right) \tag*{} \\
      &= \inp {e'} {\seq a, x}.\res{d, c}\left(
        \res {e, e''}(Q \mid \out {e''}{x}) \mid  \res{b, b''}(P \mid \inp{b''}{y}. \out{b'}{\seq a, y})
        \right) \tag{\( \alpha \) conversion on \( \comp d \)} \\
      &\equiv \inp {e'} {\seq a, x}.\resBP b(
        \resBP e ( \res {d, c} (P \mid Q) \mid \out {\comp e}{x})  \mid \inp{\comp b}{y}. \out{b'}{\seq a, y})) \\
    &= \negenvTwo \multimap \res{d, c}(P \mid Q)
  \end{align*}
\end{proof}

\closedStructure*

\begin{proof}
  We already sketched the proof of this lemma, but we shall write the detail here.

  Suppose that \( \negenvTwo \) is \( a_1 \colon \ty_1, \ldots, a_n \colon \ty_n \).
  For a process \( P \colon \negenv \otimes \negenvTwo \to b : \lich{T} \), we define \( \Lambda(P) \) as
  \begin{align*}
    c(\seq a, x).\resBP b (P \mid \out {\comp b} x).
  \end{align*}

  The inverse of this mapping is given by
  \begin{align*}
    \Lambda^{-1}(Q) \defeq \inp b x. \resBP c (\out {\comp c}{\seq a, x} \mid Q)
  \end{align*}
  for \( Q \colon \negenv \to c \colon \negenvTwo \multimap \lich{\ty} \).

  These are indeed the inverse of each other because
  \begin{align*}
    \Lambda(\Lambda^{-1}(Q))
    &= \inp c {\seq a, x}. \resBP b (\inp b x . \resBP c (\out {\comp c}{\seq a, x} \mid Q) \mid \out {\comp b} x) \\
    &\contr \inp c {\seq a, x}. \resBP c(\out {\comp c}{\seq a, x} \mid Q)) \\
    &\sbisim \inp c {\seq a, x}. \resBP c(\out {\comp c}{\seq a, x} \mid \inp c {\seq a, x}. Q_0)) \tag{by Lemma~\ref{lem:negative-strict-proc}}\\
    &\contr \inp c {\seq a, x}. Q_0 \\
    &\sbisim Q
  \end{align*}
  and
  \begin{align*}
    \Lambda^{-1}(\Lambda(P))
    &= \inp b x . \resBP c(\out{\comp c} {\seq a, x} \mid \inp c {\seq a, x} . \resBP b(P \mid \out{\comp b}{x})) \\
    &\contr \inp b x . \resBP b(P \mid \out{\comp b}{x})) \\
    &\contr P \tag{by Lemma~\ref{lem:negative-strict-proc} and communication between \( b \) and \( b' \). }
  \end{align*}

  Naturality is easy to check.
\end{proof}
\subsection{On the Exponential Modality}
\promotion*
\begin{proof}
  We already proved~\ref{it:lem:promotion:left-unit} in the body of the paper.

  The law~\ref{it:lem:promotion:right-unit} holds because
  \begin{align*}
    ! \inp c x. \resBP b(\inp b x .\out a x \mid \out {\comp b} x)
    &\contr ! \inp c x . \out a x.
  \end{align*}

  Now we prove the law~\ref{it:lem:promotion:assoc}.
  Let \( P \colon a : \ich{S} \to b : \lich{T} \) and \( Q \colon c : \ich{T} \to d : \lich{U}\).
  Observe that
  \begin{align*}
    &\prom{(\prom P; Q)} \colon a : \ich{S} \to e : \ich{U} \\
    &= ! \inp e x . \res {d, d'}( \res{c', c} (!\inp {c'} x . \res {b, b'}(P \mid \out{b'} x) \mid  Q) \mid \out {d'} x) \\
    &\contr ! \inp e x . \res {d, d'}( \res{c', c} (!\inp {c'} x . P_0 \mid  Q) \mid \out {d'} x) \\
    &\contr ! \inp e x . \res{c', c} (!\inp {c'} x . P_0 \mid  Q_0)
  \end{align*}

  On the other hand, we have
  \begin{align*}
    \prom P ; \prom Q \wbc
    \resBP c (!\inp {\comp c}{x} . P_0 \mid !\inp e x . Q_0 ) \\
    \wbc ! \inp e x . \resBP c (!\inp {\comp c} x . P_0 \mid  Q_0) \tag{by replication theorem}
  \end{align*}
\end{proof}

\promotionMediating*
\begin{proof}
  The above diagram is equivalent to the following, where we have associated names to types:
  \[
    \begin{tikzcd}
      c : \ich{\valtyThree} \ar[d, "\prom {\tuple{\der, \der}}"] \ar[rr, "\prom {\tuple{P, Q}}"]& & d : \ich{(\valtyTwo + \valty)} \ar[d, "m^{-1}"]\\
      e : \ich{(\valtyThree + \valtyThree)} \ar[rd, "m^{-1}"] & & a : \ich{\valtyTwo}, b: \ich{\valty} \\
      & f: \ich{\valtyThree}, g : \ich{\valtyThree} \ar[ru,"\prom P \otimes \prom Q"]&
    \end{tikzcd}
  \]
  Observe that
  \begin{align*}
    \prom {\tuple{P, Q}} &\wbc !\inp d x. \res{d', d''}( \inp {d'} x. \binCase x y {P_0} {y} {Q_0} \mid \out{d''}{x}) \\
                         &\contr  !\inp d x. \binCase x y {P_0} {y} {Q_0} \\
    \prom P \otimes \prom Q &\wbc !\inp a x.P_0 \sub f c \mid !\inp b x. Q_0 \sub g c \\
   \prom {\tuple{\der, \der}} &\wbc !\inp e x. \res{e', e''}( \inp {e'} x. \binCase x y {\out c y} {y} {\out c y} \mid \out{e''}{x}) \\
                         &\contr  !\inp e x. \binCase x y {\out c y} {y} {\out c y}
  \end{align*}
  where \( P_0 \) and \( Q_0 \) are the processes that satisfy \( a(x).P_0 \sbisim P \) and \( a(x).Q_0 \sbisim Q\), respectively, that exist thanks to Lemma~\ref{lem:negative-strict-proc}. (Here we are assuming that \( P \) and \( Q \) have the free names indicated by \( P \colon c : \ich{\tyThree} \to  a: \lich{S} \) and  \( Q \colon c : \ich{\tyThree} \to  a: \lich{S} \).)

  Thus, we have
  \begin{align*}
    &\prom {\tuple{P, Q}}; m^{-1} \\
    &\wbc \resBP d(!\inp d x. \binCase x y {P_0} {y} {Q_0}  \mid   !\inp a y . \out {\comp d} {\inl y} \mid !\inp b z . \out {\comp d} {\inr z}) \\
    &\sbisim
      \begin{aligned}[t]
        &!\inp a y . \resBP d(!\inp d x. \binCase x y {P_0} {y} {Q_0}  \mid   \out {\comp d} {\inl y} )\\
          &\mid !\inp b y . \resBP d(!\inp d x. \binCase x y {P_0} {y} {Q_0}  \mid   \out {\comp d} {\inr y} )
        \end{aligned}
      \tag{by the replication theorem} \\
    &\contr !\inp a y . {P_0}   \mid !\inp b y . Q_0.
  \end{align*}

  Similarly, we have
  \begin{align*}
       \prom {\tuple{\der, \der}}; m^{-1} \wbc !\inp f x . \out c x \mid !\inp g x . \out c x.
  \end{align*}
  We, therefore, have
  \begin{align*}
    &\prom {\tuple{\der, \der}}; m^{-1} ; \prom P \otimes \prom Q \\
    &\wbc \resBP f \resBP g((!\inp {\comp f} x . \out c x \mid !\inp {\comp g} x . \out c x) \mid (!\inp a x.P_0 \sub f c \mid !\inp b x. Q_0 \sub g c)) \\
    &\equiv \resBP f(!\inp {\comp f} x . \out c x  \mid !\inp a x.P_0 \sub f c) \mid \resBP g( !\inp {\comp g} x . \out c x \mid   !\inp b x. Q_0 \sub g c) \\
    &\wbc ! \inp a x . P_0 \mid ! \inp b x . Q_0 \tag{Lemma~\ref{t:wire-subst-law}}
  \end{align*}
  as desired.
  \end{proof}

\subsection{On the Sequoidal Structure}
\label{app:sequoidal}

As we briefly mentioned, sequoidal CCC is a category in which the sequoidal structure and the closed structure `interact'.
Formally, we require that the sequoid is dual to the closed structure in the following sense.
 \begin{definition}
  A \emph{dual} action of a \( \category C \)-action \( (\category L, \oslash) \) is a \( \category{C}^\op \)-action \( (\category L, \multimap) \) such that, for each \( C \in \category C \), \( C \multimap - \) is a left and right adjoint of \( - \oslash C \).
\end{definition}

\begin{lemma}
  \label{lem:multimap-is-dual-of-oslash}
  We have \( (\ProcSL , \multimap ) \) as the dual of the \( \ProcN \)-action \( (\ProcSL, \oslash) \).
\end{lemma}
\begin{proof}
  The fact that \( (\ProcSL, \multimap) \) is a \( (\ProcN)^{\op} \)-action can be shown by an argument similar to the proof of \( ( \ProcSL, \oslash) \) being a \( \ProcN \)-action.

  The natural bijection
  \begin{align*}
    \varphi \colon \ProcSL(\lich{\dual {\seq \ty} \times \valtyTwo},  \lich{\valty}) \cong \ProcSL(\lich{\valtyTwo}, \lich{\valty \times \seq \ty})
  \end{align*}
  is given as
  \begin{align*}
    \varphi(P) \defeq  \inp {a'} {x, \seq d} . \resBP b (\resBP{a}(P \mid \out {\comp a}{x}) \mid \inp {\comp b} {\seq e, y} . (\out {b'}{y} \mid \links {\seq d} {\seq e}))
  \end{align*}
  where \( P \colon b : \lich{\dual{\seq T} \times \valtyTwo}\to a : \lich{\valty} \).
  Its inverse is given as
  \begin{align*}
    Q \mapsto \inp a x . \resBP {\seq c} \resBP{b'}(\resBP {a'}(Q \mid  \out{\comp a'}{x. \seq c} \mid \inp {\comp b'} {y} . \out{b}{\seq {\comp c}, y}))
  \end{align*}
  where \( Q : b' : \lich{\valtyTwo} \to a' : \lich{\valty \times \seq {T}}\).

  These are indeed the inverse of each other.
  \begin{align*}
    &\varphi^{-1}(\varphi(P))\\
    &=\inp a x . \resBP {\seq d} \resBP{b'}\left(\resBP {a'}\left(
    \begin{aligned}
    &\inp {a'} {x, \seq d} . \resBP b (\resBP{a}(P \mid \out {\comp a}{x}) \mid \inp {\comp b} {\seq e, y} . (\out {b'}{y} \mid \links {\seq d} {\seq e}))  \\
      &\mid\out{\comp a'}{x. \seq d} \mid \inp {\comp b'} {y} . \out{b}{\seq {\comp d}, y}
    \end{aligned}\right)
      \right) \\
    &\contr
    \inp a x . \resBP {\seq d} \resBP{b'}\left(
      \resBP b (\resBP{a}(P \mid \out {\comp a}{x}) \mid \inp {\comp b} {\seq e, y} . (\out {b'}{y} \mid \links {\seq d} {\seq e}))  \mid \inp {\comp b'} {y} . \out{b}{\seq {\comp d}, y}\right)   \tag{communication at \( a' \)}  \\
    &=
      \inp a x . \resBP {\seq d} \resBP{b'}\left(
      \resBP {b''} (\resBP{a}(P \sub{b''}{b} \mid \out {\comp a}{x}) \mid \inp {\comp b''} {\seq e, y} . (\out {b'}{y} \mid \links {\seq d} {\seq e}))  \mid \inp {\comp b'} {y} . \out{b}{\seq {\comp d}, y}\right) \tag{\( \alpha \)-conversion on \( b \)}\\
    &\sbisim
            \inp a x .
      \resBP {b''} \left(\resBP{a}(P \sub{b''}{b} \mid \out {\comp a}{x}) \mid \inp {\comp b''} {\seq e, y} .  \resBP {\seq d} \resBP {b'}(\out {b'}{y} \mid \links {\seq d} {\seq e} \mid \inp {\comp b'} {y} . \out{b}{\seq {\comp d}, y}) \right)  \\
    &\contr
      \inp a x .
      \resBP {b''} \left(\resBP{a}(P \sub{b''}{b} \mid \out {\comp a}{x}) \mid \inp {\comp b''} {\seq e, y} .  \resBP {\seq d}(\links {\seq d} {\seq e} \mid  \out{b}{\seq {\comp d}, y}) \right)  \tag{communication at \( b' \) }\\
    &\wbc{}{}
      \inp a x .
      \resBP {b''} \left(\resBP{a}(P \sub{b''}{b} \mid \out {\comp a}{x}) \mid \inp {\comp b''} {\seq e, y} . \out{b}{\seq e, y}\right) \tag{by Lemma~\ref{l:wire}} \\
    &=
      \inp a x .
      \resBP {b''} \left(\resBP{a}(P \sub{b''}{b} \mid \out {\comp a}{x}) \mid \links {b''} b \right) \tag{def.\ of \( \links{}{} \)} \\
    &\wbc
      \inp a x . \resBP{a}(P \mid \out {\comp a}{x}) \tag{by Theorem~\ref{t:wire-subst-law}}\\
    &\wbc P.
  \end{align*}
  Here the last step follows from the fact that \( P \wbc \inp a x . P_0 \) for some \( P_0 \) (Lemma~\ref{lem:negative-strict-proc}).
  We also have
  \begin{align*}
    &\varphi(\varphi^{-1}(Q)) \\
    &= \inp {a} {x, \seq d} . \resBP {b'} \left(
      \begin{aligned}
        &\resBP{a'}(\inp {a'} x . \resBP {\seq c} \resBP{b}(\resBP {a}(Q \mid  \out{\comp a}{x, \seq c} \mid \inp {\comp b} {y} . \out{b'}{\seq {\comp c}, y})) \mid \out {\comp a'}{x}) \\
        &\mid \inp {\comp b'} {\seq e, y} . (\out {b}{y} \mid \links {\seq d} {\seq e})
      \end{aligned} \right) \\
    &\contr
      \inp {a} {x, \seq d} . \resBP {b'} \left(
        \resBP {\seq c} \resBP{b}(\resBP {a}(Q \mid  \out{\comp a}{x, \seq c} \mid \inp {\comp b} {y} . \out{b'}{\seq {\comp c}, y})) \mid \inp {\comp b'} {\seq e, y} . (\out {b}{y} \mid \links {\seq d} {\seq e}) \right) \tag{communication at \( a' \)}\\
    &=    \inp {a} {x, \seq d} . \resBP {b'} \left(
      \resBP {\seq c} \resBP{b''}(\resBP {a}(Q \sub{b''}{b} \mid  \out{\comp a}{x, \seq c} \mid \inp {\comp b''} {y} . \out{b'}{\seq {\comp c}, y})) \mid \inp {\comp b'} {\seq e, y} . (\out {b}{y} \mid \links {\seq d} {\seq e}) \right) \tag{\( \alpha \)-conversion at \(b \) }\\
    &\sbisim
      \inp {a} {x, \seq d} .
      \resBP {\seq c} \resBP{b''}(\resBP {a}(Q \sub{b''}{b} \mid  \out{\comp a}{x, \seq c} \mid \inp {\comp b''} {y} . \resBP{b'}(\out{b'}{\seq {\comp c}, y}) \mid \inp {\comp b'} {\seq e, y} . (\out {b}{y} \mid \links {\seq d} {\seq e})) \\
    &\contr
      \inp {a} {x, \seq d} .
      \resBP {\seq c} \resBP{b''}(\resBP {a}(Q \sub{b''}{b} \mid  \out{\comp a}{x, \seq c} \mid \inp {\comp b''} {y} .  (\out {b}{y} \mid \links {\seq d} {\seq {\comp c}})) \\
       &\contr
      \inp {a} {x, \seq d} .
         \resBP {\seq c} (\resBP {a}(Q \mid  \out{\comp a}{x, \seq c}) \mid \links {\seq d} {\seq {\comp c}}) \tag{since \(Q \) is linear communication at \( b'' \) is guaranteed}  \\
    &\wbc{}{}
       \inp {a} {x, \seq d} .
         \resBP {a}(Q \mid  \out{\comp a}{x, \seq d})
      \tag{by Lemma~\ref{l:wire}} \\
    &\wbc{}{} Q \tag{by Lemma~\ref{lem:negative-strict-proc} as in the final step of the above equations}
  \end{align*}

  Similarly we can show that \( \seq{\ty} \multimap - \) is the right-adjoint of \( - \oslash \seq {\ty}\).
\end{proof}

\begin{lemma}
\label{lem:exp-is-dual-of-oslash-bang}
  The \( \ProcN_{!} \)-action \( (\ProcSL, \oslash_{!}) \) has the dual \((\ProcSL, \Rightarrow) \).
\end{lemma}
\begin{proof}
  We show that \( \lich{\valty} \Rightarrow - \) is the left adjoint of \( - \oslash_{!} \lich{\valty}\); the fact that it is also the right adjoint can be proved similarly.
  We have the following natural bijection
  \begin{align*}
    \ProcSL(\lich{\valty} \Rightarrow \lich{\valtyTwo}, \lich{\valtyThree})
    &\cong \ProcSL(!\lich{\valty} \multimap \lich{\valtyTwo}, \lich{\valtyThree}) \tag{by def.~of \( \Rightarrow \)} \\
    &\cong \ProcSL(\lich{\valtyTwo}, \lich{\valtyThree} \oslash !\lich{\valty}) \tag{by Lemma~\ref{lem:multimap-is-dual-of-oslash}} \\
    &\cong \ProcSL(\lich{\valtyTwo}, \lich{\valtyThree} \oslash_{!} \lich{\valty}) \tag{by def.~of \( \oslash_{!}\)}.
  \end{align*}
\end{proof}

Now we are ready to define the notion of sequoidal CCC.
\begin{definition}[\cite{Laird19}]
  A \emph{sequoidal CCC} is given by a tuple \( (\category C, \category L, \oslash, J) \) where
  \begin{itemize}
    \item \( \category C \) is a cartesian closed category
    \item \( (\category L, \oslash ) \) is a \( \category C \)-action with a dual \( (\category L, \multimap) \)
    \item \( J \) is a strong morphism of \( \category{C}^\op \)-actions \( (\category C, \Rightarrow) \to (\category L, \multimap)\).
      That is, there exists an isomorphism \( \omega\) that makes \( (J, \omega) \) is a weak morphism from \( (\category C, \Rightarrow) \) to \((\category L, \multimap)\).
  \end{itemize}
\end{definition}

Finally, we show the Kleisli category \( \ProcN_{!} \) yields a sequoidal CCC.
\sequoidalCCC*
\begin{proof}
  Since we know that \( \ProcN_{!} \) is a CCC and \( (\ProcSL, \oslash_{!}) \) has the dual \( (\ProcSL, \Rightarrow )\) by Lemma~\ref{lem:exp-is-dual-of-oslash-bang}, it suffices to show that \( J_{!} \) is a strong morphism of \( (\ProcN_{!})^\op  \)-actions from \( (\ProcN_{!}, \Rightarrow) \) to \( (\ProcSL, \Rightarrow) \).
  This is trivial because \(J_!(\lich{\tyTwo}) \Rightarrow \lich{\ty} \) and \(J_{!}(\lich{\tyTwo} \Rightarrow \lich{\tyTwo}) \) are the same objects, and thus, we can take the identity as the isomorphism.
\end{proof}

\end{document}